\documentclass[11pt]{article}
\usepackage{amsmath,amssymb,amsfonts,amsthm,epsfig}
\usepackage[usenames,dvipsnames]{xcolor}
\usepackage{tikz}
\usepackage{bm,xspace}
\usepackage{cancel}
\usepackage{fullpage}
\usepackage{liyang}
\usepackage{framed}
\usepackage{verbatim}
\usepackage{enumitem}
\usepackage{array}
\usepackage{multirow}
\usepackage{afterpage}
\usepackage{mathrsfs}

\newcommand{\erfc}{\mathsf{erfc}}

\newcommand{\mind}{\mathsf{Indsample}}
\newcommand{\ber}{\mathsf{Bern}}

\newcommand{\pois}{\mathsf{Poisample}}

\newcommand{\BR}{\mathbf{R}} 

 \usetikzlibrary{calc}

\newcommand{\inter}{\mathrm{inter}}
\usepackage[normalem]{ulem}
\newcommand{\umemp}{U_{M, \mathrm{emp}}}
\newcommand{\id}{\mathsf{id}}
\newcommand{\itt}{\mathsf{in}}

\newcommand{\Suppm}{S_{\actual}}
\newcommand{\Pro}{\mathop{\Pr}}
\newcommand{\dwa}{\mathrm{d}_{\mathrm{W},1}}
\newcommand{\rel}{\mathrm{Rel}}
\newcommand{\sac}{S_{\actual}}

\usepackage{caption} 

\usepackage{todonotes}

\makeatletter
\newtheorem*{rep@theorem}{\rep@title}
\newcommand{\newreptheorem}[2]{
\newenvironment{rep#1}[1]{
 \def\rep@title{#2 \ref{##1}}
 \begin{rep@theorem}\itshape}
 {\end{rep@theorem}}}
\makeatother
\theoremstyle{plain}

\def\centerarc[#1](#2)(#3:#4:#5)% Syntax: [draw options] (center) (initial angle:final angle:radius)
    { \draw[#1] ($(#2)+({#5*cos(#3)},{#5*sin(#3)})$) arc (#3:#4:#5); }
    \def\drawdash[#1](#2)(#3:#4:#5)
     { \draw[#1] ($(#2)+({#5*cos(#3)},{#5*sin(#3)})$) -- ($(#2)+({#5*cos(#4)},{#5*sin(#4)})$); }
     
     \newcommand{\Ay}{20 + 5*sin(30)}
          \newcommand{\Ayy}{20 + 5*cos(30)}
          \newcommand{\Ayz}{20 - 5*sin(30)}
          \newcommand{\Ayw}{20 - 5*cos(30)}

     \newcommand{\Cy}{20 - 5*sin(30)}

\newcommand{\newangle}{(asin(5/3 * sin(30)))}

\newcommand{\out}{\mathsf{outer}}
\newcommand{\ovp}{\overline{p}}
\makeatletter
\newenvironment{proofof}[1]{\par
  \pushQED{\qed}%
  \normalfont \topsep6\p@\@plus6\p@\relax
  \trivlist
  \item[\hskip\labelsep
\emph{    Proof of #1\@addpunct{.}}]\ignorespaces
}{%
  \popQED\endtrivlist\@endpefalse
}
\makeatother

\newcommand{\ignore}[1]{}

\def\colorful{0}

\ifnum\colorful=1

\newcommand{\blue}[1]{{{\color{blue}#1}}}
\newcommand{\red}[1]{{\color{red} {#1}}}

\newcommand{\gray}[1]{{\color{gray}{#1}}}
\newcommand{\lnote}[1]{\footnote{{\bf \color{blue}Li-Yang}: {#1}}}
\newcommand{\rnote}[1]{\footnote{{\bf \color{red}Rocco}: {#1}}}

\fi
\ifnum\colorful=0

\newcommand{\blue}[1]{{{#1}}}
\newcommand{\red}[1]{{{#1}}}

\newcommand{\gray}[1]{{{#1}}}
\newcommand{\lnote}[1]{}
\newcommand{\rnote}[1]{}

\fi

\usepackage{boxedminipage}

\newreptheorem{theorem}{Theorem}
\newtheorem*{theorem*}{Theorem}
\newreptheorem{lemma}{Lemma}
\newreptheorem{proposition}{Proposition}
\newtheorem*{noclaim*}{Claim}

\newcommand{\Weight}{\mathrm{W}}

\newcommand{\slab}{\mathrm{slab}}

\newcommand{\spann}{\mathrm{span}}
\newcommand{\rsm}{r_{\mathsf{small}}}

\newcommand{\Nn}{N(0,1)^n}

\newcommand{\median}{\mathrm{median}} 

\newcommand{\region}{\mathrm{Region}}
\newcommand{\ideal}{\mathrm{ideal}} 
\newcommand{\actual}{\mathrm{actual}}

\newcommand{\halfspace}{s}

\newcommand{\Ind}{\mathds{1}}

\usepackage{dsfont}

\renewcommand{\N}{\mathds{N}} 
\renewcommand{\R}{\mathds{R}}

\newcommand{\uniform}{\mathrm{U}_{\mathbb{S}^{n-1}}}

\newcommand{\vol}{\mathrm{vol}}

\begin{document}

\title{Kruskal-Katona for convex sets, with applications}

\author{Anindya De\thanks{{\tt anindyad@cis.upenn.edu}. Supported by NSF grant CCF-1926872}\\
University of Pennsylvania\\
\and Rocco A.\ Servedio\thanks{{\tt rocco@cs.columbia.edu}. Supported by NSF grants CCF-1814873, IIS-1838154, CCF-1563155, and by the Simons Collaboration on Algorithms and Geometry.}\\
Columbia University\\
}

\maketitle

\setcounter{page}{0}

\begin{abstract}
The well-known Kruskal-Katona theorem in combinatorics says that (under mild conditions) every monotone Boolean function $f: \zo^n \to \zo$ has a nontrivial ``density increment.'' This means that the fraction of inputs of Hamming weight $k+1$ for which $f=1$ is significantly larger than the fraction of inputs  of Hamming weight $k$ for which $f=1.$

We prove an analogous statement for convex sets. Informally, our main result says that (under mild conditions) every convex set $K \subset \R^n$  has a nontrivial density increment. This means that the fraction of the radius-$r$ sphere that lies within $K$ is significantly larger than the  fraction of the radius-$r'$ sphere that lies within $K$, for $r'$ suitably larger than $r$.  For centrally symmetric convex sets we show that our density increment result is essentially optimal.

As a consequence of our Kruskal-Katona type theorem, we obtain the first efficient weak learning algorithm for convex sets under the Gaussian distribution.  We show that any convex set can be weak learned to advantage $\Omega(1/n)$ in $\poly(n)$ time under any Gaussian distribution and that any centrally symmetric convex set can be weak learned to advantage $\Omega(1/\sqrt{n})$ in $\poly(n)$ time.  We also give an information-theoretic lower bound showing that the latter advantage is essentially optimal for $\poly(n)$ time weak learning algorithms.  As another consequence of our Kruskal-Katona theorem, we give the first nontrivial Gaussian noise stability bounds for convex sets at high noise rates. Our results extend the known correspondence between monotone Boolean functions over $\zo^n$ and convex bodies in Gaussian space.

\end{abstract}

\thispagestyle{empty}

\newpage

%!TEX root = draft1b.tex

\section{Introduction} \label{sec:intro}

Several results in Boolean function analysis and computational learning theory suggest an analogy between convex sets in Gaussian space and monotone Boolean functions\footnote{Recall that a function $f: \bn \to \bits$ is monotone if $f(x) \leq f(y)$ whenever $x_i \leq y_i$ for all $i \in [n].$} with respect to the uniform distribution over the hypercube.  As an example, Bshouty and Tamon~\cite{BshoutyTamon:96} gave an algorithm that learns monotone Boolean functions over the $n$-dimensional hypercube to any constant accuracy in a running time of $n^{O(\sqrt{n}
)}$. Much later, Klivans, O'Donnell and Servedio~\cite{KOS:08} gave an algorithm that learns convex sets over $n$-dimensional Gaussian space with the same running time. While the underlying technical tools in the proofs of correctness are different, the algorithms in \cite{KOS:08} and \cite{BshoutyTamon:96} are essentially the same:  \cite{BshoutyTamon:96} (respectively \cite{KOS:08}) show that the Fourier spectrum (respectively Hermite spectrum\footnote{The Hermite polynomials form an orthonormal basis for the space of square-integrable real-valued functions over Gaussian space; the Hermite spectrum of a function over Gaussian space is analogous to the familiar Fourier spectrum of a function over the Boolean hypercube.  See \Cref{sec:hermite} for more details.}) of monotone functions (respectively convex sets) is concentrated in the first $O(\sqrt{n})$ levels. Other structural  analogies between convex sets and monotone functions are known as well; for example, an old result of Harris~\cite{harris1960lower} and Kleitman~\cite{kleitman1966families} shows that monotone Boolean functions over $\bits^n$ are positively correlated. 
The famous Gaussian correlation conjecture (now a theorem due to Royen~\cite{royen2014simple})
asserts the same for centrally symmetric convex sets under the Gaussian distribution. We note that while the assertions are analogous, the proof techniques are very different, and indeed the Gaussian correlation conjecture was open for more than  half  a century.  

Despite these analogies between convex sets and monotone functions, there are a number of prominent gaps in our structural and algorithmic understanding of convex sets when compared against monotone functions. We list several examples:
 
\begin{enumerate}
\item Nearly matching $\poly(n)$ upper and lower bounds are known for the query complexity of testing monotone functions over the $n$-dimensional Boolean hypercube~\cite{FLNRRS, KMS15, chakrabarty2016n, belovs2016polynomial, CDST15, CWX17}. However, the problem of convexity testing over the Gaussian space is essentially wide open, with the best known upper bound (in~\cite{chen2017sample}) being $n^{O(\sqrt{n})}$ queries and no nontrivial lower bounds being known.

\item Kearns, Li and Valiant~\cite{KLV:94} showed that the class of all monotone Boolean functions over $\bn$ is \emph{weakly learnable} under the uniform distribution in polynomial time, meaning that the output hypothesis $h$ satisfies $\Pr_{\bx \in \bn} [h(\bx) = f(\bx) ] \ge 1/2 + 1/\poly(n)$, where $f: \bits^n \to \bits$ is the target monotone function. \cite{KLV:94} achieved an advantage of $\Omega(1/n)$ over $1/2$; this advantage was improved by Blum, Burch and Langford~\cite{BBL:98} to $\Omega(n^{-1/2})$ and subsequently by O'Donnell and Wimmer~\cite{OWimmer:09}
to $\Omega(n^{-1/2}\log n )$ which is optimal up to constant factors for $\poly(n)$-time learning algorithms. On the other hand, prior to the current work, nothing non-trivial was known about weak learning convex sets under the Gaussian measure. 

\item Closely related to Item~2 is the folklore fact (see \Cref{app:hermite-weight}) that for every monotone function $f: \bn \rightarrow \bits$, the Fourier weight (sum of squared Fourier coefficients) at levels $0$ and $1$  is at least $\Omega({\frac {\log^2 n} n})$ .\ignore{\gray{In other words, for every monotone function $f$, there is a linear function $\ell(\cdot)$ of unit norm such that $\mathbf{E}_{\bx}[\ell(\bx) \cdot f(\bx)] = \Omega(n^{-1/2})$.}\rnote{Is it saying for every monotone $f: \bits \to \bits$ there is a linear function $\ell(x)=w\cdot x + w_0$ with $\sum_i w_i^2 = 1$ such that $\mathbf{E}_{\bx}[\ell(\bx) \cdot f(\bx)] = \Omega(n^{-1/2})$? Is this exactly right - what about the log factor? Could we just skip this sentence since the rest of the paragraph is all about Fourier/Hermite weight?}}  On the other hand, prior to this work, it was  consistent with the state of our knowledge that there is a convex set whose indicator function $f: \mathbb{R}^n \rightarrow \bits$ has zero Hermite weight (sum of squared Hermite coefficients) at levels $0,1,\dots,o(\sqrt{n})$.\ignore{that the the first $o(\sqrt{n} )$ levels of the Hermite spectrum of $f$ are all empty.}
%Now, this is clearly false for convex sets -- e.g., consider $f: \mathbb{R}^n \rightarrow \{\pm 1\}$ which is the indicator function of a sphere such that $\Pr_{\bg \sim \mathcal{N}^n(0,1)} [f(\bg) =1] = 1/2$.  Then, by construction, the level-$0$ Hermite weight of $f$ is $0$ and the level-$1$ Hermite weight is also $0$ (from the spherical symmetry of $f$). 
\end{enumerate}
\paragraph*{Main contributions of this work.}
The main technical contribution of this work is extending a fundamental result on monotone Boolean functions, called the Kruskal-Katona theorem~\cite{Kruskal:63, katona1968theorem}, to convex sets over Gaussian space. We use this result to address items~2 and 3 above. More precisely, we give a weak learning algorithm which achieves an accuracy of $1/2 + \Omega(n^{-1})$  for arbitrary convex sets, and we show that the Hermite weight at levels $0$, $1$ and $2$ of any convex set must be at least $\Omega(n^{-2})$. For centrally symmetric convex sets, we give a weak learner with accuracy $1/2 + \Omega(n^{-1/2})$ and show that the Hermite weight at levels 0 and 2 must be at least $\Omega(1/n).$

For centrally symmetric convex sets, we show that both our weak learning result and our bound on the Hermite weight at low levels are optimal up to $\polylog(n)$ factors; it follows that the corresponding results for general convex sets are also optimal up to a quadratic factor.  

We now explain our analogue of the Kruskal-Katona theorem in more detail.

%\red{we both give a weak learning algorithm for convex sets as well as show the squared mass at levels $0$, $1$ and $2$ of the Hermite spectrum of any convex set must be $\Omega(1/n)$.}

\ignore{
Bshouty/Tamon: learning monotone functions via Fourier concentration.  KOS: analogous result for learning convex sets under Gaussian distribution, via very similar techniques.  (Explain these results.)  But, the full extent of this analogy is not understood.  What about low degree Fourier weight?  While convex sets have been intensively studied as geometric objects, we don't have such a good algorithmic understanding. For example, much is known about weak learning monotone functions (forward pointer to our discussion of this), but prior to this work nothing was known about weak learning convex sets.  As another example, we don't know much about testing convexity, while testing monotonicity is very well understood.}

\ignore{Main contribution of this work:  we extend a fundamental result about monotone Boolean functions, the Kruskal-Katona theorem, to the setting of convex sets.  (Reiterate density increment stuff from abstract, at same level of detail.)  As an application, we leverage this understanding to get weak learning results and results about Hermite concentration of convex sets at low levels.  For centrally symmetric convex sets, we show that our density increment, weak learning, and Hermite concentration results are essentially optimal. \rnote{This can be 2 or 3 paragraphs I think.}}

\subsection{Background:  the Kruskal-Katona theorem}
We begin by recalling the Kruskal-Katona theorem over the Boolean hypercube. Informally, the Kruskal-Katona theorem is a \emph{density increment}  result --- it asserts that the density of the $1$-set of a monotone function must increases non-trivially over the successive slices of the hypercube. More precisely, let $f: \bn \rightarrow \bits$ be a monotone function and for any $0 \le k \le n$, let $\binom{[n]}{k}$ denote the $k^{th}$ slice of the hypercube (the $n \choose k$-size set of strings that have exactly $k$ ones). Define $\mu_k(f)$ as 
$
\mu_k(f) := \Pr_{\bx \in \binom{[n]}{k}}[f(\bx)=1],
$
i.e.~the density of $f$ restricted to the $k$-th slice of the hypercube.
The Kruskal-Katona theorem states that for every monotone $f$ and every $k \in [0,n-1]$, the density $\mu_k(f)$ satisfies \begin{equation}~\label{eq:KK-cube}
\mu_{k+1}(f)  \ge \mu_k(f)^{1-\frac{1}{n-k}} \ge \mu_k(f) + \frac{\mu_k(f) \ln (1/\mu_k(f))}{n-k}.
\end{equation}
As an example, it is instructive to consider the following specific parameter settings: Suppose $k \in [n/2 - \sqrt{n}, n/2 + \sqrt{n}]$ and that $1/3 \le \mu_k(f)  \le 2/3$ in this range of $k$. Then the theorem says that 
$
\mu_{k+1}(f) \ge \mu_k(f)  + \Theta(1/n). 
$
Consequently, for $k_{\mathsf{up}} = n/2 + \sqrt{n}$ and $k_{\mathsf{down}} = n/2 -\sqrt{n}$, 
$
\mu_{k_{\mathsf{up}}}(f) \ge \mu_{k_{\mathsf{down}}}(f) + \Theta(n^{-1/2}). 
$

We mention here that the original result of Kruskal~\cite{Kruskal:63} and Katona~\cite{katona1968theorem}
is stated in terms of the sizes of the  upper and lower shadows of sets $A \subseteq \binom{[n]}{k}$, and that examples which are exactly extremal for the precise bounds given in those papers can be obtained by considering the ordering of the elements of $\binom{[n]}{k}$ in the so-called \emph{colexicographic order}. While their result is tight, it is often not as convenient to work with as the above formulation. The above version is due to  Bollob{\'{a}}s and Thomason~\cite{BT87} and is the form most often used in computer science applications (for example, the weak learning results for monotone functions given in \cite{BBL:98} and \cite{OWimmer:09} used this version). For completeness, we mention that prior to \cite{BT87}, Lov{\'a}sz~\cite{lovasz2007combinatorial} also gave a simplified version of the Kruskal-Katona theorem, but in this paper, we will refer to \Cref{eq:KK-cube} as the Kruskal-Katona theorem. 

\subsection{Our main structural result:  a Kruskal-Katona type theorem for convex sets}
We now describe our main structural result for convex sets, which is closely analogous to the Kruskal-Katona theorem. In order to do this, we first need to identify an analogue of hypercube slices in the setting of Gaussian space. The most obvious choice is to consider spherical shells; namely, for $r>0$, define the radius-$r$ spherical shell to be $\mathbb{S}^{n-1}_r := \{x \in \mathbb{R}^n: \Vert x \Vert_2 =r\}$. Note that, analogous to slices of the hypercube, spherical shells are the level sets of the Gaussian distribution. 

Given a convex set $K \subseteq \mathbb{R}^n$, we define 
the \emph{shell-density function} $\alpha_
K: (0, \infty) \rightarrow [0,1]$ to be 
\begin{equation}~\label{eq:shell-density-def}
\alpha_K(r) := \Prx_{\bx \sim \mathbb{S}^{n-1}_r} [\bx \in K]. 
\end{equation}
Having defined $\alpha_K(\cdot)$, the most obvious way to generalize Kruskal-Katona to Gaussian space would be to conjecture that for $K$ a convex set $\alpha_K(\cdot)$ is a non-increasing function, and further, that as long as $\alpha_K(r)$ is bounded away from $0$ and $1$, it exhibits a non-trivial rate of decay as $r$ increases (similar to \eqref{eq:KK-cube}). However, a moment's thought shows that this conjecture cannot be true because of the following examples: 
\begin{enumerate}
\item Let $K \subseteq \mathbb{R}^n$ be a convex body with positive Gaussian volume whose closest point to the origin is at some distance $t>0$.  Then the shell density function $\alpha_K(r)$ is zero for $0 < r \leq t$ but subsequently becomes positive.  Thus for $\alpha_K(\cdot)$ to be non-increasing, we require $0^n \in K$. 

In fact, it is easy to see that if $0^n \in K$ and $K$ is convex then $\alpha_K(\cdot)$ is in fact non-increasing (since by convexity the intersection of $K$ with any ray extending from the origin is a line segment starting at the origin). However, this does not mean that there is an actual decay in the value of $\alpha_K$, as witnessed by the next example: 

\item Let $K$ be an origin-centered halfspace, i.e.~$K = \{x : w \cdot x \ge 0\}$ for some nonzero $w \in \R^n$.  $K$ is convex and $0^n \in K$, but $\alpha_K(r) =1/2$ for all $r >0$, and hence $\alpha_K(r)$ exhibits no decay as $r$ increases. 
\end{enumerate}

The second example above shows that in order for $\alpha_K(\cdot)$ to have decay, it is not enough for the origin to belong to $K$; rather, what is needed is to have $B(0^n,s) \subseteq K$ for  some $s>0$, where $B(0^n,s)$ is the ball of radius $s$ centered at the origin. Our Kruskal-Katona analogue, stated below in simplified form, shows that in fact the above examples are essentially the only obstructions to getting a decay for $\alpha_K(r)$. 

In order to avoid a proliferation of parameters at this early stage, for now we only state a corollary of our more general result, \Cref{lem:key-general}  (the more general result does not put any restriction on the value of $\alpha_K(r)$):

\begin{theorem} [Kruskal-Katona for general convex sets, informal statement]\label{thm:informal-convex-density-increment}
Let $K \subseteq \R^n$ be a convex set which contains the origin-centered ball of radius $\rsm$, i.e.~$B(0^n,\rsm) \subseteq K.$   Let $r> \rsm$ be such that $0.1 \le  \alpha_K(r) \le 0.9$ and let $0 \le \kappa \le 1/10$. Then 
\[
\alpha_K((1-\kappa) r) \geq \alpha_K(r) + \Theta\bigg(\kappa \cdot \frac{ \rsm}{r} \bigg).
\]
%\rnote{I'm not sure of the best / strongest statement of our result in this setting, the following is likely BS\dots} 
%\red{Suppose that $r>1$ is such that $0.01 \leq \alpha_K(r) \leq 0.99$. Then for $r' \leq 0.999r$ we have $\alpha_K(r') \geq \alpha(r) + {\frac {0.0001}{n}}$.}
\end{theorem}
A convex set is \emph{centrally symmetric} if $x \in K$ iff $-x \in K$. For centrally symmetric convex sets we obtain a density increment result without requiring an origin-centered ball to be contained in $K$. As with Theorem~\ref{thm:informal-convex-density-increment},  below we give a special case of our main density increment theorem for centrally symmetric sets (see \Cref{lem:key} for the more general result):

\begin{theorem} [Kruskal-Katona for centrally symmetric convex sets, informal statement]\label{thm:informal-centrally-symmetric-density-increment}
Let $K \subseteq \R^n$ be a centrally symmetric convex set. Let $r>0$ be such that  $0.1 \le  \alpha_K(r) \le 0.9$ and let $0 \le \kappa \le 1/10$. Then
\[
\alpha_K((1-\kappa) r) \geq \alpha_K(r) + \Theta(\kappa).
\]
\end{theorem}
An important feature of our density increment theorems is that while the results are for convex sets in $\mathbb{R}^n$, the density increment statements are independent of $n$. We give an overview of the high-level ideas underlying our \Cref{thm:informal-convex-density-increment,thm:informal-centrally-symmetric-density-increment} in the next subsection.

\subsection{The ideas underlying the Kruskal-Katona type  \Cref{thm:informal-convex-density-increment,thm:informal-centrally-symmetric-density-increment}} \label{sec:high-level-ideas}

Let us provide an intuitive argument for why results of this sort should hold, focusing on \Cref{thm:informal-centrally-symmetric-density-increment} (\Cref{thm:informal-convex-density-increment} uses similar ideas).  At the highest level, a probabilistic argument is used to reduce the $n$-dimensional geometric scenario to a two-dimensional scenario. In more detail, as described below, the proof essentially combines two extremely simple observations with two technical results.

Recalling the setup of \Cref{thm:informal-centrally-symmetric-density-increment}, (after rescaling) we have a symmetric convex body $K \subset \R^n$ whose intersection with the unit sphere $\mathbb{S}^{n-1}_1$ is a certain fraction $\alpha_K(1)$ of $\mathbb{S}^{n-1}_1$. As stated in \Cref{thm:informal-centrally-symmetric-density-increment}, let us think of this fraction as being neither too close to 0 nor to 1.  The goal is to argue that the intersection of $K$ with the slightly smaller sphere $\mathbb{S}^{n-1}_{1-\kappa}$ is a noticeably larger fraction of $\mathbb{S}^{n-1}_{1-\kappa}$. 

The first simple but crucial observation is that the density of $K$ in $\mathbb{S}^{n-1}_1$ is an average of two-dimensional ``cross-sectional'' densities, and the same is true for the density of $K$ in $\mathbb{S}^{n-1}_{1-\kappa}$.  More precisely, the density of $K$ in $\mathbb{S}^{n-1}_1$ is the average over a random two-dimensional subspace $\bV$ of the density of the two-dimensional convex body $K \cap \bV$ in the two-dimensional unit circle obtained by intersecting $\mathbb{S}^{n-1}_1$ with $\bV$, and the same is true for $\mathbb{S}^{n-1}_{1-\kappa}.$ (See \Cref{eq:avg-1} for a precise formulation.)

The next simple but crucial observation is that within any given specific cross-section (two-dimensional subspace $V$), the density of $K$ in the radius-$(1-\kappa)$ circle must be at least the density of $K$ in the radius-$1$ circle.  In other words, within any specific cross-section, ``density is never lost'' by contracting from radius 1 to radius $1-\kappa.$ As mentioned already in the previous subsection, this is an immediate consequence of convexity and the fact that $K$ contains the origin. (See \Cref{fact:convex-decreasing}.)  

Now the first technical result mentioned above enters the picture:  Fix a particular two-dimensional subspace $V$ and suppose that within $V$, the density of $K$ in the radius-$1$ circle is (like the original density of $K$ in the $n$-dimensional unit sphere $\mathbb{S}^{n-1}_1$) neither too close to 0 nor to 1.  Then using elementary geometric arguments and the central symmetry of $K$, it can be shown that the density of $K$ in the radius-$(1-\kappa)$ circle must be ``noticeably higher'' than the density of $K$ in the radius-1 circle --- i.e.~``density is gained'' within this cross-section by contracting.  See Figure~1 for an illustration and \Cref{clm:two-d-increment} for a precise formulation.

\begin{figure}[t] \label{fig:1}
\begin{tikzpicture}
\centerarc[black, thick](20,20)(0:360:5);
\centerarc[black, thick](20,20)(0:360:3);
\centerarc[green, ultra thick](20,20)(0:30:5);
\centerarc[green, ultra thick](20,20)(330:360:5);
\centerarc[green, ultra thick](20,20)(150:210:5);
\draw[ red, ultra thick] ({12},{\Ay})--({28},{\Ay});
\draw[ red, ultra thick] ({12},{\Cy})--({28},{\Cy});
\draw[ red, ultra thick] ({12},{\Cy})--({12},{\Ay});
\draw[ red, ultra thick] ({28},{\Cy})--({28},{\Ay});
\centerarc[blue, ultra thick](20,20)(30:{\newangle}:3);
\centerarc[blue, ultra thick](20,20)({360-\newangle}:330:3);
\centerarc[blue, ultra thick](20,20)({180-\newangle}:150:3);
\centerarc[blue, ultra thick](20,20)({210}:{180+\newangle}:3);
\draw[ black, dashed] ({20},{20})--({\Ayy},{\Ay});
\draw[ black, dashed] ({20},{20})--({\Ayy},{\Ayz});
\draw[ black, dashed] ({20},{20})--({\Ayw},{\Ayz});
\draw[ black, dashed] ({20},{20})--({\Ayw},{\Ay});
%\centerarc[red,ultra thick](20,20)(60:120:5);
\end{tikzpicture}
\caption{
Two concentric circles of radius $1$ and $1-\kappa$ and their intersections with a symmetric convex set.  The boundary of the convex set is indicated in {\color{red}red}. The {\color{green}green} arcs are the portion of the radius-1 circle which intersects the convex set.  Observe that the fraction of the radius-$(1-\kappa)$ circle which intersects the convex set is larger than the fraction of the radius-$1$ circle which intersects the convex set (by the angular measure of the {\color{blue}blue} arcs).}
\end{figure}
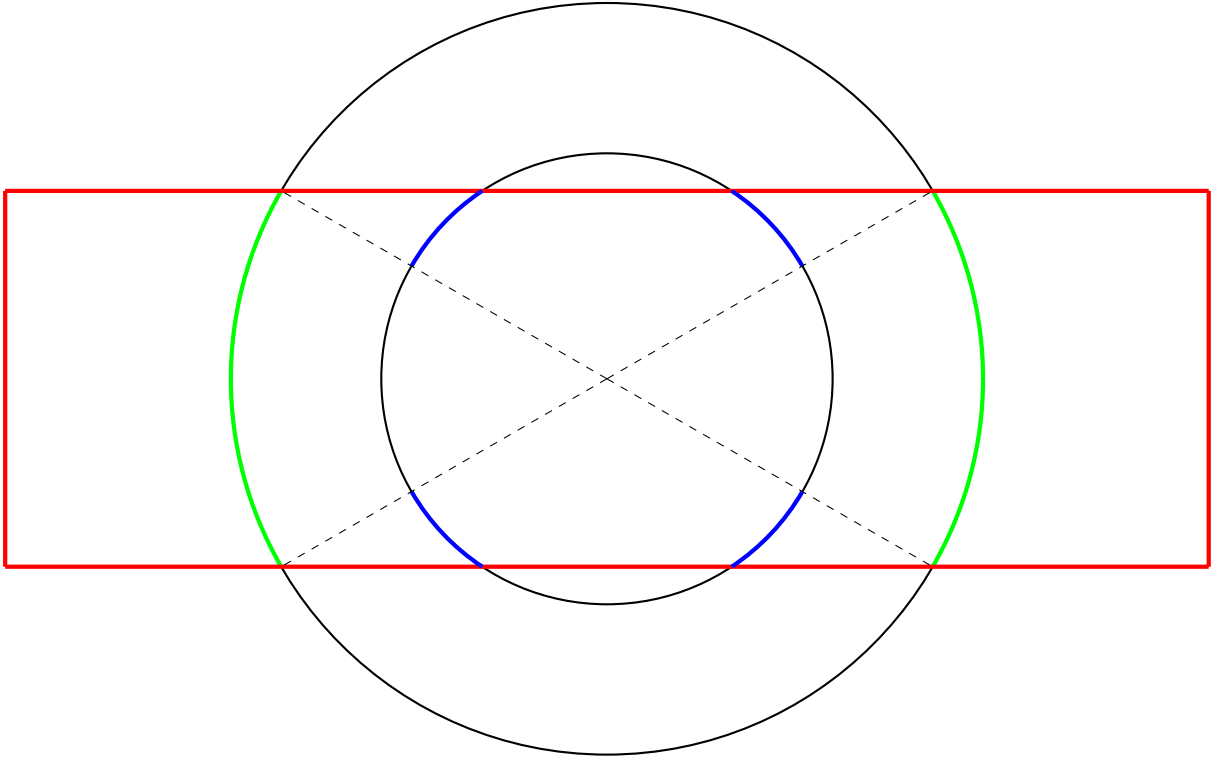

Given this, a natural proof strategy suggests itself:  Suppose that for a random two-dimensional subspace $\bV$, with non-negligible probability the density of $K$ in the two-dimensional circle $\mathbb{S}^{n-1}_1 \cap \bV$ is ``not too far'' from the density of $K$ in the $n$-dimensional sphere $\mathbb{S}^{n-1}_1$. Then, by the preceding paragraph, there would be a noticeable density gain on a non-negligible fraction of subspaces; since density is never lost and the overall density is an average of the density over subspaces, this would give the result.  The second technical result, \Cref{claim:raz}, shows precisely that the above supposition indeed holds.  

We give some elaboration on this second technical result.  It is a variant of a lemma of Raz \cite{raz1999exponential}, who showed that for any subset $A \subset \mathbb{S}^{n-1}$ with $\mu_1(A)$ bounded away from 0 and 1, with high probability a random subspace $\bV$ of $\mathbb{R}^n$ of dimension roughly $1/\epsilon^2$ is such that the density of $A$ in $\mathbb{S}^{n-1} \cap \bV$ is $\pm \eps$-close to the density of $A$ in $\mathbb{S}^{n-1}.$ We establish a variant of this result in a different parameter regime; our requirement is that the measure of $A \cap \bV$ as a fraction of the unit sphere in $\bV$ remain bounded away from $0$ and $1$ with non-negligible probability \emph{even if $\bV$ is a random subspace of dimension only $2$}.  This requires some modification of Raz's original arguments, as we highlight in \Cref{sec:claim:raz:proof}.

\subsection{Applications and consequences of our Kruskal-Katona type results}

\subsubsection{Weak learning under Gaussian distributions}
In \cite{KOS:08} Klivans et al.~showed that convex sets are \emph{strongly learnable} (i.e.~learnable to accuracy $1-\epsilon$ for any $\eps > 0$) in time $n^{{O}(\sqrt{n}/\epsilon^2)}$ under the Gaussian distribution, given only random examples drawn from the Gaussian distribution. Up to a mildly better dependence on $\epsilon$, this  matches the running time of the algorithm of \cite{BshoutyTamon:96} for learning monotone functions on the hypercube. 

However, there is a large gap in the state of the art between monotone Boolean functions on the cube and convex sets in the Gaussian space when it comes to \emph{weak learning}. In particular, while \cite{OWimmer:09} showed that monotone functions can be weakly learned to accuracy $1/2 + \Omega(n^{-1/2} \log n)$ in polynomial time, prior to this work nothing better than the $n^{\sqrt{n}}$ running time of~\cite{KOS:08} was known for weakly learning convex sets to any nontrivial accuracy (even accuracy $1/2 + \exp(-n)$). In particular, the \cite{KOS:08} result 
%which shows that any convex body can be learned to accuracy $\eps$ in $n^{n^{1/2}/\eps^2}$ time using the same number of examples, 
in and of itself does not imply anything about polynomial-time weak learning; the \cite{KOS:08} result is proved using Hermite concentration, but prior to this work it was conceivable that all of the Hermite weight of a convex body might sit at levels $\Omega(n^{1/2})$, which would necessitate an $n^{\Omega(\sqrt{n})}$ runtime for the \cite{KOS:08} algorithm.

The main algorithmic contribution of this paper is to bridge this gap and give a polynomial-time weak learning algorithm for convex sets. We prove the following:\footnote{As stated  below \Cref{thm:weak-learn-convex} is only for learning under the standard Gaussian distribution $N(0,1^n)$, but since convexity is preserved under affine transformations, the result carries over to weak learning with respect to any Gaussian distribution $N(\mu,\Sigma).$} 
\begin{theorem} [Weak learning convex sets] \label{thm:weak-learn-convex}
There is a $\poly(n)$-time algorithm which uses only random samples from $N(0,1)^n$ and weak learns any unknown convex set $K \subseteq \R^n$ to accuracy $1/2 + \Omega(1/n)$ under $N(0,1^n)$ .
\end{theorem}

%As mentioned earlier, the state of the art in strong learnability of monotone functions over $\{\pm 1\}^n$ and convex sets over the Gaussian measure is the same -- namely, recall that Bshouty and Tamon~\cite{BshoutyTamon:96} showed that 

%Motivate:  strong learning under Gaussian distributions is well understood [KOS07], and weak learning of monotone Boolean functions is extremely well understood.  

%In contrast, nothing better than the $n^{\sqrt{n}}$ running time of [KOS07] was known for weak learning monotone functions to any nonzero (even inverse exponential) accuracy.  

%State our weak learning results (can give full details here):

For centrally symmetric convex sets we give a result which is stronger in two ways.  We achieve a stronger  advantage, and we show that one of three fixed hypotheses always achieves this stronger advantage:  the empty set, all of $\R^n$, or the origin-centered ball of radius $r_{\median}$ where $r_{\median}$ is the median of the chi-distribution with parameter $n$. 
This result is as follows:\footnote{Similar to the previous footnote, since a centrally symmetric convex set remains centrally symmetric and convex under any linear transformation, \Cref{thm:weak-learn-centrally-symmetric} directly implies an analogous result for any origin-centered Gaussian distribution $N(0^n,\Sigma).$}

\begin{theorem} [Weak learning centrally symmetric convex sets] \label{thm:weak-learn-centrally-symmetric}
For any centrally symmetric convex set $K \subseteq \R^n$, one of the following three hypotheses $h$ has $\Pr_{\bx \sim N(0,1)^n}[h(\bx)=K(\bx)] \geq 1/2 + \Omega(1/\sqrt{n})$: $h=$ the empty set, $h=$ all of $\R^n$, or $h=$ the origin-centered ball of radius $r_{\median}$.
\end{theorem}
From \Cref{thm:weak-learn-centrally-symmetric} it is straightforward to get a $\poly(n)$-time learner for centrally symmetric convex sets with advantage $\Omega(1/\sqrt{n})$. This is entirely analogous to the result of \cite{BBL:98}, who showed that for any monotone function $f: \bn \to \bits$ over the Boolean hypercube, one of the following three functions achieves an advantage of  $\Omega( n^{-1/2})$ with respect to the uniform distribution:  the constant $1$ function, the constant $-1$ function, or the majority function.
We note that the main technical ingredient in proving \Cref{thm:weak-learn-convex} (respectively \Cref{thm:weak-learn-centrally-symmetric}) is \Cref{thm:informal-convex-density-increment} (respectively \Cref{thm:informal-centrally-symmetric-density-increment}). In particular, the proof of \Cref{thm:weak-learn-centrally-symmetric} is a modification of the argument of \cite{BBL:98} which uses the Kruskal-Katona theorem (over the hypercube) to show that one of the functions $\{+1,-1,\mathsf{MAJ}\}$ is a good weak hypothesis for any monotone Boolean function.

\medskip

\noindent {\bf A lower bound for weak learning convex sets.}
We complement \Cref{thm:weak-learn-centrally-symmetric} with an information theoretic lower bound. This lower bound shows that any $\poly(n)$-time algorithm, even one which is allowed to query the target function on arbitrary inputs of its choosing, cannot achieve a significantly better advantage than our simple algorithm achieves for centrally symmetric convex sets:

\begin{theorem} [Lower bound for weak learning centrally symmetric convex sets] \label{thm:our-BBL-lb}
For sufficiently large $n$, for any $s \geq n$, there is a distribution ${\cal D}$ over centrally symmetric convex sets with the following property:  for a target convex set $\boldf \sim {\cal D},$ for any membership-query (black box query) algorithm $A$ making at most $s$ many queries to $\boldf$, the expected error of $A$ (the probability over $\boldf \sim {\cal D}$, over any internal randomness of $A$, and over a random Gaussian $\bx \sim N(0,1^n)$, that the output hypothesis $h$ of $A$ predicts incorrectly on $\bx$) is at least $1/2 - {\frac {O(\log(s) \cdot \sqrt{\log n})}{n^{1/2}}}$.\ignore{\red{$1/2 - {\frac {O(\log(s))}{n^{1/2}}}$}.}
\end{theorem}
Theorem~\ref{thm:our-BBL-lb} shows that the advantage of our weak learner for centrally symmetric convex sets (Theorem~\ref{thm:weak-learn-centrally-symmetric}) is tight up to a logarithmic factor for polynomial time algorithms. It follows that the advantage of our weak learner for general convex sets (Theorem~\ref{thm:weak-learn-convex}) is tight  up to a quadratic factor. 

\subsubsection{Noise stability / low-degree Hermite weight of convex sets}

A well known structural fact about monotone functions is that they cannot have too little Fourier weight (sum of squared Fourier coefficients) at levels 0 and 1: every monotone function $f: \bn \to \bits$ has at least $\Omega({\frac {\log^2 n}{n}})$ amount of Fourier weight on levels 0 and 1, and further, this $\Omega({\frac {\log^2 n}{n}})$ lower bound is best possible. (For the sake of completeness we give a proof of this in~\Cref{app:hermite-weight}.)

In contrast, it is easy to see that a convex set can have zero Hermite weight (sum of squared Hermite coefficients) at levels 0 and 1:  this is the case for any centrally symmetric set of Gaussian volume $1/2$ (having Gaussian volume $1/2$ implies that the degree-0 Hermite coefficient is zero, and central symmetry implies that each degree-1 Hermite coefficient is also zero.) It is also possible for a convex set to have zero Hermite weight at levels 0 and 2; indeed the set $\{x: x_1 \geq 0\}$ is one such set.  However, we show that any convex set must have some non-negligible Hermite mass at levels 0, 1 and 2.  In particular, we show that the Hermite level $0$-and-$2$ weight of any centrally symmetric convex sets must be at least $\Omega(1/n)$:

\begin{theorem}~\label{thm:centrally-symmetric-weight}
Let $K$ be a centrally symmetric convex set (viewing it as a function $K: \mathbb{R}^n \rightarrow \bits$). Then the Hermite weight of $K$ at levels $0$ and $2$ is at least $\Omega(1/n)$. 
\end{theorem}

Since a suitably scaled origin-centered cube has Hermite weight $O(\log^2(n)/n)$ at levels $0$ and $2$ (see \Cref{fact:cube}), our lower bound is tight up to logarithmic factors.\ignore{\gray{In fact, the $\tilde{O}(1/n)$ upper bound on the Hermite weight holds for the origin centered cube up to any $O(1)$-levels.}\rnote{I don't quite understand this, can we rephrase?}} For general convex sets we prove a quadratically weaker lower bound:
 \begin{theorem}~\label{thm:centrally-asymmetric-weight}
Let $K$ be an arbitrary convex set (viewing it as a function $K: \mathbb{R}^n \rightarrow \bits$). Then the Hermite weight at levels $0$, $1$ and $2$ is at least $\Omega(1/n^2)$. 
\end{theorem}

\paragraph{Noise stability of convex sets at high noise rates.}
One motivation for understanding the low-level Hermite weight of bodies in $\R^n$ comes from its connection to the notion of noise stability. Recall that for $t \ge 0$ the \emph{Ornstein-Uhlenbeck} operator $P_t$ in the Gaussian space is defined as follows: for $f: \mathbb{R}^n \rightarrow \mathbb{R}$ the function $P_t f: \mathbb{R}^n \rightarrow \mathbb{R}$ is defined as
\[
P_tf(x) := \Ex_{\by \sim N(0,1)^n} [f(e^{-t}x + \sqrt{1-e^{-2t}} \by)]. 
\]
With the definition of $P_t$, the noise stability of $f$  at noise rate $t$ (denoted by $\mathsf{Stab}_t(f)$ ) is defined to be
\[
\mathsf{Stab}_t(f) := 
\Ex_{\bx}[f(\bx) P_t f(\bx)] = \Ex_{\bx, \by \sim N(0,1)^n} [f(\bx) f(e^{-t}\bx + \sqrt{1-e^{-2t}} \by)]. 
\]
The quantity $\mathsf{Stab}_t(f)$ is a measure of how sensitive $f$ is to perturbation in its input. In particular, as $t$ becomes large $\mathsf{Stab}_t(f)$ measures the correlation of $f$ on positively but only very mildly correlated inputs. On the other hand, as $t\rightarrow 0$, for $\bits$-valued functions $f$ this quantity measures the so-called \emph{Gaussian surface area} of the region $f^{-1}(1)$ (denoted by $\mathsf{surf}(f)$). If the set $\mathcal{A} =f^{-1}(1)$ has a smooth or piecewise smooth boundary,  $\mathsf{surf}(f)$ is defined as
\[
\mathsf{surf}(f) = \int_{x \in\partial \mathcal{A}} \gamma_n(x) d\sigma(x), 
\]
where $\gamma_n(x) = (2\pi)^{-n/2} \cdot \exp(-\Vert x\Vert_2^2/2)$ is the standard Gaussian measure, $\partial \mathcal{A}$ is the boundary of $\mathcal{A},$ and $d\sigma(x)$ is the standard surface measure of $\mathbb{R}^n$. Ledoux~\cite{Ledoux:94}  (and implicitly, earlier Pisier~\cite{Pisier:86}) showed that 
for $t>0$, 
\begin{equation}~\label{eq:surfnoise}
\Prx_{\bx, \by \sim N(0,1)^n} [f(\bx)  = f(e^{-t} \bx + \sqrt{1-e^{-2t}} \by)] \ge 1- \frac{2\sqrt{t}}{\sqrt{\pi}} \mathsf{surf}(f). 
\end{equation}
Hence when $t$ is small the surface area of $f$ provides a good lower bound on the noise stability of $f$ (and as $t\rightarrow 0$ this inequality in fact becomes tight). 
We refer to  \cite{Janson:97} for a detailed discussion.

In \cite{Ball:93} Ball showed that the Gaussian surface area of any 
convex set $K \subseteq \mathbb{R}^n$ is at most $O(n^{1/4})$. Consequently, we get that for any convex set $K : \mathbb{R}^n \rightarrow \{-1,1\}$, 
\begin{equation} \label{eq:stab}
\mathsf{Stab}_t(K) \ge 1- O(\sqrt{t} \cdot n^{1/4}).
\end{equation}
While this bound is meaningful for $t = o(n^{-1/2})$, it is vacuous once $t$ exceeds $Cn^{-1/2}$. In fact, the above inequality can be extended (see \cite{KOS:08}) to show that 
$
\mathsf{Stab}_t(K) \ge \exp (- O(t\sqrt{n})); 
$ however this bound is still quite weak for $t= \omega(n^{-1/2})$.

Theorems~\ref{thm:centrally-symmetric-weight} and~\ref{thm:centrally-asymmetric-weight} yield the first nontrivial noise stability lower bounds for convex sets for large $t$. These bounds follow from a simple and standard fact from Hermite analysis (see Proposition~11.37 of \cite{ODBook}) which is that 
$\mathsf{Stab}_t(f) \ge e^{-2t} W^{\le 2}[f]$ where $W^{\le 2}[f]$ is the weight at levels $0$, $1$ and $2$ of the Hermite spectrum of $f$. Combining this fact with \Cref{thm:centrally-symmetric-weight,thm:centrally-asymmetric-weight}, we get the following corollary:

\begin{corollary}~\label{corr:noise-stab}
For any $t \ge 0$ and any convex set $K \subseteq \R^n$, it holds that $\mathsf{Stab}_t(K) \ge e^{-2t}/{n^2}.$ This bound can be improved to $\mathsf{Stab}_t(K) \ge e^{-2t} /{n}$ if $K$ is centrally symmetric and convex.
\end{corollary}

We note that the bound given by \Cref{corr:noise-stab} is significantly better than the bound that follows from \cite{Ball:93} as described above for $t \gg \frac{\log n}{\sqrt{n}}$. Beyond the quantitative aspect, we feel that there is an interesting qualitative distinction between our noise stability bound and the noise stability bound that follows from~\cite{Ball:93} for convex sets (as well as several others in the literature, as we explain below). 

Roughly speaking, one can analyze the noise stability of sets in two limiting cases of noise rates.  The first (i) is the \emph{low noise rate (LNR) regime}: here the noise rate $t$ is close to zero and hence the correlation $e^{-t}$ is close to $1$.  The second (ii) is the \emph{high noise rate (HNR) regime}: here  the noise rate $t$ is large so the correlation $e^{-t}$ is close to $0$. Note that the noise stability bounds of Ball \eqref{eq:stab}  are most interesting in the LNR regime, whereas ours (\Cref{corr:noise-stab}) are most interesting in the HNR regime.  A number of other results in the literature, such as \cite{Nazarov:03, Kane14GL, DRST14, HKM12}, also show  noise stability bounds for various families of functions (such as polytopes with few facets and low-degree polynomial threshold functions) in the LNR regime.

The regime of applicability of noise stability bounds is closely connected to the methods used to prove those bounds. In particular, noise stability in the LNR regime is essentially controlled by the surface area of a set (this connection is made explicit in \Cref{eq:surfnoise}): if the surface area is low then the noise stability is high and vice versa. Not surprisingly, since surface area is a geometric quantity, methods of geometric analysis are used to show noise stability bounds in the LNR regime (as in \cite{Ball:93} and the other works cited above). However, this connection between surface area and noise stability breaks down in the HNR regime.  For intution on why this occurs, note that the noise stability of $K:\mathbb{R}^n \rightarrow \{-1,1\}$ at noise rate $t$ is captured by $\Pr[K(\bg) =   K \ (\bg' )]$ where $(\bg,\bg')$ is a correlated pair of $n$-dimensional Gaussians with variance $1$ and correlation $e^{-t}$ in each coordinate. In the LNR regime,  i.e., when $e^{-t}$ is close to $1$, $\bg$ and $\bg'$ can be visualized as tiny perturbations of each other, and it is intuitively plausible that the above probability can be understood by analyzing the geometry of $K$. However, 
in the HNR regime, $e^{-t}$ approaches $0$ and hence (at least at an intuitive level) the resulting geometric picture is essentially indistinguishable from the case when $e^{-t} =0$, i.e., $\bg$ and $\bg'$ are totally uncorrelated. Thus, it is not clear how geometric arguments can be used to prove 
lower bounds on noise stability in the HNR regime. In fact, 
\Cref{app:stability} gives an example of two functions which have the same surface area (and hence essentially the same noise stability in the LNR regime), but in the HNR regime the noise stability of the second function is exponentially worse than the first.  (In fact, the proof of this separation relies on \Cref{corr:noise-stab} --- one of the functions is convex and the other is not.) 

Despite the above intuitions, somewhat surprisingly our noise stability lower bounds are in fact established using geometric arguments.  Our geometric arguments do not directly analyze mildly correlated Gaussians; rather, we use the (easy to prove but) deep connection between noise stability in the HNR regime  and correlation with  low degree polynomials. The phenomenon of noise stability in the HNR regime being completely controlled by correlation with low degree polynomials is, in our view, a completely non-geometric one; it relies on the fact that 
Hermite polynomials are eigenfunctions of the noise operator, or in more detail, on the fact that for every vector $\alpha \in \N^n$ the
Hermite polynomial $h_\alpha$ is an eigenfunction of the noise operator with $e^{-t|\alpha|}$ as the corresponding eigenvalue.   
 (See \Cref{sec:hermite} for basic background
on Hermite analysis). Equipped with this connection, we use geometric arguments to show that convex bodies are either correlated with a degree-$1$ polynomial or else with one special type of degree-$2$ polynomial, namely one corresponding to a Euclidean ball.  

We close this discussion by remarking that, as mentioned before, several papers~\cite{Nazarov:03, Kane14GL, DRST14, HKM12} have studied noise stability in the LNR regime, but much less appears to be known about noise stability in the HNR regime. We believe that studying the noise stability of Boolean-valued functions in the HNR regime is well motivated both from the vantage point of structural analysis and through algorithmic applications such as learning. Our results can be viewed as a step in this direction.

\subsection{Directions for future work}

Our results in this paper suggest a number of directions for future work; we close this introduction with a brief discussion of some of these.

One natural goal is to establish quantitatively sharper versions of our results for general convex sets.  While the weak learning results and bounds on low-degree Hermite concentration in this paper are essentially best possible (up to log factors) for centrally symmetric convex sets, there is potentially more room for improvement in our results for general convex sets.

In a different direction, the basic Kruskal-Katona theorem for monotone Boolean functions has been extended in a number of different ways.  Keevash \cite{Keevash:08} and O'Donnell and Wimmer \cite{OWimmer:09} have given incomparable ``stability'' results which extend the Kruskal-Katona theorem by giving information about the approximate structure of monotone functions for which the Kruskal-Katona density increment lower bound is close to being tight. In particular, \cite{OWimmer:09} show that (under mild conditions) if a monotone function $f: \bn \to \bits$ is not noticeably correlated with any single coordinate when restricted to the $k$-th slice of $\bn$, then the density increment $\mu_{k+1}(f) - \mu_k(f)$ must be at least $\Omega({\frac {\log n}{n}})$, strengthening the $\Omega({\frac 1 { n}})$ lower bound which follows from the original Kruskal-Katona theorem.  \cite{OWimmer:09} use this sharper result to give a $\poly(n)$-time weak learning algorithm for monotone functions that achieves advantage $\Omega({\frac {\log n}{\sqrt{n}}})$, which is the best possible by the lower bound of \cite{BBL:98}. As another extension, in \cite{Bukh12} Bukh proves a multidimensional generalization of the Kruskal-Katona theorem.  An intriguing goal for future work is to investigate possible Gaussian space analogues of the \cite{Keevash:08,OWimmer:09,Bukh12} results. In particular, if a Gaussian space analogue of \cite{OWimmer:09} could be obtained, this might lead to a $\poly(n)$-time weak learner for centrally symmetric convex sets achieving advantage $\Omega({\frac {\log n}{\sqrt{n}}})$, which would be quite close to optimal by our lower bound result \Cref{thm:our-BBL-lb}.

A last goal is to obtain quantitatively stronger weak learning results for convex sets that have a ``simple structure.''  In addition to giving an $n^{O(\sqrt{n})}$-time strong learning algorithm for general convex sets, \cite{KOS:08} also gave an $n^{O(\log k)}$-time strong learning algorithm for convex sets that are intersections of $k$ halfspaces.  Is there a $\poly(n)$-time weak learning algorithm that achieves accuracy $1/2 + o(1/\sqrt{n})$ for intersections of a small number of halfspaces?

\section{Kruskal-Katona for convex sets} \label{sec:main}

%\noindent {\bf Basic results on convex sets.}
%For any convex set $K \subseteq \R^n$, we define the associated function $K :\mathbb{R}^n \rightarrow \{\pm 1\}$  such that $K(x) =1$ iff $x$ lies in the convex set $K$.  

In this section we give formal statements and proofs of our main structural results, \Cref{lem:key,lem:key-general} below, which are analogues of the Kruskal-Katona theorem for convex and centrally symmetric convex sets.
To do this, we first recall the definition of the shell density function $\alpha_K (\cdot)$  from \Cref{eq:shell-density-def}: for $r \geq 0$,
\[
\alpha_K(r) := \Prx_{\bx \in \mathbb{S}^{n-1}_r} [\bx \in K]. 
\]
So $\alpha_K(r)$ is equal to the fraction of the origin-centered radius-$r$ sphere which lies in $K$.  (A view which will be useful later is that it is the probability that a random Gaussian-distributed point $\bg \sim N(0,1)^n$ lies in $K$, conditioned on $\|\bg\|=r.$) An easy fact about the function $\alpha_K(\cdot)$ is the following: 
\begin{fact}~\label{fact:convex-decreasing}
If $K$ is convex and $0^n \in K$ then $\alpha_K(\cdot)$ is non-increasing. 
\end{fact}
\begin{proof}
By convexity, if $x \in K$ then $\lambda x \in K$ for any $\lambda \in [0,1]$. This immediately implies that $ \Prx_{\bx \in \mathbb{S}^{n-1}_r} [\bx \in K] \leq \Prx_{\bx \in \mathbb{S}^{n-1}_{\lambda r}} [\bx \in K]$ and consequently $\alpha_K(\cdot)$ is non-increasing. 
\end{proof}

We begin with our analogue of the Kruskal-Katona theorem for centrally-symmetric convex sets, since it is somewhat easier to state.  The following is a more general version of \Cref{thm:informal-centrally-symmetric-density-increment}:
\begin{theorem} [Kruskal-Katona for centrally symmetric convex sets] \label{lem:key}
Let $K \subset \mathbb{R}^n$ be a centrally symmetric convex body and let $r>0$ be such that $\alpha
_K(r) \in (0,1)$. Let $0 < \kappa <1/10$. Then 
\[
\alpha_K (r (1-\kappa))  \geq \alpha_K(r) + \kappa \cdot \Theta((\alpha_K(r)(1-\alpha_K(r)))^2).
\]
%\begin{eqnarray*}
%\alpha_K (r (1-\kappa))  \geq \begin{cases}   \alpha_K(r) + \Theta(\kappa \cdot \alpha_K^2(r)) \ &\textrm{if } \ 0 < \alpha_K(r) \le 1/2 \\ 
%\alpha_K(r) + \Theta(\kappa \cdot (1-\alpha_K(r))^2) \ &\textrm{if } \ 1/2 < \alpha_K(r) \le 1 \\
%\end{cases} 
%\end{eqnarray*}
%%Suppose that $r$ is such that $\alpha(r) \in (0,1).$ Let $0< \kappa \le 1/10$. Then 
%%\[
%%\alpha (r (1-\kappa)) \geq \alpha (r) + \Theta(\kappa) \cdot \min\{ \alpha(r), (1-\alpha(r)) \}^2. 
%%\]
\end{theorem}
Intuitively, the above theorem says that at any input $r$ where the shell density function $\alpha_K(r)$ is not too close to 0 or 1, slightly decreasing the input $r$ will cause the shell density function to noticeably increase. 

As was noted earlier, such a density increment statement does not hold for general convex bodies $K$ that contain the origin (for example, if $K = \{x: x_1 \ge 0\}$ then $\alpha_K(r) =1/2$ for all $r>0$). The next theorem establishes a density increment for general convex bodies that contain an origin-centered ball (and implies \Cref{thm:informal-convex-density-increment}):

\begin{theorem} [Kruskal-Katona for general convex sets] \label{lem:key-general}
Let $K \subset \mathbb{R}^n$ be a  convex body that contains the radius-$r$ origin-centered ball $B(0^n, r_{\mathsf{small}})$.  Let $r$ be such that $\alpha
_K(r) \in (0,1)$ and let $0 < \kappa <1/10$. Then 
\begin{eqnarray*}
\alpha_K (r (1-\kappa))  \geq \begin{cases}   \alpha_K(r) + \Theta\big(\kappa \cdot \alpha_K(r
) \cdot \frac{ r_{\mathsf{small}}}{r}\big) \ &\textrm{if } \ 0 < \alpha_K(r) \le 1/2 \\ 
\alpha_K(r) + \Theta\big( \kappa \cdot (1-\alpha_K(r)) \cdot \min \big\{1-\alpha_K(r), \frac{r_{\mathsf{small}}}{r} \big\} \big)
 \ &\textrm{if } \ 1/2 < \alpha_K(r) \le 1. \\
\end{cases} 
\end{eqnarray*} 
\end{theorem}
Note that the increment in the above theorem is linearly dependent on $r_{\mathsf{small}}/r$ (and thus vanishes when no origin-centered ball is contained in $K$). It is not difficult to see that this dependence is best possible by considering some fixed $r_{\mathsf{small}}>0$ and the convex set $K = \{x \in \mathbb{R}^n: x_1 + r_{\mathsf{small}}\ge 0\}$.

%\begin{lemma} \label{lem:key}
%%Suppose that $r$ is such that $\alpha(r) \in (0,1).$ Let $p(n),q(n)$ be such that $q(n) =\Omega( p(n))$ and $p(n)$ is at least some sufficiently large constant. Then $\alpha((1-{\frac 1 {p(n)}})r) \geq \alpha(r) + {\frac 1 {q(n)}}$.
%\end{lemma}

%It is instructive to consider how \Cref{lem:key} aligns with the goal, stated at the end of the preceding subsection, of showing that ``$\alpha(\cdot)$ does not stay too flat for too long around $r_{1/2}$.''
%We prove \Cref{lem:key} in Section~\ref{sec:proof-of-key-lemma}. 

\subsection{Proof of \Cref{lem:key}} \label{sec:proof-of-key-lemma}
%\red{We started discussing this in Rhodes: We want a version of this which doesn't require $\alpha(r) \in [0.4,0.6]$ but rather gives us something for any $\alpha(r) \in (0,1)$.}
The proofs of \Cref{lem:key} and \Cref{lem:key-general} have a substantial overlap; in particular, the first part of the proofs are identical.
To avoid repetition we will explicitly note the places where the two proofs diverge. 

We now start with the proof of \Cref{lem:key}. Note that because $K$ contains the origin, $\alpha_K(\cdot)$ is a non-increasing function. Set $\beta = \min \{\alpha_K(r), 1- \alpha_K(r)\}$ so that $0 < \beta \leq 1/2$ and thus $\alpha_K(r) \in [\beta, 1-\beta]$. For simplicity of exposition, we now rescale the convex body by a factor of $1/r$; after this rescaling we have that $\alpha_K(1)
 \in [\beta, 1-\beta]$, and thus we need to prove a lower bound on $\alpha_K(1-\kappa)$.

%Recall the setup of \Cref{lem:key}: $K$ is a centrally symmetric convex body and $r>0$ is a value such that the shell density function $\alpha = \alpha_K$ associated with $K$ has $\alpha(r) \in [\beta, 1 - \beta].$ For simplicity we rescale $K$ by a factor of $1/r$ so that $\alpha(1) \in [\beta, 1-\beta].$  

Let $C_1 := K \cap \mathbb{S}^{n-1}$, and let us write\footnote{We include the subscript $1$ on $\mu$ because we will soon be considering spheres of radii other than 1.} $\mu_1(C_1)$ to denote the measure of $C_1$ as a fraction of $\mathbb{S}^{n-1}$, so $\mu_1(C_1)$ satisfies $\mu_1(C_1) \in [\beta, 1-\beta].$   In other words, under the Haar measure (i.e.~the uniform distribution) on the unit sphere, $\mu_1(C_1)$ is the probability that a randomly drawn point lies in $C_1$.  

Our argument makes crucial use of a variant of a lemma due to Raz~\cite{raz1999exponential}. In particular, Raz showed that for any subset $A \subset \mathbb{S}^{n-1}$ with $\mu_1(A)$ bounded away from 0 and 1\ignore{\rnote{This was a condition that Raz required, right - didn't Klartag and Regev dial this down to $e^{-n^{1/3}}$ or something?}} and a random subspace $\bV$ of $\mathbb{R}^n$ of dimension $\approx 1/\epsilon^2$, the Haar measure of $A \cap \bV$ (as a fraction of the unit sphere in $\bV$) is $\epsilon$-close to $\mu_1(A)$  with high probability.  
We adapt Raz's arguments to show a variant of this result. Roughly speaking, our variant implies that under the above conditions, \ignore{in particular, provided the measure of $A$ is bounded away from $0$ and $1$,} the measure of $A \cap \bV$ as a fraction of the unit sphere in $\bV$ is bounded away from $0$ and $1$ with non-negligible probability \emph{even if $\bV$ is a random subspace of dimension only $2$}. This variant is useful for us because once the ambient dimension is $2$, we can use elementary geometric arguments to prove \Cref{lem:key}. 

We now state our variant of Raz's lemma:

\begin{claim}
[Variant of the main lemma of \cite{raz1999exponential}]
 \label{claim:raz}
Let $\bV$ be a uniform random 2-dimensional subspace of $\R^n$ and let $C$ be a subset of $\mathbb{S}^{n-1}$ such that 
$\mu_1(C) \in [\beta,1/2]$ for $0 < \beta \le 1/2$.
 Then 
$$
 \Prx_{\bV} [\mu_{\bV,1}(C \cap \bV) \in [\beta/4, 9/10]] \ge \frac{\beta}{2}. 
 $$
Similarly, if $\mu_1(C) \in [1/2,1-\beta]$, then 
$$
 \Prx_{\bV} [\mu_{\bV,1}(C \cap \bV) \in [1/10, 1-\beta/4]] \ge \frac{\beta}{2}. 
 $$ 
 Here $\mu_{\bV,1}(C \cap \bV)$ denotes the measure of $C \cap \bV$ as a fraction of the unit sphere $\bV \cap \mathbb{S}^{n-1}$ of the $2$-dimensional subspace $\bV$.
\end{claim}
 %(We write the ``$1$'' subscript because we are thinking about the unit sphere; we will consider a sphere with a different radius soon.)  
 We defer the proof of \Cref{claim:raz} to \Cref{sec:claim:raz:proof} and continue with the proof of \Cref{lem:key} assuming \Cref{claim:raz}. 
For $ 0 < \kappa <1$ let $\mathbb{S}^{n-1}_{1-\kappa}$ denote the origin-centered $n$-dimensional sphere of radius $1-\kappa$.  Define $C_{1-\kappa} := K \cap \mathbb{S}^{n-1}_{1-\kappa}$ and let $\mu_{1-\kappa}(C_{1-\kappa})$ denotes the fractional density of $C_{1-\kappa}$ in $\mathbb{S}^{n-1}_{1-\kappa}$. 
Note that $\alpha_K(1) = \mu_1(C_1)$
and 
$\alpha_K(1-\kappa) = \mu_{1-\kappa}(C_{1-\kappa})$. 
For $V$ any 
$2$-dimensional subspace of $\mathbb{R}^n$, define $K_V := K \cap V$, $C_{V, 1} := K \cap V \cap \mathbb{S}^{n-1},$ and $C_{V, 1-\kappa} :=  K \cap V \cap \mathbb{S}^{n-1}_{1-\kappa}$. 
%\rnote{I changed the notation; these had been $C_{V,0}$ and $C_{V,\kappa}$ respectively but it seemed to align better with the $\mu$ notation to have it the way it is now. To do: check that I didn't botch the change and that it's consistent throughout.}
We further define $\mu_{V,1}(C_{V, 1})$ (respectively $\mu_{V,1-\kappa}(C_{V, 1-\kappa})$)  as the measure of $C_{V, 1}$ (respectively $C_{V, 1-\kappa}$) as a fraction of $\mathbb{S}^{n-1}_{1}\cap V$ (respectively $\mathbb{S}^{n-1}_{1-\kappa}\cap V$). 
Note that $\mathbb{S}^{n-1}_{1}\cap V$ (respectively $\mathbb{S}^{n-1}_{1-\kappa}\cap V$) is the origin-centered $2$-dimensional sphere of radius $1$ (respectively $1-\kappa$) inside the subspace $V$.

\Cref{lem:key} and \Cref{lem:key-general} are essentially lower bounds on $ \mu_{1-\kappa}(C_{1-\kappa}) - \mu_1(C_1)$. 
To establish these lower bounds, 
we first observe that the density of $K$ in an $n$-dimensional sphere is an average of two-dimensional ``cross-sectional'' densities; more precisely, for $\bV$ a uniform random $2$-dimensional subspace of $\R^n$, we have that
\begin{equation}~\label{eq:avg-1}
\mu_1(C_1) = \Ex_{\bV} [\mu_{\bV,1}(C_{\bV,1})] \quad \quad \textrm{and} \quad \quad \mu_{1-\kappa}(C_{1-\kappa}) = \Ex_{\bV} [\mu_{\bV,1-\kappa}(C_{\bV,1-\kappa})]. 
\end{equation}
Another simple but crucial observation is that for any fixed $2$-dimensional subspace $V$, it follows directly from \Cref{fact:convex-decreasing}  that 
\begin{equation}~\label{eq:avg-2}
\mu_{V,1}(C_{V, 1}) \leq \mu_{V,1-\kappa}(C_{V, 1-\kappa}).
\end{equation}

The high level idea of our argument is to strengthen \Cref{eq:avg-2} to a strict inequality for a non-negligible fraction of subspaces $V$ and thereby by \Cref{eq:avg-1} obtain an overall density increment.
%We will now strengthen the above inequality. 
%Let $\mu_{V,1}(C_{V, 1}) = p \in (0,1)$. 
Towards this end, let us partition 
$C_{V,1-\kappa}$ into two sets $A_{K,V}$ and $B_{K,V} = C_{V,1-\kappa} \setminus A_{K,V}$ as follows: 
\[
A_{K,V} := \bigg\{z \in \mathbb{S}^{n-1}_{1-\kappa} \cap V: \frac{1}{1-\kappa} \cdot z \in C_{V, 1} \bigg\}. 
\]
We observe that $\mu_{V,1-\kappa}(A_{K,V}) = \mu_{V,1}(C_{V, 1})$, and hence we have that
\begin{eqnarray}~\label{eq:diff-reexp}
\mu_{V,1-\kappa}(C_{V, 1-\kappa})- \mu_{V,1}(C_{V, 1}) =
\mu_{V,1-\kappa}(B_{K,V}). 
\end{eqnarray}
The next claim proves a lower bound on $\mu_{V,1-\kappa}(B_{K,V})$. We note that this is the first and essentially the only point of departure between the proofs of \Cref{lem:key} and \Cref{lem:key-general}. 

\begin{claim}~\label{clm:two-d-increment}
Let $\mu_{V,1}(C_{V, 1}) =p \in (0,1)$. Then for all $0 \le \kappa \le \frac{1}{10}$, we have that
\[
\mu_{V,1-\kappa}(B_{K,V}) \ge \frac{2 \pi \cdot \kappa \cdot (1-p)}{2} \cdot \sin (\pi \cdot p/2). 
\]
\end{claim}
\begin{proof}
In this part of the proof we will refer to $\mathbb{S}_1^{n-1} \cap V$ as ``the unit circle'' and to $\mathbb{S}_{1-\kappa}^{n-1} \cap V$ as ``the circle of radius $1-\kappa$.'' 
Observe that $C_{V, 1}$ is a subset of the unit circle, and let us partition $\overline{C_{V, 1}}$ into
a collection of disjoint arcs (whose end points belong to $C_{V, 1}$). Now, for any such arc $\mathcal{F}$, define $\mathcal{F}_{1-\kappa} := \{z (1-\kappa) : z\in \mathcal{F}\}$. Now we note three simple but crucial facts: (i) $\mathcal{F}_{1-\kappa} \subseteq B_{K,V}$; (ii) if $\mathcal{F}$ and $\mathcal{G}$ are disjoint arcs then so are $\mathcal{F}_{1-\kappa}$ and $\mathcal{G}_{1-\kappa}$; and (iii) the angle of any such arc $\mathcal{F} \subseteq \overline{C_{V, 1}}$ is strictly less than $\pi$. (The last fact holds because $C_{V, 1}$ is symmetric and $\mu_{V,1}(C_{V, 1}) >0$.) 

To finish the proof of \Cref{clm:two-d-increment}, we need the following useful claim:

\begin{claim}~\label{clm:angle-include}
Suppose that the angular measure of an arc $\mathcal{F} \subseteq \overline{C_{V, 1}}$ whose end points belong to $C_{V, 1}$ is $0 < t < \pi$. Then the angular measure of the arc $\mathcal{F}_{1-\kappa} \cap K$ is at least $(t \kappa \cdot \cos (t/2)) /2$.\ignore{\rnote{Sorry, right now I am confused:  since $\mathcal{F}_{1-\kappa} := \{z (1-\kappa) : z\in \mathcal{F}\}$, why isn't the angular measure of ${\cal F}_{1-\kappa}$ just the same as the angular measure of ${\cal F}$, i.e. exactly $t$?}}
\end{claim}
\begin{proof}
Without loss of generality assume that the center of the unit circle is $(0,0)$ and that the two endpoints of the arc $\mathcal{F}$ are located at $(\cos 0, \sin 0)$ and $(\cos t, \sin t)$. By definition both endpoints are in the set $C_{V, 1}$ and hence the line segment $L$ joining $(\cos 0, \sin 0)$ and $(\cos t, \sin t)$ is in the convex set $K$. Using this and the fact that the origin lies in $K$, it follows from a simple geometric argument that the angular measure of $\mathcal{F}_{1-\kappa}$ inside $K$ is exactly\ignore{\rnote{I guess it is possible that $\cos(t/2)>1-\kappa$, right --- in this case we are feeding a number bigger than 1 into $\arccos$, which doesn't seem kosher}}
\begin{equation}~\label{eq:lb-ang-measure}
\begin{cases} 
= t &\textrm{if} \ \cos(t/2) > 1-\kappa  \\ 
= 2 \bigg( \frac{t}{2} - \arccos \bigg(\frac{\cos (t/2)}{1-\kappa}\bigg)\bigg) &\textrm{if} \ \cos(t/2) \leq 1-\kappa 
\end{cases}
\end{equation}
If $\cos(t/2) > 1-\kappa$ then by \eqref{eq:lb-ang-measure} we are done, since $t \geq (t \kappa \cdot \cos (t/2)) /2.$ 
If $\cos(t/2) \le 1-\kappa$, then we recall the fact that if $0 \le x,\Delta x$ and $x + \Delta x \le 1$ then
\[
\arccos(x) - \arccos(x+ \Delta x) \ge \Delta x.
\]
Applying this inequality with $x=\cos(t/2)$ and $x + \Delta x = {\frac {\cos(t/2)}{1-\kappa}}$, we get that the angular measure of $\mathcal{F}_{1-\kappa}$ inside $K$ is at least ${\frac {2 \kappa}{1-\kappa}} \cos(t/2)$, which is easily seen to be at least $(t \kappa \cdot \cos (t/2)) /2$ since ${\frac 2 {1-\kappa}} \geq 2 \geq \pi/2 \geq t/2.$
%
%
% \rnote{check this too} 
%\[
%\red{2 \bigg( \frac{t}{2} - \arccos \bigg(\frac{\cos (t/2)}{1-\kappa}\bigg)\bigg) \ge \frac{t \kappa \cos (t/2)}{2}. }
%\]
%This proves \Cref{clm:angle-include}. 
\end{proof} 

Armed with \Cref{clm:angle-include}, we can now prove \Cref{clm:two-d-increment}. In particular, suppose $\overline{C_{V, 1}}$ is a union of disjoint arcs $\{\mathcal{F}^{(i)}\}_{i \in \mathbb{N}}$ of length $\{t_i\}_{i \in \mathbb{N}}
$.\ignore{\footnote{{\color{red} I am assuming that a polytope intersects a circle at countable number of points. Check??}}}. We now make two observations: 
\begin{enumerate}
\item Since $\mu_{V,1}(C_{V, 1})  = p$, the total angular measure of the arcs, $\sum_{i \in \mathbb{N}} t_i$, is $2\pi (1-p)$. 
\item Since $C_{V, 1}$ (and hence its complement) is centrally symmetric, each $t_i$ is at most $\pi(1-p)$. 
\end{enumerate}
For each arc $\mathcal{F}^{(i)}$, by \Cref{clm:angle-include} the angular measure of $\mathcal{F}^{(i)}_{1-\kappa} \cap K$ is at least $\frac{t_i \kappa \cos (t_i/2)}{2}$. This means that the total angular measure of all the arcs $\mathcal{F}^{(i)}_{1-\kappa}$ is at least
\[
\sum_i \frac{t_i \kappa \cos (t_i/2)}{2} \geq 
 \frac{(1-p)  \cdot 2\pi \cdot \kappa}{2} \cdot \cos (\pi(1-p)/2)
 =
 \frac{(1-p)  \cdot 2\pi \cdot \kappa}{2} \cdot \sin (\pi \cdot p/2), 
\]
where the inequality follows from items 1 and 2 above and the fact that the cosine function is monotonically decreasing in the interval $[0,\pi)$.
%This implies that 
%\begin{equation}~\label{eq:lb-4}
%\mu_{V,1}(\overline{C \cap V}) \ge \frac{(1-p)  \cdot  \kappa}{2} \cdot \sin (\pi \cdot p/2). 
%\end{equation}
%This proves \eqref{eq:increment1}.
Translating from the total angular measure of all the arcs $\mathcal{F}^{(i)}_{1-\kappa}$ to $\mu_{V,1-\kappa}(B_{K,V})$ via facts (i) and (ii) from the beginning of the proof, we get \Cref{clm:two-d-increment}.
\end{proof}

To finish the proof of \Cref{lem:key}, we consider two cases: (I) when $\mu_{1}(C_1) \le 1/2$, and (II) when  $\mu_{1}(C_1) > 1/2$.  In case (I) we set $\beta = \mu_1(C_1) $ and in case (II) we set 
$1-\beta = \mu_1(C_1)$, so in both  cases it holds that $\beta \le 1/2$. We now define a two-dimensional subspace $V \subset \R^n$ to be \emph{good} if 
\begin{enumerate}
\item In case (I), $\beta /4 \le \mu_{V,1}(C_{V, 1})  \le 9/10$; 
\item  In case (II), $1/10 \le \mu_{V,1}(C_{V, 1})  \le 1-\beta/4$. 
\end{enumerate}
Note that by \Cref{claim:raz}, in both cases $\Prx_{\bV} [\bV \textrm{ is good}] \ge \beta/2$. We thus have that 
\begin{align}
\mu_{1-\kappa}(C_{1-\kappa}) &= \mathbf{E}_{\bV}[\mu_{\bV, 1-\kappa}(C_{\bV,1-\kappa})] \ \tag{by \eqref{eq:avg-1}}\nonumber \\ 
&= \mathbf{E}_{\bV}[\mu_{\bV, 1-\kappa}(C_{\bV,1-\kappa}) \ | \  \bV \textrm{ is not good}] \cdot \Pr[\bV \textrm{ is not good}]  \nonumber \\ & \ \ \ \ + \mathbf{E}_{\bV}[\mu_{\bV, 1-\kappa}(C_{\bV,1-\kappa}) \ | \  \bV \textrm{ is  good}] \cdot \Pr[\bV \textrm{ is  good}] \nonumber  \\ 
&\geq \mathbf{E}_{\bV}[\mu_{\bV, 1}(C_{\bV,1}) \ | \  \bV \textrm{ is not good}] \cdot \Pr[\bV \textrm{ is not good}]  \nonumber \\ & \ \ \ \ + \mathbf{E}_{\bV}[\mu_{\bV, 1-\kappa}(C_{\bV,1-\kappa}) \  | \ \bV \textrm{ is  good}] \cdot \Pr[\bV \textrm{ is  good}] \ \ \ \ \ \ \ \text{(by \eqref{eq:avg-2})}\label{eq:ineq-Ck-1} 
\end{align}

By applying \eqref{eq:diff-reexp} and \Cref{clm:two-d-increment}, we get that if $V$ is {good}, then 
\[
\mu_{V,1-\kappa}(C_{V, 1-\kappa}) \ge \mu_{V,1}(C_{V, 1}) + \Theta(\kappa \beta). 
\]
Using \Cref{claim:raz}, we have that $\Prx_{\bV} [\bV \textrm{ is {good}}] \ge \beta/2$. Combining these two inequalities with \eqref{eq:ineq-Ck-1}, we get that 
\begin{align*}
\mu_{1-\kappa}(C_{1-\kappa}) &\ge \mathbf{E}_{\bV}[\mu_{\bV, 1}(C_{\bV,1}) \ | \  \bV \textrm{ is not good}] \cdot \Pr[\bV \textrm{ is not good}]   \\ & \ \ \ \ + \mathbf{E}_{\bV}[\mu_{\bV, 1}(C_{\bV,1}) \ | \  \bV \textrm{ is  good}] \cdot \Pr[\bV \textrm{ is  good}] + \Theta (\kappa \cdot \beta^2) \\ &\ge \mathbf{E}_{\bV}[\mu_{\bV, 1} (C_{\bV,1})] + \Theta (\kappa \cdot \beta^2) = \mu_{1}(C_1) + \Theta (\kappa \cdot \beta^2),
\end{align*}
where the last inequality again uses \eqref{eq:avg-1}. The proof of \Cref{lem:key} is complete modulo the proof of \Cref{claim:raz}, which we give below.

\subsubsection{Proof of \Cref{claim:raz}} \label{sec:claim:raz:proof}

Recall that $C \subset \mathbb{S}^{n-1}$ is such that $\mu_1(C)$ (the measure of $C$ as a fraction of $\mathbb{S}^{n-1}$) satisfies $0 < \min\{\mu_1(C) , 1 - \mu_1(C) \} = \beta  \leq 1/2$. For conciseness let $c$ denote $\mu_1(C),$ so $c \in [\beta,1-\beta].$  We follow the general structure of Raz's original argument with some careful modifications.
%with one significant modification which is highlighted below.

Let $\by^{(1)},\by^{(2)}$ be independent uniform random elements of $\mathbb{S}^{n-1}$.  Let $\bY$ be the number of elements of $\{\by^{(1)},\by^{(2)}\}$ that lie in $C$, so $\bY$ is supported in $\{0,1,2\}$. Then we have  
\begin{equation}~\label{eq:Yprob} \Pr[\bY=2]=c^2, \ \ \Pr[\bY=0] = (1-c)^2. 
\end{equation}

Given vectors $u,v \in \R^n$ let $\spann(\{u,v\})$ denote the span of $u$ and $v$. We record a few easy but subtle facts about the distribution of independent uniform random $\by^{(1)},\by^{(2)}$ and their span:
\begin{fact}~\label{fact:paradoxical}
\begin{enumerate}
\item  For $\by^{(1)},\by^{(2)}$ chosen as above, with probability $1$ the vector space $\spann(\{\by^{(1)},\by^{(2)}\})$ is uniform random over all 2-dimensional subspaces of $\mathbb{R}^n$. 
\item For any fixed 2-dimensional subspace $V'$, conditioned on $\spann(\{\by^{(1)},\by^{(2)}\}) = V'$, each of $\by^{(1)},\by^{(2)}$ is uniformly randomly distributed over $V' \cap \mathbb{S}^{n-1}$. 
\end{enumerate}
\end{fact}
\begin{remark}
We remark that conditioned on $\spann(\{\by^{(1)},\by^{(2)}\}) = V'$, the distribution of  $\by^{(1)}$ and $\by^{(2)}$ is no longer independent. An earlier draft of this paper gave an argument that $\by^{(1)}$ and $\by^{(2)}$ are independent conditioned on $\spann(\{\by^{(1)},\by^{(2)}\}) = V'$, but there was a subtle flaw in the argument, which was pointed out to us by Raz in a personal communication~\cite{raz2019pc}, arising from the Borel-Kolmogorov paradox~\cite{Kolmogorov-Borel:wikipedia}.  The purpose of this remark is to highlight the fact that subtle issues can arise when conditioning on measure zero events (such as $\spann(\{\by^{(1)},\by^{(2)}\}) = V'$). 
\end{remark}

We are now ready to prove \Cref{claim:raz}.  We start with the case when $\mu_1(C)  = c  = \beta \leq 1/2$. Later we will consider the case when $\mu_1(C) = c = 1-\beta > 1/2$. 

\medskip

\noindent \textbf{Case I: $\mu_1(C)  = c  = \beta <1/2$.} Our aim is to bound the probabilities $\Prx_{\bV}[\mu_{\bV,1}(C \cap \bV) > 9/10]$ and $\Prx_{\bV}[\mu_{\bV,1}(C \cap \bV) < \beta/4]$. 

First, fix a two-dimensional subspace $V'$ such that $\mu_{V',1}(C \cap V') > 9/10.$   
By the second item in \Cref{fact:paradoxical}, we have that
\[
\Prx_{\red{\by^{(1)},\by^{(2)}}}[\bY=2 \ | \ \spann(\{\by^{(1)},\by^{(2)}\}) = V'] > \frac{8}{10}. 
\]
Since this is true for any such subspace $V'$, we have that 
\[
\Prx_{\red{\by^{(1)},\by^{(2)}}}[\bY=2 \ | \ \mu_{ \red{\spann(\{\by^{(1)},\by^{(2)}\})}    ,1}(C \cap \red{\spann(\{\by^{(1)},\by^{(2)}\}})) > 9/10] > \frac{8}{10}, 
\]
where $\bV$ is a random variable distributed as a uniform random 2-dimensional subspace of $\R^n$.
We thus have \begin{align*}
c^2 &= \Prx_{\red{\by^{(1)},\by^{(2)}}}[\bY=2] \geq
\Prx_{\red{\by^{(1)},\by^{(2)}}}[\bY = 2 \ \& \ \mu_{\red{\spann(\{\by^{(1)},\by^{(2)}\})},1}(C \cap \red{\spann(\{\by^{(1)},\by^{(2)}\})}) > 9/10]\\
&=\Prx_{\red{\by^{(1)},\by^{(2)}}}[\bY = 2 \ | \ \mu_{\red{\spann(\{\by^{(1)},\by^{(2)}\})},1}(C \cap \red{\spann(\{\by^{(1)},\by^{(2)}\})}) > 9/10] \cdot\\
& \ \ \ \ 
\Prx_{\red{\by^{(1)},\by^{(2)}}}[ \mu_{\red{\spann(\{\by^{(1)},\by^{(2)}\})},1}(C \cap \red{\spann(\{\by^{(1)},\by^{(2)}\})}) > 9/10]\\
&>\frac{8}{10} \cdot \Prx_{\red{\by^{(1)},\by^{(2)}}}[ \mu_{\red{\spann(\{\by^{(1)},\by^{(2)}\})},1}(C \cap \red{\spann(\{\by^{(1)},\by^{(2)}\})}) > 9/10],
\end{align*}
which gives
\[
\Prx_{\red{\by^{(1)},\by^{(2)}}}[ \mu_{\red{\spann(\{\by^{(1)},\by^{(2)}\})},1}(C \cap \red{\spann(\{\by^{(1)},\by^{(2)}\})})> 9/10] < {c^2 \cdot \frac{10}{8}}.
\]
\red{Since by the first item of \Cref{fact:paradoxical}, the vector space $\spann(\{\by^{(1)},\by^{(2)}\})$ is with probability 1 a uniform random two-dimensional subspace $\bV$ of $\R^n$, we may restate this last bound as
}
\begin{equation}~\label{eq:bv-lb1}
\Prx_{\bV}[ \mu_{\bV,1}(C \cap \bV)> 9/10] < {c^2 \cdot \frac{10}{8}}.
\end{equation}
An identical argument (now using the event $\bY=0$ rather than the event $\bY=2$) gives 
\begin{equation}~\label{eq:bv-lb2}
\Prx_{\bV}[ \mu_{\bV,1}(C \cap \bV)<\beta/4] < {\frac {(1-c)^2}{1-\beta/2}}.
\end{equation}
Combining (\ref{eq:bv-lb1}) and (\ref{eq:bv-lb2}) (and using $c=\beta$ in this case), we get 
\begin{equation}~\label{eq:bv-lb-f1}
\Prx_{\bV} [\mu_{\bV,1}(C \cap \bV)\in [\beta/4, 9/10]] \ge 1- {\beta^2 \cdot \frac{10}{8}} - {\frac {(1-\beta)^2}{1-\beta/2}}.
\end{equation}

\medskip
\noindent
\textbf{Case II: $\mu_1(C)  = c  = 1-\beta >1/2$.} Applying the above analysis \emph{mutatis mutandis} in this setting, we get 
\begin{equation}~\label{eq:bv-lb-f2}
\Prx_{\bV} [\mu_{\bV,1}(C \cap \bV)\in [1/10, 1-\beta/4]] \ge 1- {\beta^2 \cdot \frac{10}{8}} - {\frac {(1-\beta)^2}{1-\beta/2}}.
\end{equation}

Elementary calculus shows that for all $0 \le \beta \le 1/2$, 
\[
1- {\beta^2 \cdot \frac{10}{8}} - {\frac {(1-\beta)^2}{1-\beta/2}} \ge \frac{\beta}{2}. 
\]
Combining the above inequality with \eqref{eq:bv-lb-f1} (respectively \eqref{eq:bv-lb-f2}) gives Case I (respectively~Case II) of the claim, and the proof of \Cref{claim:raz} is complete. \qed

\subsection{Proof of Theorem~\ref{lem:key-general}} 

The proof of \Cref{lem:key-general} is almost exactly  the same as the proof of \Cref{lem:key} up to the statement of \Cref{clm:two-d-increment} (including exactly the same definitions). The only difference is that now as we rescale to set $r=1$, the guarantee for $K$ is that $B(0^n, r_{\mathsf{inner}}) \subseteq K$ where $r_{\mathsf{inner}} = r_{\mathsf{small}} / r$. Finally, we can assume in the current context that $1-\kappa \ge r_{\mathsf{inner}}$, since otherwise, the conclusion trivially holds.

Instead of \Cref{clm:two-d-increment}, we now have the following claim. 

\begin{claim}~\label{clm:2d-increment-new}
Let $\mu_{V,1}(C_{V, 1}) = p \in (0,1)$. Then, for all $0 \le \kappa \le 1/20$, 
\[
\mu_{V,1-\kappa}(B_{K,V}) \ge \min \bigg\{ {\kappa \cdot r_{\mathsf{inner}}}{}, \frac{(1-p) \cdot \kappa}{4} \bigg\}. 
\]
\end{claim}

\begin{proof}
As in \Cref{clm:two-d-increment}, we will refer to $\mathbb{S}_1^{n-1} \cap V$ as ``the unit circle'' and $\mathbb{S}_{1-\kappa}^{n-1} \cap V$ as ``the circle of radius $1-\kappa$.'' As before, $C_{V, 1}$ is a subset of the unit circle and we partition $\overline{C_{V, 1}}$ into
a collection of disjoint arcs (whose end points belong to $C_{V, 1}$). 
As before, for any such arc  $\mathcal{F}$, define $\mathcal{F}_{1-\kappa} = \{z (1-\kappa) : z\in \mathcal{F}\}$. Now we make two observations: (i) $\mathcal{F}_{1-\kappa} \subseteq B_{K,V}$, (ii) if $\mathcal{F}$ and $\mathcal{G}$ are disjoint arcs, then so are $\mathcal{F}_{1-\kappa}$ and $\mathcal{G}_{1-\kappa}$. 

Note that unlike \Cref{clm:two-d-increment}, now it is possible for a single arc to have measure more than $\pi$. We deal with ``large arcs'' through the following claim (recall that in our setting we have $\mathcal{B}(0^n, r_{\mathsf{inner}}) \subseteq K$):

\begin{claim}~\label{clm:angle-include2}
Suppose the angular measure of an arc $\mathcal{F} \subseteq \overline{C_{V, 1}}$ whose end points belong to $C_{V, 1}$ is $\pi/2 \le t < 2\pi$.  
Then the angular measure of the arc $\mathcal{F}_{1-\kappa} \cap K$ is at least ${\kappa \cdot r_{\mathsf{inner}}}$. 
\end{claim}
\begin{proof}
Without loss of generality we may suppose that one of the two endpoints of the arc $\mathcal{F}$ is $y=(\cos 0, \sin 0)$ and the other is $y' = (\cos t, \sin t)$. Define $\mathcal{F}_{r_{\mathsf{inner}}}$ to be the arc of the radius-$r_{\mathsf{inner}}$ circle corresponding to $\mathcal{F}$, i.e.~$\mathcal{F}_{r_{\mathsf{inner}}} = \{z: z/r_{\mathsf{inner}} \in \mathcal{F}\}$. 
Next, define $z$ to be the point on $\mathcal{F}_{r_{\mathsf{inner}}}$ such that the tangent at $z$ passes through $y$. (Such a point $x$ must exist because the angular measure of the arc $\mathcal{F}$ is at least $\pi/2$.) Recalling that $B(0^n, r_{\mathsf{inner}}) \subseteq K$, we have that the point $z$, the origin, and $y$ all lie in $K$. By elementary trigonometry, it follows that the angular measure of $\mathcal{F}_{1-\kappa}$ inside $K$ is at least 
\begin{eqnarray}
\arccos ({r_{\mathsf{inner}}}) -\arccos \bigg(\frac{r_{\mathsf{inner}}}{1-\kappa}\bigg) &\ge&  \arcsin \bigg(\frac{r_{\mathsf{inner}}}{1-\kappa}\bigg)
-\arcsin ({r_{\mathsf{inner}}}) \nonumber \\ 
&\ge& \bigg(\frac{r_{\mathsf{inner}}}{1-\kappa}\bigg)  - r_{\mathsf{inner}}.  \nonumber 
\end{eqnarray}
The last inequality uses the simple fact that $\arcsin x - \arcsin y \ge x-y$ when $0\le y \le x\le 1$. Using the fact that $\kappa \le 1/10$, we get \Cref{clm:angle-include2}.
\end{proof}

To prove \Cref{clm:2d-increment-new}, we consider two possibilities.  The first is that there is a single arc contained in $\overline{C_{V, 1}}$ whose angular length is at least $\pi/2$; in this case we get \Cref{clm:2d-increment-new} using \Cref{clm:angle-include2}. The other possibility is 
that $\overline{C_{V, 1}}$ is split into arcs $\{\mathcal{F}^{(i)}\}$ of length $\{t_i \}$ where each $t_i \le \pi/2$. Note that the total angular measure of the arcs is $\sum t_i = 2\pi (1-p)$. For each such arc $\mathcal{F}^{(i)}$, by \Cref{clm:angle-include}, we get that the angular measure of $\mathcal{F}^{(i)}_{1-\kappa} \cap K$ is at least $\frac{t_i \kappa \cos (t_i/2)}{2}$. Thus, the total angular measure of the intersection of $K$ with all the arcs $\mathcal{F}^{(i)}_{1-\kappa}$ is 
\[
\sum_i \frac{t_i \kappa \cos(t_i/2)}{2}  \ge \sum_i \frac{t_i \kappa }{4} = \frac{2\pi \kappa (1-p)}{4},
\]
where the inequality holds because each $\cos(t_i/2)$ is at least $1/2.$
Translating from angular measure of $K \cap \cup_i \mathcal{F}^{(i)}_{1-\kappa}$ to $\mu_{V,1-\kappa}(B_{K,V})$, we get the stated claim. 
%If the angular length is at least $\pi/2$, then we apply Claim~\ref{clm:angle-include2}. (ii) If the angular length is at most $\pi/2$, then we apply Claim~\ref{clm:angle-include}  to get that the angular measure of the arc $\mathcal{F}_{1-\kappa}$ is at least $\frac{t\kappa }{20} \cdot \cos (t/2)$ which is at least $\frac{t\kappa }{40}$ (using $t\le \pi/2$). This finishes the proof of the claim. 
%and (iii) The angle of any such arc $\mathcal{F}$ is strictly less than $\pi$. The last fact crucially exploits that $C_{V, 1}$ is symmetric and $\mu_{V,1}(C_{V, 1}) >0$.
\end{proof}

As with the proof of \Cref{lem:key}, we split the analysis into two cases: (i) when $\mu_{1}(C_1) \le 1/2$, and (ii) when  $\mu_{1}(C_1) > 1/2$.  In case (i) we set $\mu_1(C_1) = \beta$ and in case (ii) we set 
$\mu_1(C_1) = 1-\beta$ so that in both cases $\beta \le 1/2$. We now define a two-dimensional subspace $V$ to be \emph{good} if 
\begin{enumerate}
\item In case (i), $\beta /4 \le \mu_{V,1}(C_{V, 1})  \le 9/10$. 
\item  In case (ii), $1/10 \le \mu_{V,1}(C_{V, 1})  \le 1-\beta/4$. 
\end{enumerate}
Note that in both cases $\Prx_{\bV} [\bV \textrm{ is good}] \ge \beta/2$ (using \Cref{claim:raz}). Recall that \eqref{eq:ineq-Ck-1} says that
\begin{eqnarray}
\mu_{1-\kappa}(C_{1-\kappa}) 
&\geq& \mathbf{E}_{\bV}[\mu_{\bV, 1}(C_{\bV,1}) | \bV \textrm{ is not good}] \cdot \Pr[\bV \textrm{ is not good}]  \nonumber \\ &+& \mathbf{E}_{\bV}[\mu_{\bV, 1-\kappa}(C_{\bV,1-\kappa}) | \bV \textrm{ is  good}] \cdot \Pr[\bV \textrm{ is  good}] \  \nonumber
\end{eqnarray} 
When $V$ is good, by applying \eqref{eq:diff-reexp} and \Cref{clm:2d-increment-new}, we get
\begin{eqnarray}
\mu_{V,1-\kappa}(C_{V, 1-\kappa}) &\ge& \mu_{V,1}(C_{V, 1}) + \Theta(\kappa r_{\mathsf{inner}}) \ \ \textrm{in case (i);} \nonumber \\ 
\mu_{V,1-\kappa}(C_{V, 1-\kappa}) &\ge& \mu_{V,1}(C_{V, 1}) + \Theta(\kappa \min\{r_{\mathsf{inner}}, \beta\}) \ \ \textrm{in case (ii).} \nonumber
\end{eqnarray}
Using the fact that $\Prx_{\bV}[\bV  \textrm{ is  good}] \ge \beta/2$ and doing exactly the same calculation as the end of Lemma~\ref{lem:key-general}, 
\begin{eqnarray*}
\mu_{1-\kappa}(C_{1-\kappa}) &\ge& \mu_1(C_1) + \Theta (\kappa r_{\mathsf{inner}} \beta) \ \ \textrm{ in Case 1;} \\ 
\mu_{1-\kappa}(C_{1-\kappa}) &\ge& \mu_1(C_1) + \Theta (\kappa \beta \min\{r_{\mathsf{inner}} , \beta\}) \ \ \textrm{ in Case 2.} 
\end{eqnarray*}
Plugging in $r_{\mathsf{inner}} = r_{\mathsf{small}}/r$, we get \Cref{lem:key-general}. 

\section{Background results on the Gaussian distribution}
In this brief section we give some technical preliminaries for the Gaussian distribution, which will be used in our weak learning and Hermite concentration results.

We endow $\mathbb{R}^n$ with the standard Gaussian measure $N(0,1)^n$ (i.e. each coordinate is independently distributed as a standard normal). We define the Gaussian volume of a region $K \subseteq \R^n$, denoted $\vol(K)$, to be $\Pr_{\bg \sim N(0,1)^n}[K(\bg)=1]$.

We note some basic but crucial properties of the chi-squared distribution with $n$ degrees of freedom. Recall that a non-negative random variable $\br^2$  is distributed according to the chi-squared distribution $\chi^2(n)$ if $\br^2 = \bg_1^2 + \cdots + \bg_n^2$ where $\bg \sim N(0,1)^n,$ and that a draw from the chi distribution $\chi(n)$ is obtained by making a draw from $\chi^2(n)$ and then taking the square root.  We recall the following tail bound:
\begin{lemma} [Tail bound for the chi-squared distribution \cite{Johnstone01}] \label{lem:johnstone}
Let $\br^2 \sim \chi^2(n)$.
Then we have
\[\Prx\big[|\br^2-n| \geq tn\big] \leq e^{-(3/16)nt^2}\quad\text{for all $t \in [0, 1/2)$.}\]
It follows that  for $\br \sim \chi(n)$, 
\[
\Prx \big[ \sqrt{{n}/{2}} \le \br \le \sqrt{{3n}/{2}} \big] \ge 1-  e^{-\frac{3n}{64}}. 
\]
\end{lemma}
The following fact about the anti-concentration of the chi distribution will be useful:
\begin{fact} \label{fact:chi-squared-2}
For $n > 1$, the maximum value of the pdf of the chi distribution $\chi(n)$ is at most $1$, and hence for any interval $I=[a,b]$ we have 
$\Pr_{\br^2 \sim \chi^2(n)}[\br \in [a,b]] \leq b-a.$
\end{fact}

\section{First application of our Kruskal-Katona theorems:  Weak learning convex sets and centrally symmetric convex sets}

\noindent {\bf Intuition.}  Before entering into the detailed analysis of our weak learners we give some basic intuition for why a Kruskal-Katona type statement for convex sets should be useful for obtaining a weak learning result. In particular, below we give an informal explanation of why 
\Cref{lem:key} should be useful for weak learning.

Let $K \subset \R^n$ be an unknown nonempty symmetric convex body.\ignore{with $\vol(K)<1$, so by \Cref{lem:alpha-K-decay} we have $\alpha_K(0)=0$ and $\lim_{r \to \infty} \alpha_K(r)=0.$} For the purpose of this intuitive explanation let us suppose that there is a value $r_{1/2}$ such that $\alpha_K(r_{1/2})=1/2$.\footnote{In general the function $\alpha_K(\cdot)$ need not be continuous, but it can be made continuous by perturbing $K$ by an arbitrarily small amount, so this is essentially without loss of generality.} The high-level idea is that in this case the polynomial threshold function  $f(x) := \sign \left((r_{1/2})^2 - \sum_{i=1}^n x_i^2 \right)$ (i.e. the indicator function of the origin-centered ball of radius $r_{1/2}$) must have some non-negligible correlation with $K$ and can serve as a weak hypothesis.

To justify this claim, we first establish that the advantage of $f$ is at least non-negative. To see this, first observe that 
\[
\Prx_{\bg \sim N(0,1)^n}[K(\bg) = f(\bg)] = \Ex_{\br^2 \sim \chi^2(n)}
\Prx_{\bx \sim \mathbb{S}^{n-1}_{\br}} [K(\bx) = f(\bx)],
\]
and next observe that for each $r>0$, by the choice of $r_{1/2}$ and the definition of $f$, we have that
\[
\Prx_{\bx \sim \mathbb{S}^{n-1}_{r}} [K(\bx) = f(\bx)]
=
\begin{cases}
\alpha_K(r)  & \text{if~}r<r_{1/2}\\
1 -  \alpha_K(r) & \text{if~}r \geq r_{1/2},
\end{cases}
\]
which is at least $1/2$ in each case by ~\Cref{fact:convex-decreasing}.

Extending this simple reasoning, it is easy to see that if we have
\begin{equation} \label{eq:advantage}
\Prx_{\br^2 \sim \chi^2(n)} [\overbrace{|\Prx_{\bx \sim \mathbb{S}^{n-1}_{\br}} [K(\bx)=1]}^{=\alpha_K(\br)} - 1/2| \ge \beta] \ge \gamma,
\end{equation}
for some $\beta,\gamma>0$, 
then $f$ is a weak hypothesis for $K$ with advantage $\Omega(\gamma \beta)$.  Putting it another way, the only way that $f$ could fail to be a weak hypothesis with non-negligible advantage would be if the function $\alpha_K(\cdot)$ ``stayed very close to $1/2$'' for a ``wide range of values around $r_{1/2}$'' --- but this sort of behavior of $\alpha_K(\cdot)$ is precisely what is ruled out by our density increment result, \Cref{lem:key}.

\subsection{A weak learner for centrally symmetric convex sets}
\label{sec:wl-given-kgl}
In this subsection we prove \Cref{thm:weak-learn-centrally-symmetric}, which gives a weak learner for centrally symmetric convex sets. In the next subsection we will prove \Cref{thm:weak-learn-convex}, which gives a weak learner for general convex sets. As a major technical ingredient in proving \Cref{thm:weak-learn-convex} is a variant of \Cref{thm:weak-learn-centrally-symmetric}, we will explicitly note the places in the current proof where we use the central symmetry of $K$. 

Recall from \Cref{sec:intro} that $r_{\median}$ is the median value of $\chi(n)$.  Let us define the function $r:[0,1) \rightarrow [0,\infty)$ by 
\[
\Pr_{\br \sim \chi(n)} [\br \le r(c)] =c. 
\]
Observe that since the pdf of $\chi^2(n)$ is always positive, the function $r(c)$ is well-defined. Also, with this notation, we have that $r(1/2) = r_{\median}$. 
\Cref{lem:johnstone} and \Cref{fact:chi-squared-2} together easily yield the following claim:
\begin{claim}~\label{clm:chi-percentile}
The median $r_{\median}$ of the $\chi(n)$ distribution satisfies $|r_{\median} - \sqrt{n}| = O(1)$.\footnote{In fact it is known that $r_\median \approx \sqrt{n} \cdot (1 - {\frac 2 {9n}})^{3/2}$, though we will not need this more precise bound.} Further, there exist positive constants $A, \ B \ge 1/4$\ignore{\footnote{@Rocco: Earlier, $A$, $B$ were posited to be in the interval $[1/4, 1/2]$. I don't see why the upper bound is true or for that matter, useful}} such that $r(1/4)  = r_{\median} - A$ and $r(3/4) = r_{\median} + B$. 
\end{claim}

Now we are ready to embark on the proof of \Cref{thm:weak-learn-centrally-symmetric}.
As noted earlier, the high level structure of the proof is similar to the argument used in \cite{BBL:98} to show that one of the three functions $+1$, $-1$ or Majority is a good weak hypothesis for any monotone Boolean function.

\paragraph{Proof of \Cref{thm:weak-learn-centrally-symmetric}.} 

\Cref{thm:weak-learn-centrally-symmetric} is an immediate consequence of \Cref{lem:dictator-learning} and \Cref{lem:sphere-learning} which are stated below. Before presenting these lemmas, let us define the following three functions:
\begin{equation}~\label{eq:def-h12} 
h_{1/2}(x) := \sign (r_{\median}^2 -(x_1^2 +\ldots + x_n^2 )), \quad h_{0}(x):= -1, \quad \textrm{ and } \quad
 h_{1}(x) := 1. 
\end{equation}
Note that one can interpret $h_0(x)$ (respectively $h_1(x)$) as the indicator function of the ball of radius $0$ (respectively $\infty$). We set $c:= 1/40$ for Lemmas \ref{lem:dictator-learning} and \ref{lem:sphere-learning} (the precise value is not important as long as it is positive and sufficiently small). Finally, recall that $\vol(K) = \Pr_{\bg \sim \Nn} [K(\bg)=1]$. 
\begin{lemma}~\label{lem:dictator-learning}
If $|\vol(K) -1/2| > c \cdot n^{-1/2}$ for $c$ defined above, then either $h=h_0$ or $h=h_1$ achieves 
\[
\Prx_{\bg \sim \Nn}[h(\bg) = K(\bg)] \ge \frac{1}{2} + \Theta(n^{-1/2}). 
\]
\end{lemma}
\begin{lemma}~\label{lem:sphere-learning}
If $|\vol(K) -1/2| \le c \cdot n^{-1/2}$ for $c$ defined above, then 
\[
\Prx_{\bg \sim \Nn}[h_{1/2}(\bg) = K(\bg)] \ge \frac{1}{2} + \Theta(n^{-1/2}). 
\]
\end{lemma}

\Cref{lem:dictator-learning} is immediate, so it remains to prove \Cref{lem:sphere-learning}. 

\begin{proofof}{\Cref{lem:sphere-learning}}
We begin by defining the function $\beta: [0,1) \rightarrow [0,1)$ as
\[
\beta(c) := \Prx_{\bx \in \mathbb{S}^{n-1}_{r(c)}}[\bx  \in K] = 
\alpha_K(r(c)).
\]
\begin{fact}~\label{fact:centre-decreasing}
If $K$ is a convex body that contains the origin, then
$\beta (\cdot)$ is a non-increasing function. 
\end{fact}
\begin{proof}
This holds since $r(\cdot)$ is strictly increasing and the function $\alpha_K(\cdot)$ is non-increasing when $0^n \in K$ (\Cref{fact:convex-decreasing}). 
\end{proof}

Next, we prove the following claim (in the current section we will only need it for the case in which the function $p$ is identically $1$, but later we will need the more general version): 
\begin{claim}~\label{clm:Hermite-deg2}
Let $p: [0, \infty) \rightarrow \mathbb{R}$ and extend it to $\overline{p}: \mathbb{R}^n \rightarrow \mathbb{R}$ by defining $\overline{p}(x) = p(\Vert x \Vert_2)$.  Let $\Gamma: \mathbb{R}^n \rightarrow \mathbb{R}$ and define $\beta_\Gamma: [0,1) \rightarrow \mathbb{R}$ as
\[
\beta_\Gamma(\nu) := \Ex_{\bx \sim \mathbb{S}^{n-1}_{r(\nu)}} [\Gamma(\bx)]. 
\]
Then
\[
\Ex_{\bg \sim \Nn} [ \Gamma(\bg) \ovp(\bg)]  = \int_{\nu=0}^1 \beta_{\Gamma}(\nu) p(r(\nu)) d\nu. 
\] 
\end{claim}
\begin{proof}
Let $\chi(n,r)$ denote the pdf of the $\chi$-distribution with $n$ degrees of freedom at $r$. Then, 
\begin{align*}
\Ex_{\bg \sim \Nn}[\Gamma(\bg) \cdot \overline{p}(\bg)] &=  \int_{r=0}^{\infty} \chi(n,r) \bigg( \Ex_{x \sim \mathbb{S}^{n-1}_r}[ \Gamma(x) \overline{p}(x) ] \bigg) dr\\ 
 &=  \int_{r=0}^{\infty} \chi(n,r) p(r) \bigg(\Ex_{x \sim \mathbb{S}^{n-1}_r}[ \Gamma(x)]\bigg) dr. 
\end{align*}
Substituting $r$ by $r(\nu)$ (as $\nu$ ranges from $0$ to $1$), we have 
\begin{equation}\label{eq:beta-correlation}
\Ex_{\bg \sim \Nn}[\Gamma(\bg) \cdot \overline{p}(\bg)] = \int_{\nu=0}^1 \chi(n,r(\nu)) r'(\nu) p(r(\nu)) \beta_\Gamma(\nu) d\nu. 
\end{equation}
Finally, by definition of $r(\nu)$, we have that
\[
\int_{z=0}^{r(\nu)} \chi(n,z) dz= \nu. 
\] 
Taking derivative of this with respect to $\nu$, we get that 
$\chi(n,r(\nu)) r'(\nu) =1$, and substituting this back into \eqref{eq:beta-correlation}, we get the claim. 
\end{proof}
By instantiating Claim~\ref{clm:Hermite-deg2} with $p=1$ and $\Gamma(x)=\mathbf{1}_{x \in K}$, we have the following corollary. 
\begin{corollary}~\label{claim:area}
$\int_{x \in [0,1)} \beta(x) dx = \vol(K)$.
\end{corollary}

%\begin{proof}
%For $a \geq 0$, define $\chi(n,a)$ to be the value of the pdf of the $\chi$-distribution with $n$-degrees of freedom at $a$. Then 
%\[
%\int_{a=0}^{\infty} \alpha_K(a) \chi(n,a) da = \vol(K).
%\]
%Substituting $a = r(x)$ for $x \in [0,1)$, the above equation translates to 
%\begin{equation}~\label{eq:substitute-density} 
%\int_{x =0}^1 \alpha_K(r(x)) \chi(n,r(x)) r'(x) dx=\vol(K). 
%\end{equation}
%Now observe that for any $x \in [0,1)$ we have $\alpha_K(r(x)) = \beta(x)$ and 
%$\int_{z=0}^{r(x)} \chi(n,z) dz = x$. Taking the derivative of the latter, we get that $r'(x) \chi(n,r(x))  = 1$. Plugging this back into (\ref{eq:substitute-density}), we get the claim. 
%\end{proof}
%Let us now define $r_{\median}$  as $r_{\median} = r^{-1}(1/2)$ and the classifier $h_{1/2} : \mathbb{R}^n \rightarrow \{-1,1\}$  as 
%\[
%h_{1/2}(x) = \mathsf{sign} (r_{\median}^2 -(x_1^2 +\ldots + x_n^2 )).  
%\]

Now we are ready to analyze $h_{1/2}$.  The following claim says that if $\beta(1/4)$ is ``somewhat large'', then $h_{1/2}$ is a weak hypothesis with constant advantage:
\begin{claim}~\label{clm:beta-large}
If $\beta(1/4) \ge \frac{3}{4}$ then $\Pr_{\bg \sim N(0,1)^n} [h_{1/2}(\bg) = K(\bg) ] \ge \frac{1}{2} + \frac{1}{24}$. 
\end{claim}
\begin{proof}
Define $s= \int_{x=0}^{1/4} \beta(x) $
and $t = \int_{x=1/4}^{1/2} \beta(x) dx $. Using the fact that $\beta(\cdot)$ is non-increasing 
we have 
\begin{eqnarray}
\textrm{(i)} \ \ s = \int_{x=0}^{1/4} \beta(x)dx  \ge \frac{3}{4} \cdot \frac{1}{4} = \frac{3}{16}, \  \ \ \textrm{(ii)} \ \ t= \int_{x=1/4}^{1/2} \beta(x)dx \ge \frac{1}{3} \cdot \bigg( \int_{x=1/4}^1 \beta(x) dx\bigg) = \frac{\vol(K) -s}{3} ~\label{eq:two-items-bound}
\end{eqnarray}
(where Corollary~\ref{claim:area} was used for the last inequality of (ii)).
We thus get
\begin{eqnarray}
\int_{x=0}^{1/2} \beta(x) - \int_{x=1/2}^1 \beta(x) =   2s + 2t - \vol(K) \ge \frac{4s}{3} - \frac{\vol(K)}{3} \ge \frac{1}{24}, \label{eq:st-bound}
\end{eqnarray}
where the first inequality above follows by item (ii) of \eqref{eq:two-items-bound} and the second inequality uses item (i) of of \eqref{eq:two-items-bound} along with the hypothesis $|\vol(K) -1/2| \le c/\sqrt{n} \le 1/40$. Combining these bounds, we have
\begin{align*}
 \Prx_{\bg \in N(0,1)^n} [h_{1/2}(\bg) = K(\bg)] &=  \int_{x=0}^{1/2} \beta(x) + \int_{x =1/2}^{1} (1-\beta(x)) \\ &\ge \frac{1}{2} + \int_{x=0}^{1/2} \beta(x) - \int_{x=1/2}^{1} \beta(x)  \geq \frac{1}{2} + \frac{1}{24}, \tag{by \eqref{eq:st-bound}}
\end{align*}
and \Cref{clm:beta-large} is proved.
\end{proof}

Thus to prove \Cref{lem:sphere-learning}, it remains to consider the case that $\beta(1/4) \le 3/4$.  By the monotonicity of $\beta(\cdot)$, Corollary~\ref{claim:area}, and the hypothesis of \Cref{lem:sphere-learning}, we have that
\[
\frac{1}{2} - \frac{1}{40} \le \int_{x=0}^{1} \beta(x) dx \le  \frac{1}{4} + \frac{3}{4} \cdot \beta (1/4). 
\]
and hence $\beta(1/4) \ge 3/10$, so we subsequently assume that $3/10 \le \beta(1/4) \le 3/4$. Now, recall that 
\[
r(1/4) = r_{\median} - A \ \textrm{and} \ r(3/4) = r_{\median} + B, 
\]
where $A, B \ge 1/4$ and $r_{\median} = \sqrt{n} \pm O(1)$.
Thus
\begin{equation}~\label{eq:rel-alpha}
\beta(1/4) = \alpha_K(r_{\median}- A) \quad \quad \text{and} \quad \quad \beta(3/4) = \alpha_K(r_{\median}+B) \quad \quad \text{where~}A,B \geq 1/4.
\end{equation}
Using the fact that $3/10 \le \beta(1/4) \le 3/4$, \Cref{eq:rel-alpha}, and \Cref{lem:key}, it follows that 
\begin{equation}~\label{eq:beta-dec-symmetric}
\beta(1/4) \ge \beta(3/4) \red{+} C \cdot n^{-1/2}, 
\end{equation}
for some absolute constant $C>0$. This implies that 
\begin{eqnarray}
&& \int_{x=0}^{1/2} \beta(x) dx - \int_{x=1/2}^1 \beta(x) dx \nonumber \\ 
&=& \int_{x=0}^{1/4} \beta(x) dx - \int_{x=3/4}^1 \beta (x) dx +  \int_{x=1/4}^{1/2} \beta(x) dx - \int_{x=1/2}^{3/4} \beta (x) dx \nonumber \\
&\ge& \frac{C}{4 \sqrt{n}} +  \int_{x=1/4}^{1/2} \beta(x) dx - \int_{x=1/2}^{3/4} \beta (x) dx \ge\frac{C}{4 \sqrt{n}},  \label{eq:beta-dec2}
\end{eqnarray}
where the penultimate inequality uses \Cref{eq:beta-dec-symmetric} and the last inequality uses the monotonicity of $\beta(\cdot)$. 
%Thus, if $\beta(1/4) \ge 3/4$, $h_{1/2}$ is a weak-learner for $K$ with $\Omega(1)$ advantage over $1/2$. Thus, from now on, we can assume that $\beta(1/4) \le 3/4$. On the other hand, a simple Markov argument also shows that $\beta(1/4) \ge 7/30$. Thus, from now on, we can assume that $3/4 \ge \beta(1/4) \ge 7/30$. 
%Next, by applying Fact~\ref{fact:chi-squared-2}, we have that {\color{red} Just list the fact about the $\chi^2$ distribution that we need}
%\[
%r(1/4) = r_{\median} - A  \ \  \textrm{and} \ \ r(3/4) = r_{\median} +B,
%\]
%where $A$ and $B$ are positive constants lying in $[1/4,1/2]$ and $r_{\median} = \sqrt{n} \pm O(1)$. Next, observe that
%\[
%\beta(1/4) = \alpha(r(1/4)) = \alpha (r^\ast -A), \ \ %\beta(3/4) = \alpha(r(3/4)) = \alpha (r^\ast +B). 
%\]
%Now, by applying  Lemma~\ref{lem:key} 
%\begin{equation}~\label{eq:beta-dec1}
%\beta(1/4) \ge \beta(3/4) - C n^{-1/2}, 
%\end{equation}
%where $C$ is an absolute positive constant. 
%This implies that 
%\begin{eqnarray}
%&& \int_{x \in [0,1/2]} \beta(x) dx - \int_{x \in [1/2,1]} \beta(x) dx \nonumber \\ 
%&=& \int_{x \in [0,1/4]} \beta(x) dx - \int_{x \in [3/4,1]} \beta (x) dx +  \int_{x \in [1/4,1/2]} \beta(x) dx - \int_{x \in [1/2,3/4]} \beta (x) dx \nonumber \\
%&\ge& \frac{C}{4 \sqrt{n}} +  \int_{x \in [1/4,1/2]} \beta(x) dx - \int_{x \in [1/2,3/4]} \beta (x) dx \ge\frac{C}{4 \sqrt{n}}.  ~\label{eq:beta-dec2}
%\end{eqnarray}
Applying this, we get 
\begin{eqnarray}
\Prx_{\bg \in N(0,1)^n} [h_{1/2}(\bg)  = K(\bg) ] &=& \int_{x=0}^{1/2} \beta(x) dx + \int_{x=1/2}^1 (1-\beta(x)) dx \nonumber \\ 
&=& \frac12 + \int_{x=0}^{1/2} \beta(x) dx - \int_{x=1/2}^{1} \beta(x) dx \nonumber \\
&\geq& \frac12 + \frac{C}{4 \sqrt{n}}. \label{eq:advantage-symmetric}
\end{eqnarray}
%where the inequality is by \Cref{eq:beta-dec2}.
This finishes the proof of \Cref{lem:sphere-learning} and hence also the proof of
\Cref{thm:weak-learn-centrally-symmetric}.
\end{proofof}

%!TEX root = draft1b.tex

\subsection{A weak learner for general convex sets} \label{sec:weak-learner-general-convex}

In this section we prove \Cref{thm:weak-learn-convex}. 
%As in \Cref{sec:wl-given-kgl}, recall that $\vol(K) = \Pr_{\bg \sim \Nn} [K(\bg)=1]$. 
The proof uses the fact that there are efficient ``weak agnostic'' learning algorithms for halfspaces under the Gaussian distribution.  Several papers in the literature, including \cite{KKMS:08, DDFS14, ABL13, DKS18-nasty}, can be straightforwardly shown to yield a result which suffices for our purposes.  For concreteness we will use the following:

\begin{theorem} [Theorem~1.2 from \cite{DKS18-nasty}, taking ``$d=1$''] \label{thm:DKS18}
There is an algorithm \textsf{Learn-halfspace} with the following guarantee: Let $f: \R^n \rightarrow \bits$ be a target halfspace such that the 
 algorithm gets access to samples of the form $(\bg, h(\bg))$ where $\bg \sim \Nn$  and $h: \R^n \to \bits$ satisfies $\Pr_{\bg} [h(\bg) \not = f(\bg)] \le \epsilon$. Then \textsf{Learn-halfspace} runs in time $\mathsf{poly}(n,1/\epsilon)$ and outputs a hypothesis halfspace $f': \mathbb{R}^n \rightarrow \bits$ such that $\Pr_{\bg \sim \Nn} [f(\bg) \not = f'(\bg)] \le \eps^c$, where $c>0$ is an absolute constant.
\end{theorem}

We note that the type of noise in the above theorem statement is referred to in the literature as \emph{adversarial label noise} (see~\cite{KSS:94}); while we will not need this stronger guarantee, the algorithm of \cite{DKS18-nasty} in fact also works in the stronger \emph{nasty noise} model. An immediate corollary of \Cref{thm:DKS18} is the following. 
\begin{corollary}~\label{corr:DKS}
There is a positive constant $\zeta>0$ such that the algorithm \textsf{Learn-halfspace} has the following guarantee: let $f: \mathbb{R}^n \rightarrow \bits$ be a target halfspace such that the 
 algorithm gets access to samples of the form $(\bg, h(\bg))$ where $\bg \sim \Nn$  and $h: \R^n \to \bits$ satisfies $\Pr_{\bg} [h(\bg) \not = f(\bg)] \le \zeta$. Then  \textsf{Learn-halfspace} runs in $\poly(n)$ time and outputs a halfspace $f': \mathbb{R}^n \rightarrow \bits$ such that $\Pr_{\bg \sim \Nn} [f(\bg) \not = f'(\bg)] \le 1/16$. 
%\end{theorem}
\end{corollary}

For this section we set $c=\min\{1/40,\zeta/8\}$ where $\zeta$ is  the positive constant from \Cref{corr:DKS}. We also recall the definition of the functions $h_0(\cdot)$, $h_1(\cdot)$ and $h_{1/2}(\cdot)$  from \eqref{eq:def-h12}.  
Next, we recall \Cref{lem:dictator-learning} from 
\Cref{sec:wl-given-kgl}
(which we state below for convenience): 
%crucially need the following agnostic learning result for halfspaces from \cite{KKMS:08}. 
%
%
%\begin{theorem}~\label{thm:KKMS-learning}
%There is a polynomial time algorithm \textsf{Halfspace-learner} with the following guarantee: There exists a positive constant $\xi>0$ such that if $\mathcal{D}$ is a distribution supported on $\mathbb{R}^n \times \bits$ such that (i) the algorithm gets access to 
%\end{theorem}
 %and set $c=1/40$. 
\begin{replemma}{lem:dictator-learning}
If $|\vol(K) -1/2| > c \cdot n^{-1/2}$, then either $h=h_0$ or $h=h_1$ achieves 
\[
\Pr_{\bg \sim \Nn}[h(\bg) = K(\bg)] \ge \frac{1}{2} + \Theta(n^{-1/2}). 
\]
\end{replemma}
The next lemma (which is a key technical ingredient) gives a weak learner  for convex sets if there is a point outside $K$ which is close to the origin:
%For this lemma, we set $\zeta>0$ to a sufficiently small but positive constant. 
\begin{lemma}~\label{lem:halfspace-learning}
Let $K$ be a convex body such that $|\vol(K)-1/2| \le 
c \cdot n^{-1/2}$ and for which there is a point $z^\ast \not \in K$ such that $\Vert z^\ast \Vert_2 \leq c$. Then the output of the algorithm \textsf{Learn-halfspace}, when given samples of the form $(\bg, K(\bg))$, is a halfspace $f': \mathbb{R}^n \rightarrow \bits$ such that 
\[
\Prx_{\bg \sim \Nn} [f'(\bg) = K(\bg)] \ge \frac78.
\]
\end{lemma}

\begin{proof}
Using the supporting hyperplane theorem (see page~510 in \cite{luenberger1984linear}), there exists a unit vector $\ell \in \mathbb{R}^n$ and a threshold $\theta \in \mathbb{R}$ such that 
\begin{enumerate}
\item $K \subseteq \{ x:  \ell \cdot x  -  \theta > 0\}$, but 
\item $ \ell \cdot z^\ast  - \theta \leq0$. 
\end{enumerate}
Thus we get a halfspace $r (x) = \sign(\ell \cdot x - \theta)$ such that $r(z^\ast) =-1$ and $K \subseteq r^{-1}(1)$. Next, using the fact that $z^\ast$ lies on the hyperplane $\{x:  \ell \cdot x  -  \theta=0\}$, we get that $|\theta| \le c$. Hence (using the fact that the pdf of an $N(0,1)$ Gaussian is everywhere at most 1) we get that 
\begin{eqnarray*}
\Prx_{\bg \sim \Nn}[r(\bg) = 1] \le \frac12  + c. 
\end{eqnarray*}
This implies that 
\begin{eqnarray}~\label{eq:false-neg}
\Prx_{\bg \sim \Nn}[K(\bg) =1| r(\bg)=1] \ge \frac{\frac12 - \frac{c}{\sqrt{n}}}{\frac12 +c} \ge 1-2c -\frac{2c}{\sqrt{n}}. 
\end{eqnarray}
On the other hand, by construction of $r(\cdot)$, we have that $\Prx_{\bg \sim \Nn}[K(\bg) =-1| r(\bg)=-1]=1$.  Combining this with \eqref{eq:false-neg}, we get that 
\begin{eqnarray}\label{eq:false-pos}
\Prx_{\bg \sim \Nn}[K(\bg) = r(\bg)] \ge 1-2c -\frac{2c}{\sqrt{n}} \ge 1-4c. 
\end{eqnarray}
Recalling that $4c \le \zeta/2$, if we run the algorithm  \textsf{Learn-halfspace} on samples of the form $(\bg, K(\bg))$, then by \Cref{corr:DKS} the output $f'$ satisfies 
\[
\Prx_{\bg \sim \Nn}[r(\bg) = f'(\bg)] \ge \frac{15}{16}. 
\]
Combining this with \eqref{eq:false-pos}, we get 
\[
\Prx_{\bg \sim \Nn}[K(\bg) = f'(\bg)] \ge 1-4c- \frac{1}{16} > \frac78
\]
and the proof of \Cref{lem:halfspace-learning} is complete.
\end{proof}
The last lemma we need for this section is a variant of \Cref{lem:sphere-learning} from \Cref{sec:wl-given-kgl}. 
\begin{lemma}~\label{lem:sphere-learning-asymmetric}
If $K \subseteq \mathbb{R}^n$ is a convex body such that (i) $|\vol(K) -1/2| \le c \cdot n^{-1/2}$\ignore{for $c$ defined above} and (ii) $\mathcal{B}(0^n, c) \subseteq K$, then 
\[
\Prx_{\bg \sim \Nn}[h_{1/2}(\bg) = K(\bg)] \ge \frac{1}{2} + \Theta(n^{-1}). 
\]
\end{lemma}
\begin{proof}
The proof is essentially the same as the proof of \Cref{lem:sphere-learning}. In particular, the proof including all the notation is the same up to and including \Cref{eq:rel-alpha}. The first and only point of departure from \Cref{lem:sphere-learning} is \Cref{eq:beta-dec-symmetric}. Using the fact that $3/4 \ge \beta(1/4) \ge 3/10$, \eqref{eq:rel-alpha} and applying \Cref{lem:key-general} (instead of \Cref{lem:key}), where the ``$\kappa$'' and ``$ \frac{r_{\mathsf{small}}}{r}$'' of  \Cref{lem:key-general} are both $\Theta(n^{-1/2})$, it follows that 
\begin{equation}~\label{eq:beta-dec-symmetric-1}
\beta(1/4) \ge \beta(3/4) - C \cdot n^{-1}
\end{equation}
for some absolute constant $C>0$. Note that \Cref{eq:beta-dec-symmetric-1} is the analogue of \Cref{eq:beta-dec-symmetric} in the current setting. Finally substituting \Cref{eq:beta-dec-symmetric-1} for \Cref{eq:beta-dec-symmetric} and otherwise doing the same calculation as the one leading up to \eqref{eq:advantage-symmetric}, we get that
\[
\Prx_{\bg \in N(0,1)^n} [h_{1/2}(\bg)  = K(\bg) ] \geq  \frac{1}{2} + \frac{C}{4n}. 
\]
This finishes the proof of \Cref{lem:sphere-learning-asymmetric}.
\end{proof}
Theorem~\ref{thm:weak-learn-convex} is now a straightforward combination of \Cref{lem:dictator-learning}, \Cref{lem:halfspace-learning} and \Cref{lem:sphere-learning-asymmetric}.

%!TEX root = draft1b.tex

\section{Hermite mass at low weight levels for convex sets}~\label{sec:hermite-convex}

In this section we prove lower bounds on the Hermite weight of convex sets at levels $0$, $1$ and $2$. As mentioned in the introduction, this immediately implies a  lower bound on the noise stability of convex sets at large noise rates. We begin by proving \Cref{thm:centrally-symmetric-weight}, which gives a lower bound on the Hermite weight of centrally symmetric convex sets up to level $2$. \Cref{thm:centrally-asymmetric-weight} extends this to general convex sets, though the bound is quadratically worse than for centrally symmetric convex sets. Finally, \Cref{sec:symmetric-fourier-weight-tightness} shows that \Cref{thm:centrally-symmetric-weight} is tight up to constant factors. 
%This bound is tight up to constant factors as exhibited in Section.  Finally, Theorem~\ref{thm:centrally-asymmetric-weight} gives a lower bound on the Hermite weight of general convex sets up to level $2$ -- however, this is quadratically worse than our bound for symmetric convex sets. 

Throughout the section we can and do assume that $n$ is at least some sufficiently large absolute constant.

\subsection{Hermite mass at low weight levels for centrally symmetric convex sets}
\label{sec:hermite-centrally-symmetric}
In this subsection we prove \Cref{thm:centrally-symmetric-weight}. This will follow as an immediate consequence of the following two lemmas,
by instantiating \Cref{lem:constant-bias-Hermite-weight} with $\delta = \frac{c}{\sqrt{n}}$ where $c$ is the constant appearing in \Cref{lem:unbias-Hermite-weight}:
\begin{lemma}~\label{lem:constant-bias-Hermite-weight}
Let $\delta>0$ and $K \subseteq \mathbb{R}^n$.\ignore{Define $\vol(K) = \Pr_{\bg \sim \Nn}[\bg \in K]$.} 
If $|\vol(K) - 1/2| \ge\delta$, then $\mathsf{W}^{= 0}[K] \ge 4 \delta^2$. 
%If $|\vol(K) - 1/2| \ge \frac{c}{\sqrt{n}}$, $\mathsf{W}^{\le 2}[K] \ge \frac{c^2}{n}$.  
%If $|\vol(K) - 1/2| \ge \frac{c}{\sqrt{n \log n}}$, $\mathsf{W}^{\le 2}[K] \ge \frac{c^2}{n \log n}$.  
\end{lemma}
\begin{lemma}~\label{lem:unbias-Hermite-weight}
There exists an absolute constant $c>0$ such that the following holds: Let $K \subseteq \mathbb{R}^n$ be a centrally symmetric convex set. If $|\vol(K) - 1/2| \le \frac{c}{\sqrt{n}}$, then $\mathsf{W}^{=2}[K] = \Omega(1/n)$. 
\end{lemma}

\Cref{lem:constant-bias-Hermite-weight} follows immediately from the fact that $\mathsf{W}^{= 0}[K] = 4 \left| \vol(K) - 1/2 \right|^2$. In the rest of this subsection we prove \Cref{lem:unbias-Hermite-weight}.
To do this, we recall the definitions of the functions $r(\cdot)$ and $\beta(\cdot)$ from \Cref{sec:wl-given-kgl}. Namely, for $r: [0,1) \rightarrow [0, \infty)$,  $r(\nu)$ is defined as 
\[
\Prx_{\br  \sim \chi(n)} [\br \leq r(\nu)] = \nu,
\]
and $\beta: [0,1) \rightarrow [0,1]$ is defined as 
\[
\beta(\nu) = \Prx_{x \in \mathbb{S}^{n-1}_{r(\nu)}} [x \in K]. 
\]
Let us now define the function $\overline{\beta} : [0,1) \rightarrow [-1,1]$ as $\overline{\beta}(\nu) = \beta(\nu) -\vol(K)$. By performing an arbitrarily small perturbation of $K$, we may assume that $\beta(\cdot)$ is a continuous function; since $\overline{\beta}$ is non-increasing, there exists a value $\nu_\ast \in [0,1)$ such that  $\overline{\beta}(\nu_\ast)=0$. 
%\[
%\overline{\beta}(\nu_\ast) = \overline{\rho}_K \ \textrm{and} \ \textrm{for all } \nu\ge \nu_\ast, \ \ \overline{\beta}(\nu) \le  \overline{\rho}_K. 
%\]
Let us define $r_{\ast} := r(\nu_\ast)$. Now we define $p_{r_\ast}: [0,\infty) \rightarrow \mathbb{R}$ and $\ovp_{r_\ast}: \mathbb{R}^n \rightarrow \mathbb{R}$ as 
\[
p_{r_\ast}(r) := r_\ast^2 -r^2 \quad \quad \text{and} \quad \quad \ovp_{r_\ast}(x) := p_{r_\ast}(\Vert x \Vert_2). 
\]
The next claim  uses our Kruskal-Katona theorem for centrally symmetric convex sets (\Cref{lem:key}) to prove upper and lower bounds on $r_\ast$: 

\begin{claim}~\label{clm:r-ast-bound1}
 There exists an absolute constant $c>0$ such that the following holds: Let $K$ be a symmetric convex set such that $|\vol(K) - 1/2| \le c$. Then, 
 \[
\frac{n}{4} \le r_\ast^2 \le  4n. 
 \]
\end{claim}

\begin{proof}
We will show $r_\ast \le 2\sqrt{n}$. The other direction is similar. Towards a contradiction, 
assume that $r_\ast > 2\sqrt{n}$. 
Define $r_{\ast,\out} := \sqrt{2n}$. Since $r_{\ast,\out} < r_\ast/\sqrt{2}$, by applying \Cref{lem:key} we get that
\begin{equation}\label{eq:inc-r-ast-r}
\alpha_K(r_{\ast,\out}) \ge \alpha_K(r_\ast) + \kappa = \vol(K) + \kappa, 
\end{equation}
for an absolute constant $\kappa>0$. 
Next, we have \begin{eqnarray}
\vol(K) = \int_{r=0}^{\infty} \chi(n,r) \alpha_K(r) dr \ge \int_{r=0}^{r_{\ast,\out}} \chi(n,r) \alpha_K(r) dr  \ge \alpha_K(r_{\ast,\out}) \cdot \int_{r=0}^{r_{\ast,\out}} \chi(n,r) dr, \label{eq:bias-calc1}
\end{eqnarray}
where the last inequality uses that $\alpha_K(\cdot)$ is non-increasing. By applying  \Cref{lem:johnstone}, we have 

\begin{equation}~\label{eq:kappa-gap2}
\int_{r=0}^{r_{\ast,\out}} \chi(n,r) dr = 1 -\Prx_{\bg \sim \Nn}[\Vert \bg \Vert_2 \ge r_{\ast, \out}]  \ge 1- e^{-\frac{3n}{64}}. 
\end{equation}
Since as stated earlier we may assume that $n \ge \frac{64}{3} \ln (4/\kappa)$, the right hand side is at least $1-\frac{\kappa}{4}$. Plugging this back into \eqref{eq:bias-calc1} and applying \Cref{eq:inc-r-ast-r}, we get 
\begin{eqnarray}~\label{eq:kappa-gap1}
\vol(K) \ge \alpha_K(r_{\ast,\out}) \cdot (1-\kappa/4) \ge (\vol(K) + \kappa) \cdot (1-\kappa/4) \ge \vol(K) + \frac{\kappa}{2}. 
\end{eqnarray}
This contradiction implies that we must have $r_{\ast} \le 2\sqrt{n}$. As stated earlier, the proof of the other direction is similar. 
\end{proof}

The main ingredient in the proof of \Cref{lem:unbias-Hermite-weight} is the following:

\begin{claim}~\label{lem:integral-lower-bound}
There is an absolute constant $c>0$ such that the following holds: Let $K \subset \R^n$ be a symmetric convex body such that $|\vol(K) -1/2| \le c$. Then 
\[
\Ex_{\bg \sim \Nn}[(K(\bg)-\vol(K))  \cdot {\ovp_{r_\ast}}(\bg)] = \Omega(1). 
\]
%Further, the constant on the right hand side is independent of the choice of $c$ (as long as $c>0$ is sufficiently small so that Claim~\ref{clm:r-ast} holds). 
\end{claim}

\begin{proof}
Applying \Cref{clm:Hermite-deg2}, we have that
\begin{equation}~\label{eq:expect-equiv-beta1}
\Ex_{\bg \sim \Nn}[(K(\bg)-\vol(K)) \cdot {\ovp_{r_\ast}}(\bg)] = \int_{\nu=0}^{1} \overline{\beta}(\nu)  p_{r_\ast}(r(\nu)) d\nu. 
\end{equation}
A crucial observation is that the value of $\overline{\beta}(\nu)$ is positive (respectively, negative) only if 
$\nu < \nu_\ast$ (respectively, $\nu > \nu_\ast$). Similarly, $p_{r_\ast}(r(\nu))$ is positive  if and only if $r(\nu) \le r_\ast$, which holds if and only if $\nu \le \nu_\ast$. Thus we have that
\begin{equation}~\label{eq:c-positive1}
\overline{\beta}(\nu)  p_{r_\ast}(r(\nu))  \ge 0 \text{~for all~} \nu \in [0,1]. 
\end{equation}

We define the values $r_{\ast, \downarrow}< r_\ast$ and $r_{\ast, \uparrow}>r_\ast$ as 
\begin{equation}~\label{eq:r-star-up-down}
r_{\ast, \downarrow} := r_\ast \cdot \bigg( 1 - \frac{1}{10 \sqrt{n}} \bigg) \quad \quad \textrm{and} \quad \quad r_{\ast, \uparrow} := r_\ast \cdot \bigg( 1 + \frac{1}{10 \sqrt{n}} \bigg). 
\end{equation}

We also choose $\nu_{\ast,\downarrow}$ and $\nu_{\ast,\uparrow}$ to be such that 
$r(\nu_{\ast,\downarrow}) = r_{\ast,\downarrow}$ and $r(\nu_{\ast,\uparrow}) = r_{\ast,\uparrow}$. By \Cref{lem:key}, it follows that
\begin{align}
\alpha_K(r_{\ast,\downarrow} ) &\ge \alpha_K(r_{\ast})  + \frac{\Theta(1)}{\sqrt{n}} =\vol(K) + 
\frac{\Theta(1)}{\sqrt{n}} \textrm{  and } \nonumber \\
\alpha_K(r_{\ast,\uparrow} ) &\le \alpha_K(r_{\ast})  - \frac{\Theta(1)}{\sqrt{n}} = \vol(K) -
\frac{\Theta(1)}{\sqrt{n}}, \label{eq:alpha-k-bounds}
\end{align}
where $\Theta(1)$ is an absolute positive constant (independent of $c$ in the statement of the claim).
Recalling the definition of $\overline{\beta}$, this implies that
%\[ \alpha_K(r_{\ast,\uparrow}) \leq \alpha_K(r_{\ast})  - \Theta \big(\frac{1}{\sqrt{n}} \big) = \vol(K) + 
%\Theta \big(\frac{1}{\sqrt{n}} \big).  
%\] 
\begin{equation}~\label{eq:c-gap1}
\overline{\beta}(\nu_{\ast, \downarrow}) \ge \frac{{c_1}}{\sqrt{n}} 
\quad \quad \textrm{and} \quad \quad \overline{\beta}(\nu_{\ast, \uparrow}) \leq -\frac{{c_1}}{\sqrt{n}}
\end{equation}
{for an absolute constant $c_1>0$.} Likewise, \ignore{choosing $c$ to be a sufficiently small constant,}\Cref{clm:r-ast-bound1} implies that
 $r_\ast = \Theta(\sqrt{n})$.
 Using this, we have that 
\begin{equation}~\label{eq:r-gap1}
\text{for all~} \nu \le \nu_{\ast, \downarrow}, \ \ 
p_{r_{\ast}}(r(\nu)) \ge {c_2}\sqrt{n} \quad \quad \textrm{and} \quad \quad \text{for all~} \nu \ge \nu_{\ast, \uparrow}, \ \ 
p_{r_{\ast}}(r(\nu)) \le -{c_2}\sqrt{n}
\end{equation}
{for an absolute constant $c_2>0$.}  We now can infer that 
\begin{align}
\Ex_{\bg \sim \Nn}[(K(\bg) -\vol(K)) \cdot \overline{p}_{r_\ast}(\bg)] &= \int_{\nu=0}^{1} \overline{\beta}(\nu)  p_{r_\ast}(r(\nu)) d\nu \ \tag{using \eqref{eq:expect-equiv-beta1}} \nonumber \\
&\ge \int_{\nu=0}^{\nu_{\ast, \downarrow}} \overline{\beta}(\nu)  p_{r_\ast}(r(\nu)) d\nu +\int_{\nu_{\ast, \uparrow}}^1 \overline{\beta}(\nu)  p_{r_\ast}(r(\nu)) d\nu \ \tag{using \eqref{eq:c-positive1}}\nonumber \\
&\ge {c_1 c_2}\nu_{\ast, \downarrow} + {c_1 c_2}(1-\nu_{\ast, \uparrow}),  \label{eq:lower-bound-K-f-corr1}
\end{align}
where the last line is using \eqref{eq:c-gap1} and \eqref{eq:r-gap1}.
%~\label{eq:lower-bound-K-f-corr}
%
%
%\ge  
%
%
%
%\Theta(1) \cdot (c_{\ast, \mathsf{in}} + (1- c_{\ast, \mathsf{out}})). 

Now we observe that combining \Cref{clm:r-ast-bound1} and the definition of $r_{\ast, \downarrow}$ and $r_{\ast, \uparrow}$, we have that
\[
 r_{\ast}-r_{\ast,\downarrow} \le \frac{1}{5} \ \   \textrm{and} \  \ r_{\ast,\uparrow} -r_{\ast}\le \frac{1}{5}, \quad \quad \text{and hence}
 \quad \quad r_{\ast,\uparrow} - r_{\ast,\downarrow} \le \frac25.
\]
By \Cref{fact:chi-squared-2}, this implies that
\[
\nu_{\ast,\uparrow} - \nu_{\ast,\downarrow} \le \frac25.
\]
It follows that \eqref{eq:lower-bound-K-f-corr1} is at least ${c_1 c_2}\cdot (3/5),$ which establishes the claim. 
\end{proof}

Now it is straightforward to prove \Cref{lem:unbias-Hermite-weight}:

\begin{proofof}{\Cref{lem:unbias-Hermite-weight}}
\ignore{\rnote{I had to stare at this for a while to get it so I added some more explanation}} For notational convenience let us define the function $K'(g) := K(g)-\vol(K)$. We first observe that $\Ex_{\bg \sim N(0,1)^n}[K'(\bg)]=0$, which means that the constant Hermite coefficient $\widetilde{K'}(0^n)$ is zero. 

We further observe that $\Var[\overline{p}_{r_\ast}(\bg)]$ is equal to the variance of the chi-squared distribution with $n$ degrees of freedom, which is $2n$. Since $\overline{p}_{r_\ast}(x) = (r_\ast^2 - (x_1^2 + \ldots + x_n^2))$ is a linear combination of degree-$0$ and degree-$2$ Hermite polynomials, the fact that $\Var[\overline{p}_{r_\ast}(\bg)]=2n$ can be rephrased in Hermite terms as $\sum_{|S|=2} \widetilde{\overline{p}_{r_\ast}}(S)^2=2n.$

We thus get that
\begin{align} 
\Ex_{\bg \sim \Nn}[K' (\bg) \cdot {\ovp_{r_\ast}}]
&=
\sum_{|S| =0,2} \widetilde{K'}(S) \widetilde{\ovp_{r_\ast}}(S)\tag{Plancherel} \nonumber\\
&=
\sum_{|S| = 2} \widetilde{K'}(S) \widetilde{\ovp_{r_\ast}}(S) \tag{since $\widetilde{K'}(0^n)=0$} \nonumber\\
&\leq
\sqrt{\mathsf{W}^{= 2}[K] \cdot \sum_{|S| = 2} \widetilde{\ovp_{r_\ast}}(S)^2} 
= \sqrt{\mathsf{W}^{= 2}[K] \cdot 2n}.
\tag{Cauchy-Schwarz} \nonumber
%&\le
%\mathsf{W}^{\le 2}[K] \cdot \sqrt{\Ex_{\bg \sim N(0,1)^n}[(\ell \cdot \bg - \theta)^2]}.\label{eq:plancherel}
\end{align}

%Next, observe that for $\bg \sim N(0,1)^n$ we have $\Ex[(K(\bg) - \vol(K)) ] =0$ and $\mathsf{Var}[ \overline{p}_{r_\ast}(\bg)] = {2n}$.\rnote{The variance is the variance of the chi-squared distribution, which is $2n$, right?} 
%\gray{From this, using Cauchy-Schwartz, it follows that 
%\[\Ex_{\bg \sim \Nn}[(K(\bg)-\vol(K)) \cdot {\ovp_{r_\ast}}(\bg)] \le \sqrt{\mathsf{W}^{\le 2}[K] \cdot {2n}}. \]

Recalling that $K' (\bg)=K(\bg)-\vol(K)$,
applying \Cref{lem:integral-lower-bound} finishes the proof. 
\end{proofof}

\subsection{Hermite mass at low weight levels for general convex sets}~\label{sec:hermite-centrally-asymmetric}
In this section we prove \Cref{thm:centrally-asymmetric-weight}, which will be a consequence of the following three lemmas. The first is \Cref{lem:constant-bias-Hermite-weight} which we repeat below for convenience:

\begin{replemma}{lem:constant-bias-Hermite-weight}
Let $\delta>0$ and $K \subseteq \mathbb{R}^n$. 
If $|\vol(K) - 1/2| \ge\delta$, then $\mathsf{W}^{= 0}[K] \ge 4 \delta^2$. 
%Let $\delta>0$ and $K \subseteq \mathbb{R}^n$ be a convex set. Define $\vol(K) = \Pr_{\bg \sim \Nn}[\bg \in K]$.  If $|\vol(K) - 1/2| \ge\delta$, $\mathsf{W}^{\le 2}[K] \ge 4 \delta^2$.
\end{replemma}

\begin{lemma}~\label{lem:W1-large}
There exist positive constants {$0 < \tau < 10^{-4}, c\geq \tau$} such that the following holds: Let $K \subset \mathbb{R}^n$ be a convex set such that  $|\vol(K) - 1/2| \le c/n$ and there is a point $x \not \in K$  with $\Vert x \Vert_2 \le \tau$.  Then $\mathsf{W}^{\le 1}[K] \ge \frac{1}{18}$. 
 \end{lemma}
 
\begin{lemma}\label{lem:W2-large-asym}
{For the constants $c,\tau$ in \Cref{lem:W1-large}} the following holds:
Let $K \subset \mathbb{R}^n$ be a convex set such that $|\vol(K)-1/2| \leq c/n$ and $B(0^n, \tau) \subseteq K$. Then $\mathsf{W}^{\le 2}[K] =\Omega(1/n^2)$. 
\end{lemma}

(As the above three lemmas suggest, the ideas in this section are reminiscent of those in \Cref{sec:weak-learner-general-convex}.)
We first prove \Cref{lem:W1-large} followed by \Cref{lem:W2-large-asym}. 

\begin{proofof}{\Cref{lem:W1-large}}
The proof of this Lemma is quite similar to the proof of \Cref{lem:halfspace-learning}.\ignore{So, the reader might find it helpful to recall the proof of that lemma.} In particular, exactly as in \Cref{lem:halfspace-learning}, using the supporting hyperplane theorem we get that there is a halfspace $\halfspace(x) = \sign(\ell \cdot x  -\theta)$ such that 
 \begin{enumerate}
 \item $K \subseteq \halfspace^{-1}(1)$; 
 \item $\ell$ is a unit vector and $|\theta| \le \tau$. 
 \end{enumerate}
 Now, {using $\tau \le c$}, following the same calculation that gave \eqref{eq:false-pos}, we get that 
 \begin{equation}~\label{eq:false-pos-2}
 \Prx_{\bg \sim \Nn} [K(\bg) = \halfspace(\bg) ] \ge 1- 2\tau -  \frac{2c}{n} \ge 1- 4 \tau, 
 \end{equation} 
{where the last inequality holds for $n$ sufficiently large}.
Next, we have 
 \begin{eqnarray}\label{eq:bounding-deg1-corr} 
 \Ex_{\bg \sim \Nn} [K(\bg) \cdot (\ell \cdot \bg - \theta)] = \Ex_{\bg \sim \Nn} [\halfspace(\bg) \cdot (\ell \cdot \bg - \theta)] - \Ex_{\bg \sim \Nn} [h(\bg) \cdot (\ell \cdot \bg - \theta)], 
 \end{eqnarray}
 where $h: \mathbb{R}^n \rightarrow \{-2,0,2\}$ is defined as $h(x) := \halfspace(x) - K(x)$. We now bound the two expectations in \eqref{eq:bounding-deg1-corr} individually.  For the first, we have that
 %To do this, we recall the Khintchine-Kahane inequality (see \cite{Haagerup82}). 
 %\begin{fact}~\label{fact:Khintchine}
 %Let $u \in \mathbb{R}^n$ be a unit vector. Then,
 %\[
 %\Ex_{x \in \{\pm 1\}^n} [|\langle u, x \rangle|] \ge \frac{1}{\sqrt{2}}. 
 %\]
 %\end{fact}
  \begin{equation}~\label{eq:Khintchine}
 \Ex_{\bg \sim \Nn} [\halfspace(\bg) \cdot (\ell \cdot \bg - \theta)] = \Ex_{\bg \sim \Nn} [|(\ell \cdot \bg - \theta)| ]
 = \Ex_{\bg_1 \sim N(0,1)}[|\bg_1 - \theta|] \ge \Ex_{\bg_1 \sim N(0,1)} [| \bg_1| ] =\sqrt{\frac{2}{\pi}}.
 \end{equation}
\ignore{ The last equality just follows from the fact that since $\ell$ is a unit vector, then $\ell \cdot \bg$ is distributed as a standard normal -- thus, the mean of the absolute value is $\sqrt{2/\pi}$.} On  the other hand, by Cauchy-Schwartz, we have 
 \begin{eqnarray}~\label{eq:CS-bound}
 \Ex_{\bg \sim \Nn} [h(\bg) \cdot (\ell \cdot \bg - \theta)] &\le& \sqrt{\Ex_{\bg \sim \Nn}[h^2(\bg)]} \cdot \sqrt{\Ex_{\bg \sim \Nn}[(\ell \cdot \bg - \theta)^2]}.
 \end{eqnarray}
Since $\Pr_{\bg \sim \Nn} [K(\bg) \not = \halfspace(\bg)] \le 4 \tau$ (by \Cref{eq:false-pos-2}) and $|h|=2$ when $\halfspace \neq K$, we have that $\Ex_{\bg \sim \Nn}[h^2(\bg)] \le 16 \tau$. For the other expectation on the right-hand side of  \eqref{eq:CS-bound}, since $\ell$ is a unit vector we have 
\begin{equation} \label{eq:square-lin}
\Ex_{\bg \sim \Nn}[(\ell \cdot \bg - \theta)^2] = 1+\theta^2.
\end{equation} 
Plugging this back into \eqref{eq:CS-bound}, {observing that $|\theta|\le \tau \le 1$}, we get that
 \[
 \Ex_{\bg \sim \Nn} [h(\bg) \cdot (\ell \cdot \bg - \theta)] \le 4\sqrt{\tau} \cdot \sqrt{1+\theta^2} \le 8 \sqrt{\tau}.  
 \]
%The last inequality uses the fact that $|\theta| \le 1$. 
Using this and \eqref{eq:Khintchine} and applying this in \eqref{eq:bounding-deg1-corr}, we obtain that 
 \begin{equation} \label{eq:onethird}
 \Ex_{\bg \sim \Nn} [K(\bg) \cdot (\ell \cdot \bg - \theta)] \ge \sqrt{\frac{2}{\pi}} - 8 \sqrt{\tau} \ge \frac{1}{3}, 
\end{equation}
where the last inequality uses $\tau \le 10^{-4}$. 

\ignore{\rnote{Similar to the earlier Hermite/Cauchy-Schwarz argument, this wasn't obvious to me, so I added a little more explanation}}
Next, we observe that by the linearity of $\ell \cdot \bg - \theta$, Plancherel's identity, and Cauchy-Schwarz,we get that (writing $v(g)$ for the function $\ell \cdot g - \theta$)
\begin{align} 
\Ex_{\bg \sim \Nn}[K (\bg) \cdot (\ell \cdot \bg - \theta)]
=
\sum_{|S| \leq 1} \widetilde{K}(S) \widetilde{v}(S) 
&\leq
\mathsf{W}^{\le 1}[K] \cdot \sqrt{\sum_{|S| \le 1} \widetilde{v}(S)^2} \nonumber\\
&\le
\mathsf{W}^{\le 1}[K] \cdot \sqrt{\Ex_{\bg \sim N(0,1)^n}[(\ell \cdot \bg - \theta)^2]}.\label{eq:plancherel}
\end{align} 
Finally, we can combine the above ingredients to  get that
\begin{eqnarray}\nonumber
\mathsf{W}^{\le 1}[K] \ge \frac{\Ex_{\bg \sim \Nn}[K (\bg) \cdot (\ell \cdot \bg - \theta)]^2}{\Ex_{\bg \sim \Nn}[ (\ell \cdot \bg - \theta)^2]} \ge \frac{1}{9(1+\theta^2)} \ge \frac{1}{18}, 
\end{eqnarray}
where the first inequality is from \Cref{eq:plancherel}, the second is from \Cref{eq:square-lin} and \Cref{eq:onethird}, and the final inequality again uses $|\theta| \le 1$. This finishes the proof of \Cref{lem:W1-large}.
\end{proofof}

\begin{proofof}{\Cref{lem:W2-large-asym}}
The proof of this lemma is essentially the same as the proof of \Cref{lem:unbias-Hermite-weight}; the main difference is that we apply \Cref{lem:key-general} instead of \Cref{lem:key}. In particular, we define the functions $r(\cdot)$, $\beta(\cdot)$, $\overline{\beta}(\cdot)$, as well as the quantities $\nu_{\ast}$ and $r_{\ast}$, exactly as in the proof of \Cref{lem:unbias-Hermite-weight} (recall that all these quantities are defined right before \Cref{clm:r-ast-bound1}). 
%%bb 

The following claim is analogous to \Cref{clm:r-ast-bound1}:

\begin{claim}~\label{clm:r-ast-asymmetric}
{For the constants $c,\tau$ in \Cref{lem:W1-large}} the following holds: Let $K$ be a  convex set such that $|\vol(K) - 1/2| \le c$ and ${B}(0^n, \tau) \subseteq K$. 
Then 
 \[
\frac{n}{4} \le r_\ast^2 \le  4n. 
 \]
\end{claim}
\begin{proof}
The proof is essentially the same as the proof of \Cref{clm:r-ast-bound1}, so we just indicate the changes vis-a-vis the proof of \Cref{clm:r-ast-bound1}. Similar to \Cref{clm:r-ast-bound1}, we will 
show that $r_\ast \le 2\sqrt{n}$, and again the other direction is similar. Towards a contradiction, assume that $r_\ast > 2\sqrt{n}$ and 
define $r_{\ast,\out} = \sqrt{2n}$. Then, by applying \Cref{lem:key-general}, we get that
\begin{equation}\label{eq:inc-r-ast-r1}
\alpha_K(r_{\ast,\out}) \ge \alpha_K(r_\ast) + \kappa = \vol(K) + \kappa 
\end{equation}
where $\kappa = \Theta(\tau/\sqrt{n})$. Note that in contrast with \Cref{eq:inc-r-ast-r}, in which $\kappa$ is an absolute constant, here $\kappa$ is $\Theta(1/\sqrt{n})$. We observe that \eqref{eq:bias-calc1} and \eqref{eq:kappa-gap2} both continue to hold. Further, as long as $n$ is sufficiently large,  $n \ge \frac{64}{3} \ln (4/\kappa)$ continues to hold. Thus, exactly as in \eqref{eq:kappa-gap1}, we get that
\begin{equation}
\vol(K) \ge \alpha_K(r_{\ast,\out}) \cdot (1-\kappa/4) \ge (\vol(K) + \kappa) \cdot (1-\kappa/4) \ge \vol(K) + \frac{\kappa}{2}. 
\end{equation}
This contradiction implies that $r_{\ast} \le 2\sqrt{n}$. As in the proof of \Cref{clm:r-ast-bound1}, 
the proof of the other direction is similar. 
\end{proof}

The following claim is analogous to \Cref{lem:integral-lower-bound}:

\begin{claim}~\label{clm:integral-lowerb}
{For the constants $c,\tau$ in \Cref{lem:W1-large}} the following holds: Let $K$ be a  convex set such that $|\vol(K) - 1/2| \le c$ and ${B}(0^n, \tau) \subseteq K$. Then 
\[
\Ex_{\bg \sim \Nn}[(K(\bg) - \vol(K)) \cdot \ovp_{r_\ast}(\bg)] \ge \Theta(n^{-1/2}). 
\]
\end{claim}
\begin{proof}
The proof is essentially the same as the proof of \Cref{lem:integral-lower-bound}, so we just indicate the changes vis-a-vis the proof of \Cref{lem:integral-lower-bound}. Equations \eqref{eq:expect-equiv-beta1} and \eqref{eq:c-positive1}  holds exactly as before. We now define 
$r_{\ast, \downarrow}$ and $r_{\ast, \uparrow}$ as 
in \eqref{eq:r-star-up-down}, i.e. 
\[
r_{\ast, \downarrow} = r_\ast \cdot \bigg( 1 - \frac{1}{10 \sqrt{n}} \bigg) \ \textrm{and} \ r_{\ast, \uparrow} = r_\ast \cdot \bigg( 1 + \frac{1}{10 \sqrt{n}} \bigg). 
\]
As before, we define $\nu_{\ast,\downarrow}$ and $\nu_{\ast,\uparrow}$ to be such that 
$r(\nu_{\ast,\downarrow}) = r_{\ast,\downarrow}$ and $r(\nu_{\ast,\uparrow}) = r_{\ast,\uparrow}$. Now applying \Cref{lem:key-general}, we get that
\begin{align}
\alpha_K(r_{\ast,\downarrow} ) &\ge \alpha_K(r_{\ast})  + \frac{\Theta(\tau)}{{n}} =\vol(K) + 
\frac{\Theta(\tau)}{{n}} \textrm{  and } \nonumber \\
\alpha_K(r_{\ast,\uparrow} ) &\le \alpha_K(r_{\ast})  - \frac{\Theta(\tau)}{{n}} = \vol(K) -
\frac{\Theta(\tau)}{{n}}.\label{eq:alpha-k-bounds1}
\end{align} 
Note that in contrast to \eqref{eq:alpha-k-bounds}, where the gap between $\alpha_K(r_{\ast,\downarrow} )$ (or $\alpha_K(r_{\ast,\uparrow} )$) and $\alpha_K(r_{\ast})$ was $\Theta(1/\sqrt{n})$, now the gap is only $\Theta(1/n)$. 
This implies that 
\begin{equation}~\label{eq:c-gap2}
\overline{\beta}(\nu_{\ast, \downarrow}) \ge \frac{{c_1}\tau}{{n}} \ \textrm{and} \ \overline{\beta}(\nu_{\ast, \uparrow}) \leq -\frac{{c_1\tau}}{{n}} 
\end{equation}
{for an absolute constant $c_1>0$.}
Now, by applying \Cref{clm:r-ast-asymmetric}, we get that 
\begin{equation}~\label{eq:r-gap2}
\text{for all~} \nu \le \nu_{\ast, \downarrow}, \ \ 
p_{r_{\ast}}(r(\nu)) \ge {c_2}\sqrt{n} \quad \quad \textrm{and} \quad\quad \text{for all~} \nu \ge \nu_{\ast, \uparrow}, \ \ 
p_{r_{\ast}}(r(\nu)) \le -{c_2}\sqrt{n} 
\end{equation}
{for an absolute constant $c_2>0$.}
This is exactly the same as \eqref{eq:r-gap1} except that we applied \Cref{clm:r-ast-asymmetric} to get this instead of \Cref{clm:r-ast-bound1} as was done in the proof of \Cref{lem:integral-lower-bound}. 
As before, we can now infer that 
\begin{align}
\Ex_{\bg \sim \Nn}[(K(\bg) -\vol(K)) \cdot \overline{p_{r_\ast}}(\bg)] &=\int_{\nu=0}^{1} \overline{\beta}(\nu)  p_{r_\ast}(r(\nu)) d\nu \ \tag{using \eqref{eq:expect-equiv-beta1}} \nonumber \\
&\ge \int_{\nu=0}^{\nu_{\ast, \downarrow}} \overline{\beta}(\nu)  p_{r_\ast}(r(\nu)) d\nu +\int_{\nu_{\ast, \uparrow}}^1 \overline{\beta}(\nu)  p_{r_\ast}(r(\nu)) d\nu \ \tag{using \eqref{eq:c-positive1}}\nonumber \\
&\ge \frac{{c_1 c_2}\tau}{\sqrt{n}} \nu_{\ast, \downarrow} + \frac{{c_1 c_2}\tau}{\sqrt{n}} (1-\nu_{\ast, \uparrow}),
\label{eq:lower-bound-K-f-corr2}
%~\label{eq:lower-bound-K-f-corr}
%
%
%\ge  
%
%
%
%\Theta(1) \cdot (c_{\ast, \mathsf{in}} + (1- c_{\ast, \mathsf{out}})). 
\end{align}
where the last line is using \eqref{eq:c-gap2} and \eqref{eq:r-gap2}.

Now we observe that combining \Cref{clm:r-ast-asymmetric} and the definition of $r_{\ast, \downarrow}$ and $r_{\ast, \uparrow}$, as before we have that
\[
 r_{\ast}-r_{\ast,\downarrow} \le \frac{1}{5} \ \   \textrm{and} \  \ r_{\ast,\uparrow} -r_{\ast}\le \frac{1}{5} \ \Rightarrow r_{\ast,\uparrow} - r_{\ast,\downarrow} \le \frac25,
\]
which implies (using \Cref{fact:chi-squared-2}) that
\[
\nu_{\ast,\uparrow} - \nu_{\ast,\downarrow} \le \frac25.
\]
Plugging the above into \eqref{eq:lower-bound-K-f-corr2}, \Cref{clm:integral-lowerb} is proved.
\end{proof}
The rest of the proof of \Cref{lem:W2-large-asym} follows exactly the lines of the proof of \Cref{lem:unbias-Hermite-weight}. As before the polynomial $\overline{p}_{r_\ast}(x) = (r_\ast^2 - (x_1^2 + \ldots + x_n^2))$ is a linear combination of degree-$0$ and $2$ Hermite polynomials, $\Ex[(K(\bg) - \vol(K)) ] =0$, and $\mathsf{Var}( \overline{p}_{r_\ast}(\bg)) = 2n$, so the exact same argument as before, but now using \Cref{clm:integral-lowerb} instead of \Cref{lem:integral-lower-bound}, finishes the proof. 
\end{proofof}

\subsection{\Cref{thm:centrally-symmetric-weight} is almost tight} \label{sec:symmetric-fourier-weight-tightness}

Recall that \Cref{thm:centrally-symmetric-weight} says that any centrally symmetric convex body $K$ (viewed as a function to $\bits$) has $\Weight^{\leq 2}[K] \geq {\Omega(1/n)}.$  In this section we show that this lower bound is  best possible up to polylogarithmic factors:

\begin{fact} \label{fact:cube}
There is a centrally symmetric convex body $K$ (in fact, an intersection of $2n$ halfspaces) in $\R^n$ which, viewed as a function to $\bits$, has $\Weight^{\leq 2}[K] \leq O({\frac {\log^2 n} {n}}).$
\end{fact}

\begin{proof}
The body $K$ is simply an origin-centered axis-aligned cube of size chosen so that $\Vol(K)=1/2$ (and hence the constant Hermite coefficient $\tilde{K}(0^n)$ is precisely $0$). In more detail, let $c =c(n) > 0$ be the unique value such that
\begin{equation} \label{eq:def-of-c}
\Prx_{\bg \sim N(0,1)}[|\bg|\leq c] = (1/2)^{1/n} = 1 - {\frac {\Theta(1)}{n}}
\end{equation}
(so by standard bounds on the tails of the Gaussian distribution we have $c = \Theta(\sqrt{\log n})$).
Let $a: \R \to \{0,1\}$ be the indicator function of the interval $[-c,c]$, so $a(t) := \mathbf{1}[|t| \leq c]$, and let $K_1: \R^n \to \{0,1\}$ be the indicator function of the corresponding $n$-dimensional cube, so
\[
K_1(x_1,\dots,x_n) = \prod_{i=1}^n a(x_i)
\]
and $K: \R^n \to \bits$ (the $\bits$-valued version of $K_1$) is $K :=2K_1-1.$ We have that $\tilde{K}(0^n) = 2 \tilde{K_1}(0^n) - 1$ and $\tilde{K}(\overline{i}) = \tilde{K_1}(\overline{i})$ for every $\overline{i} \in \N^n \setminus \{0^n\}$ so it suffices to analyze the Hermite spectrum of $K_1.$ Since $K_1$ has such a simple structure (product of univariate functions $a(x_1),\dots,a(x_n)$) this is happily simple to do; details are below.

By construction we have that $\tilde{K_1}(0^n) = 1/2$ and hence $\tilde{K}(0^n)=0$ as desired. Since $a$ is an even function we have $\tilde{a}(1)=0$ and hence $\Weight^{=1}(K_1)=0$.Since $\tilde{a}(1)=0$ the only nonzero degree-2 Hermite coefficients of $\tilde{K_1}$ are the coefficients indexed by $2e_i$, $i=1,\dots,n$, all of which are the same, so $\Weight^{\leq 2}(K_1)$ is equal to
\begin{equation} \label{eq:goal}
\Weight^{\leq 2}(K_1) = 
n \cdot \left( \tilde{a}(0)^{n-1} \cdot \tilde{a}(2)\right)^2 = n \cdot \left( \left( \Ex_{\bg \sim N(0,1)}[a(\bg)]\right)^{n-1} \cdot \tilde{a}(2) \right)^2 = \Theta(n) \cdot \tilde{a}(2)^2.
\end{equation}
Recalling that the degree-2 univariate Hermite basis polynomial is $h_2(x) = {\frac {x^2 - 1}{\sqrt{2}}}$, we have that
\begin{align}
\tilde{a}(2) &= \Ex_{\bg \sim N(0,1)}[ a(\bg) h_2(\bg)] = {\frac 1 {\sqrt{2}}} \cdot \Ex_{\bg \sim N(0,1)}[a(\bg)(\bg^2 - 1)]\nonumber \\
&= {\frac 1 {\sqrt{2}}} \int_{-c}^c e^{-x^2/2} (x^2 - 1) dx \nonumber \\
&= -\sqrt{2} c e^{-c^2/2}.\label{eq:expression-for-degree-2-coefficient}
\end{align}
We now recall the following tail bound on the normal distribution (Equation~2.58 of \cite{TAILBOUND}):
\begin{equation} \label{eq:normal-tail}
\phi(t)
\left({\frac 1 t} - {\frac 1 {t^3}} \right) \leq \Prx_{\bg \sim N(0,1)}[\bg \geq t] \leq
\phi(t)
\left({\frac 1 t} - {\frac 1 {t^3}} + {\frac 3 {t^5}}\right),
\end{equation}
where $\phi(t) = {\frac 1 {\sqrt{2 \pi}}} e^{-t^2/2}$ is the density function of $N(0,1)$.
Combining \Cref{eq:normal-tail}, \Cref{eq:expression-for-degree-2-coefficient} and \Cref{eq:def-of-c} we get that $|\tilde{a}(2)| = \Theta({\frac {\log n} n})$, which establishes the claimed fact by \Cref{eq:goal}.
\end{proof}

%!TEX root = draft1b.tex

\section{Lower bounds} 
In this section we prove \Cref{thm:our-BBL-lb}, which we restate here for the convenience of the reader:
\begin{reptheorem}{thm:our-BBL-lb}
For sufficiently large $n$, for any $s \geq n$, there is a distribution ${\cal D}_{\actual}$ over centrally symmetric convex sets with the following property:  for a target convex set $\boldf \sim {\cal D}_{\actual},$ for any membership-query (black box query) algorithm $A$ making at most $s$ many queries to $\boldf$, the expected error of $A$ (the probability over $\boldf \sim {\cal D}_{\actual}$, over any internal randomness of $A$, and over a random Gaussian $\bx \sim N(0,1^n)$, that the output hypothesis $h$ of $A$ predicts incorrectly on $\bx$) is at least $1/2 - {\frac {O(\log(s) \cdot \sqrt{\log n})}{n^{1/2}}}$.
\end{reptheorem}

We note that this lower bound holds even in the membership query (hereafter abbreviated as MQ) model. In this model the learning algorithm has query access to a black-box oracle for the unknown target function $\boldf$; note that a learning algorithm in this model and can simulate a learning algorithm in the model where the algorithm receives only random labeled examples of the form $(\bx, \boldf(\bx))$ (with $\bx \sim \Nn$) with no overhead. 
Thus a lower bound in the MQ model  holds \emph{a fortiori} for the random examples model (which is the model that our algorithms use).  In particular, by instantiating $s= \poly(n)$ in the above theorem, we get that no algorithm which receives $\poly(n)$ samples (and hence no algorithm running in $\poly(n)$ time) can achieve an advantage of $\frac{\omega(\log^{3/2} n)}{\sqrt{n}}$ over random guessing for learning centrally symmetric convex sets. Thus, our algorithm for weak learning of centrally symmetric convex sets, i.e., Theorem~\ref{thm:weak-learn-centrally-symmetric}, achieves an optimal advantage (up to an $O(\log^{3/2} n)$ factor).

Since the proof of \Cref{thm:our-BBL-lb} is somewhat involved we begin by explaining its general strategy:

\begin{enumerate}
\item We start by constructing a ``hard" distribution ${\cal D}_{\ideal}$ over centrally symmetric convex subsets of $\mathbb{R}^n$ (note that ${\cal D}_{\ideal}$ is different from the final distribution ${\cal D}_{\actual}$). We then analyze the case in which the learning algorithm is not allowed to make \emph{any} queries to the target function $\boldf \sim {\cal D}_{\ideal}$. It is easy to see that 
the maximum possible accuracy of any \red{zero-query} learning algorithm is achieved by the so-called \emph{Bayes optimal classifier} (which we denote by $BO_{{\cal D}_{\ideal}}$) which labels each $x \in \mathbb{R}^n$ as follows: 
\[
BO_{{\cal D}_{\ideal}} (x) = \begin{cases} 1 \ &\textrm{if} \ \Pr_{\boldf \sim {\cal D}_{\ideal}}[\boldf(x) = 1] \ge 1/2 \\ \red{0}  \ &\textrm{otherwise}.
 \end{cases}
\]
We show that for ``most" $\bx$ sampled from $\Nn$, the accuracy of $BO_{{\cal D}_{\ideal}} (\bx)$ is close to $1/2$  and in fact, the average advantage over $1/2$ for $\bx\sim \Nn$ is bounded by ${\frac {O(\log(s) \cdot \sqrt{\log n})}{n^{1/2}}}$. 

\item The distribution ${\cal D}_{\ideal}$ is a continuous distribution defined in terms of a so-called \emph{Poisson point process}. While the construction of ${\cal D}_{\ideal}$ is particularly well-suited to the analysis of a zero-query learner, i.e.~of the Bayes optimal classifier (indeed this is the motivation for our introducing ${\cal D}_{\ideal}$), it becomes tricky to analyze ${\cal D}_{\ideal}$ when the learning algorithm is actually allowed to make queries to the target function $\boldf$. To deal with this, we ``discretize" the distribution ${\cal D}_{\ideal}$ to construct the actual hard distribution 
${\cal D}_{\actual}$ (which is finitely supported). The discretization is carefully done to ensure that for ``most" $\bx$ (again sampled from $\Nn$), $\Pr_{\boldf \in {\cal D}_{\ideal}}[\boldf(\bx)=1]$ is close to $\Pr_{\boldf \in {\cal D}_{\actual}}[\boldf(\bx)=1]$.
This implies that the average advantage of the Bayes optimal classifier for $\boldf \sim {\cal D}_{\actual}$ \red{(corresponding to the best possible zero-query learning algorithm)}, denoted by 
$BO_{{\cal D}_{\actual}}$, remains bounded by ${\frac {O(\log(s) \cdot \sqrt{\log n})}{n^{1/2}}}$.

%As before, the  is given by 
%\[
%BO_{{\cal D}_{\actual}} (x) = \begin{cases} 1 \ &\textrm{if} \ \Pr_{\boldf \sim {\cal D}_{\actual}}[\boldf(x) = 1] \ge 1/2 \\ -1  \ &\textrm{otherwise}
% \end{cases}
%\]

\item Finally, we consider the case when the learning algorithm is allowed to makes $s$ queries to the unknown target function $\boldf$. 
 Roughly speaking, we show that  for any choice of $s$ query points $\overline{y} = (y_1,\ldots, y_s)$,  with high probability over  both $\boldf \sim {\cal D}_{\actual}$ and 
 $\bx \sim \Nn$, the advantage of the optimal classifier is close to that achieved by $BO_{{\cal D}_{\actual}}$ (see \Cref{sec:queries}). The techniques used to prove this crucially rely on the specific construction of   ${\cal D}_{\actual}$, so we refrain from giving further details here. 
 However, using this and  the upper bound on the advantage of  $BO_{{\cal D}_{\actual}}$, we obtain Theorem~\ref{thm:our-BBL-lb}. 
\end{enumerate}

We note that the strategy outlined above  (in particular, steps~2 and 3 and the general flavor of the analysis used to establish those steps) closely follows the lower bound approach of 
Blum, Burch and Langford~\cite{BBL:98}, who
showed that no $s$-query algorithm in the MQ model can achieve an advantage of $\omega(\frac{\log s}{\sqrt{n}})$ over random guessing to learn monotone functions under the uniform distribution on $\bn$. 
Of course, 
the choice of the \emph{hard distribution} is quite different in our work than in \cite{BBL:98}; in particular, a draw from ${\cal D}_{\ideal}$ is essentially a random symmetric polytope with $\mathsf{poly}(s)$ facets where the hyperplane defining each facet is at distance around $O(\sqrt{\log s})$ away from the origin. The distribution ${\cal D}_{\actual}$ is obtained by essentially discretizing ${\cal D}_{\ideal}$ while retaining some crucial geometric properties. In contrast, the hard distribution in \cite{BBL:98} is constructed in one step and is essentially a random monotone DNF of width $O(\log s + \log n)$ with roughly $s$ terms. 
Another significant distribution between our argument and that of \cite{BBL:98} is the technical challenges that arise in our case because of dealing with a continuous domain and the resulting discretization that we have to perform.
 
\red{Finally, we note that in the proof of \Cref{thm:our-BBL-lb}, which we give below, we may assume that $s = 2^{O(\sqrt{n}/\sqrt{\log n})}$, since otherwise the claimed bound trivially holds.}
 
\subsection{The idealized distribution ${\cal D}_{\ideal}$ and the Bayes optimal classifier for it}

We will define the distribution ${\cal D}_{\actual}$ by first defining a related distribution ${\cal D}_{\ideal}$. As mentioned earlier, the distribution ${\cal D}_{\actual}$ will be obtained by discretization of ${\cal D}_{\ideal}$. 
%While ${\cal D}_{\ideal}$ is supported on an uncountably infinite set, ${\cal D}$  is finitely supported.  We will explain the rationale for defining  two distributions a little later. 
To define ${\cal D}_{\ideal}$, we need to recall the notion of a spatial  Poisson point process; we specialize this notion to the unit sphere $\mathbb{S}^{n-1}$, though it is clear that an analogue of the definition we give below can be given over any bounded measurable set $B \subseteq \mathbb{R}^n$. 

\begin{definition}
A \emph{point process $\bX$} on the \emph{carrier space $\mathbb{S}^{n-1}$} is a stochastic process such that a draw from this process is a sequence of points  $\bx_1, \ldots, \bx_{\boldN} \in
 \mathbb{S}^{n-1}$. \red{(Note that each individual point $\bx_i$ as well as the number of points $\boldN$ are all random variables as described below.) }
 
 A \emph{spatial Poisson point process with parameter $\lambda$ on $\mathbb{S}^{n-1}$} is a point process on $\mathbb{S}^{n-1}$ with the following two properties: 
 \begin{enumerate}
 \item For any subset $B \subseteq \mathbb{S}^{n-1}$, let $\boldN(B)$ denote the number of points which fall in $B$. Then, the distribution of $\boldN(B)$ follows $\mathsf{Poi}(\lambda \mu(B))$ where $\mu(B)$ is the fractional density of $B$ inside $\mathbb{S}^{n-1}$. 
 \item If $B_1, \ldots, B_k  \subseteq \mathbb{S}^{n-1}$ are pairwise disjoint sets, then $\boldN(B_1), \ldots, \boldN(B_k)$ are mutually independent. 
 \end{enumerate}
 
 Finally, we note that the spatial Poisson point porcess with parameter $\lambda$ can be realized as follows: Sample $\boldN \sim \mathsf{Poi}(\red{\lambda}),$ and output $\boldN$ points $\bx_1, \ldots, \bx_{\boldN}$ that are chosen uniformly and independently at random from $\mathbb{S}^{n-1}$. 
\end{definition}

We refer the reader to \cite{last2017lectures} and \cite{daley2007introduction} for details about Poisson point processes. 
We next choose $d>0$ so that for any unit vector $v$,
\begin{equation}~\label{eq:choose-d}
\Pro_{\bu \sim \mathbb{S}^{n-1}} \bigg[| v \cdot \bu | \ge \frac{d}{\sqrt{n}} \bigg] = \frac{1}{s^{100}}. 
\end{equation}
  Note that by symmetry the choice of $v$ is immaterial.  We also recall the following fundamental  fact about inner products with random unit vectors (which is easy to establish using e.g.~\Cref{eq:baum}):
  
  \begin{claim}~\label{clm:inner-product-random}
  Let $v \in \mathbb{S}^{n-1}$. \red{For any $0 < t < 1/2$,}  
  \[
  \Pro_{\bu \in \mathbb{S}^{n-1}} [| v \cdot \bu | \ge t ] = e^{-\Theta(\red{t^2 n})}. 
  \]
  \end{claim} 
\red{Since we have $s = 2^{O(\sqrt{n/\log n})}$,} it follows from this fact that $d = \Theta(\sqrt{\log s})$ in \eqref{eq:choose-d}. 
Next, for any unit vector $z \in \mathbb{S}^{n-1}$, we define the ``slab" function 
\[
\slab_z(x) \coloneqq \Ind\left[-d \leq z \cdot x \leq d\right]. 
\]
It is clear that for any $z$ the function $\slab_z(\cdot)$ defines a centrally symmetric convex set.
Finally, we define the parameter $\Lambda$ to be
\begin{equation}~\label{eq:def-Lambda}
\Lambda := s^{100} \cdot \ln 2. 
\end{equation} 

 A function $\boldf$ is sampled from ${\cal D}_{\ideal}$  as follows:
 
 \begin{itemize}
 
 \item Sample $\bz_1, \ldots, \bz_{\boldN}$ from the spatial Poisson point process on $\mathbb{S}^{n-1}$ with parameter $\Lambda$. 
 \item Set $\boldf$ to be
\[
\boldf(x) = \bigwedge_{i=1}^{\boldN} \slab_{\bz_i}(x). 
\]

 \end{itemize}
 
We have the following observation (whose proof is immediate from the construction): 
\begin{observation}~\label{obs:lb-const-obs}
\begin{enumerate}
\item Any $\boldf \sim  {\cal D}_{\ideal}$ defines a centrally symmetric convex set. 
\item For any point $x \in \mathbb{R}^n$, the value of ${\cal D}_{\ideal}(x) \coloneqq\Pr_{\boldf \sim  {\cal D}_{\ideal}} [\boldf(x) =1]$ is determined by $\Vert x \Vert_2$, \red{the distance of $x$ from the origin.}\ignore{ (if $\Lambda$ and $d$ are fixed). }
\end{enumerate}
\end{observation}

\subsubsection{Analyzing the Bayes optimal classifier for ${\cal D}_{\ideal}$}

We now bound the advantage  of the Bayes optimal classifier (denoted by $BO_{{\cal D}_{\ideal}}$) for ${\cal D}_{\ideal}$, which, \red{as stated earlier, corresponds to the best possible learning algorithm that} makes zero queries to the unknown target function $\boldf \sim {\cal D}_{\ideal}$. 
%Given that the prior distribution over target functions is ${\cal D}_{\ideal}$, the Bayes optimal classifier is a classifier which we ;
Observe that  on input $x \in \R^n$, the classifier $BO_{{\cal D}_{\ideal}}(x)$ outputs 1 if ${\cal D}_{\ideal}(x) \geq 1/2$ and outputs 0 on $x$ if ${\cal D}_{\ideal}(x) < 1/2$.  Thus, the expected error of $BO_{{\cal D}_{\ideal}}$  is
\begin{align*}
\opt({\cal D}_{\ideal}) &:= \Ex_{\bx \sim N(0,1)^n}[\min\{{\cal D}_{\ideal}(\bx),1-{\cal D}_{\ideal}(\bx)\}],
%&= \Ex_{\bx \sim N(0,1)^n}[\min\{(1-p)^{\#S_{\actual}(\bx)}, 1 - (1-p)^{\#S_{\actual}(\bx)}\}].
\end{align*}
\red{and the expected advantage of $BO_{{\cal D}_{\ideal}}$ is $1/2 - \opt({\cal D}_{\ideal})$.}

The next lemma bounds $\opt({\cal D}_{\ideal})$ \red{and completes Step~1 of the proof outline given earlier:}

\begin{lemma}~\label{lem:BO-accuracy}
We have
\[
\ignore{\bigg| \opt({\cal D}_{\ideal}) - \frac12 \bigg|}
{\frac 1 2} - \opt({\cal D}_{\ideal}) = \frac{O(\log s \sqrt{\log n})}{\sqrt{n}}.  
\] 
\end{lemma}
\begin{proof}
Define the set $\mathrm{S}_{\mathrm{med}} = \{x \in \mathbb{R}^n: |\Vert x \Vert_2^2  -n |
 \le \red{8} \sqrt{n \red{\ln n}} \}$.
\red{ Intuitively, this is the set of points whose distance from the origin is ``roughly typical'' for the Gaussian distribution; more formally,}
by  \Cref{lem:johnstone},  we have that 
\begin{equation}~\label{eq:relevant-x}
\Prx_{\bg \sim \Nn} [\bg \red{\notin} \mathrm{S}_{\mathrm{med}}] \le \frac{1}{n^5}. 
\end{equation}

We will show that the value of ${\cal D}_{\ideal}(x)$ is close to $1/2$ for every $x \in \mathrm{S}_{\mathrm{med}}$, which easily implies the lemma. To do this, we define $\region(x)$ to be \red{the set of those unit vectors $z$ such that $x$ does \emph{not} lie within the slab defined by $z$, i.e.} 
 \[
 \region (x) := \{ z \in \mathbb{S}^{n-1}: |  z \cdot x  | > d\}. 
 \]  
Observe that the fractional density of $\region(x)$ inside $\mathbb{S}^{n-1}$, which we denote by $\mu_1(\region(x))$, is determined by $\Vert x \Vert_2$. \red{We would like to analyze $\mu_1(\region(x))$ for all points $x \in \mathrm{S}_{\mathrm{med}}$; to do this, we first analyze it for points at distance exactly $\sqrt{n}$ from the origin.}
 So choose any point $a_0 \in \mathbb{R}^n$ such that $\Vert a_0 \Vert_2=\sqrt{n}$. By the definition of $d$ in \eqref{eq:choose-d} and observing that $a_0/\sqrt{n}$ is a unit vector, we have 
\begin{equation}~\label{eq:density-root-n}
 \mu_1(\region(a_0)) = \Prx_{\bu \sim \mathbb{S}^{n-1}} \left[ {\frac {a_0}{\sqrt{n}}} \cdot \bu \ge \frac{d}{\sqrt{n}}
 \right]  = \frac{1}{s^{100}}. 
 \end{equation}
 Next, consider any $b_0 \in \mathrm{S}_{\mathrm{med}}$, and note that $\Vert b_0\Vert_2 = \sqrt{n} (1+\delta)$ where $|\delta| =O(\sqrt{\frac{\log n}{n}})$. Hence 
 \[
\mu_1(\region(b_0)) = \Prx_{\bu \sim \mathbb{S}^{n-1}} \left[ {\frac {b_0}{\sqrt{n} (1+\delta)}} \cdot \bu  \ge \frac{d}{\sqrt{n}(1+\delta)}
 \right],
\] 
\red{where ${\frac {b_0}{\sqrt{n} (1+\delta)}}$ is a unit vector.}
Recalling that we can assume \ignore{by the bound on $s$ stated at the end of that we can assume without loss of generality,} $\log s \sqrt{\log n} \le c_0 \sqrt{n}$ for a sufficiently small positive constant $c_0>0$ \red{and that $d=\Theta({\sqrt{\log s}})$,}\ignore{ -- otherwise, the conclusion of the lemma holds trivially.
However, with this,} we can apply \Cref{lem:corr} to get that
 \begin{equation}~\label{eq:density-perturbation}
\bigg|  \frac{\mu_1(\region(a_0))}{\mu_1(\region(b_0))} - 1 \bigg| = O\bigg(d^2 \cdot \frac{\sqrt{\log n}}{\sqrt{n}} \bigg) = O\bigg(\frac{\log s \sqrt{\log n}}{\sqrt{n}} \bigg). 
 \end{equation}
 From \eqref{eq:density-perturbation} and \eqref{eq:density-root-n}, we get that every $x \in \mathrm{S}_{\mathrm{med}}$ satisfies
\begin{equation}~\label{eq:bound-region}
\mu_1(\region(x)) = \frac{1}{s^{100}} \cdot \bigg( 1+ O\bigg(\frac{\log s \sqrt{\log n}}{\sqrt{n}} \bigg)\bigg). 
\end{equation}

To finish the proof, we observe that sampling $\boldf \sim {\cal D}_{\ideal}$ is equivalent to sampling $\bz_1, \ldots, \bz_{\boldN}$ from the spatial Poisson point process on $\mathbb{S}^{n-1}$ with parameter $\Lambda$.  Let $\mathbf{{Num}}_x$ 
be the random variable defined as \red{$|\{\bz_i\}_{i=1}^{\boldN} \cap \region(x)|$}. Observe that 
\begin{enumerate}
\item $\boldf(x) =1$ iff $\mathbf{{Num}}_x =0$; 
\item $\mathbf{{Num}}_x$ is distributed as $\mathsf{Poi}(\Lambda \cdot \mu(\region(x)))$. 
\end{enumerate}
Putting these two items together with \eqref{eq:bound-region} and \eqref{eq:def-Lambda}, we get that for $x \in \mathrm{S}_{\mathrm{med}}$, 
\[
\Pro_{\boldf \sim {\cal D}_{\ideal}} [\boldf(x) =1] = \Pr[\mathsf{Poi}(\Lambda \cdot \mu(\region(x))))=0] = e^{-\Lambda \cdot \mu(\region(x))} = \frac{1}{2} + O \bigg( \frac{\log s \cdot \sqrt{\log n}}{\sqrt{n}}\bigg). 
\]
Combining the above equation with \eqref{eq:relevant-x}, we get \Cref{lem:BO-accuracy}. 
\end{proof}

\subsection{Discretizing ${\cal D}_{\ideal}$ to obtain ${\cal D}_{\actual}$, and the Bayes optimal classifier for ${\cal D}_{\actual}$} 
We now discretize the distribution ${\cal D}_{\ideal}$ to construct the distribution ${\cal D}_{\actual}$. We begin by recalling some results which will be useful for this construction.

\begin{definition}
Let ${\cal X}_1$, ${\cal X}_2$ be two distributions supported on $\mathbb{R}^n$.
The \emph{Wasserstein distance} between ${\cal X}_1$ and ${\cal X}_2$, denoted by $\dwa({\cal X}_1, {\cal X}_2)$ is defined to be 
\[
\dwa ({\cal X}_1, {\cal X}_2) = \min_{{\cal Z}}\mathbf{E}_{{\cal Z}} [\Vert {\cal Z}_1 -{\cal Z}_2 \Vert_1], 
\]
where $\mathcal{Z}= (\mathcal{Z}_1, \mathcal{Z}_2)$ is a coupling of $\mathcal{X}_1$ and $\mathcal{X}_2$. 
\end{definition}
The following fundamental result is due to Dudley~\cite{dudley1969speed}:

\begin{theorem}~\label{thm:Dudley}
Let ${\cal X}$ be any compactly supported measure on $\mathbb{R}^n$. Let $\bx_1, \ldots, \bx_M$ be $M$ random samples from ${\cal X}$ and let $\bX_M$ be the resulting empirical measure. Then 
\[
\mathbf{E}[\dwa ({\cal X}, \bX_M)]  = O(M^{-1/n}). 
\]
\end{theorem}

Let $\uniform$ denote the Haar measure (i.e., the uniform measure) on $\mathbb{S}^{n-1}$. Instantiating \Cref{thm:Dudley} with $\uniform$, we get the following corollary: 
\begin{corollary}~\label{corr:unif-approx}
For any error parameter $\zeta>0$, there exists $M_{n,\zeta}$ such that for any $M \ge 
M_{n,\zeta}$, there is a distribution $\umemp$ which satisfies the following:
\begin{enumerate}
\item $\dwa(\umemp, \uniform) \le \zeta$. 
\item The distribution $\umemp$ is uniform over its $M$-element support, which we denote by $\Suppm$.
\end{enumerate}
\end{corollary} 

We are now ready to construct the distribution ${\cal D}_{\actual}$. We  fix parameters $\zeta$, $p$ and $M$ as follows: 
\begin{equation}\label{eq:set-params}
\zeta \sqrt{\log(1/\zeta)}\coloneqq \frac{1}{\Lambda \cdot\sqrt{n}}, \quad \quad \quad M := \max \left\{ M_{n,\zeta}, \frac{\Lambda^2}{\zeta}\right\}, \quad \quad \quad p :=\frac{\Lambda}{M}. 
\end{equation}
%
%
%Instantiate ${\cal X}$ to be $U_{\mathbb{S}^{n-1}}$, the uniform measure on $\mathbb{S}^{n-1}$. Then, for any $\delta>0$, Theorem~\ref{thm:Dudley} implies that there exists $M_{n,\delta}$ such that for any $M \ge M_{n,\delta}$, there is a distribution $\umemp$ which satisfies: 
%\begin{enumerate}
%\item the distribution $\umemp$  is uniform over its support (of size $M$). We will use $\Supm$   to denote the support.
%
%\end{enumerate} 

\begin{definition}~\label{def:calD}
A draw of a function $\boldf \sim {\cal D}_{\actual}$ is sampled as follows: For each $z$ in $\Suppm$, define an independent Bernoulli random variable $\bW_z$ which is $1$ with probability $p$. The function $\boldf$ is
\[
\boldf(x) := \bigwedge_{z  \in \Suppm: \bW_z=1}  \slab_z(x). 
\]
Given such a $\boldf$, we define $\rel (\boldf) := \{z \in \Suppm: \bW_z=1\}$
\end{definition}
For intuition, $\rel(\boldf)$ can be viewed as the set of those elements of $\Suppm$ that are ``relevant'' to $\boldf$.
With the definition of ${\cal D}_{\actual}$ in hand,
we define   ${\cal D}_{\actual}(x)$ (analogous to ${\cal D}_{\ideal}(x)$) as follows: 
\[
{\cal D}_{\actual}(x) = \Prx_{\boldf \sim {\cal D}_{\actual}} [\boldf(x) =1]. 
\]

Similar to ${\cal D}_{\ideal}$, we now consider the Bayes optimal classifier $BO_{{\cal D}_{\actual}}(x)$, which corresponds to the output of the best zero-query learning algorithm for an unknown target function $\boldf \sim {\cal D}_{\actual}$. The expected error of $BO_{{\cal D}_{\actual}}$ is given by 
\[
\opt({\cal D}_{\actual}) := \Ex_{\bx \sim \Nn}  [\min\{{\cal{D}}_{\actual}(\bx), 1- {\cal{D}}_{\actual}(\bx)\}]. 
\]
The next lemma is the main result of this subsection and the rest of this subsection is devoted to its proof. It relates $\opt({\cal D}_{\actual})$ to $\opt({\cal D}_{\ideal})$ \red{and completes Step~2 of the outline given earlier:}

\begin{lemma}\label{lem:rel-D-ideal-D-actual}
For ${\cal D}_{\actual}$ and ${\cal D}_{\ideal}$ as defined above and parameters $\zeta$, $M$ and $p$ as set in \eqref{eq:set-params}, 
\[
|\opt({\cal D}_{\actual})  - \opt({\cal D}_{\ideal})| = O(n^{-1/2}).  
\]
\end{lemma}
The proof of \Cref{lem:rel-D-ideal-D-actual} requires several claims.

\begin{claim}\label{clm:vector-diff}
Let vectors $z, z' \in \mathbb{S}^{n-1}$ satisfy $\Vert z- z'\Vert_2\le 1/3$. Then  
\[
\Prx_{\bx \sim \Nn}[\slab_z(\bx) \not = \slab_{z'}(\bx)] \le 5 \Vert z-z'\Vert_2 \sqrt{\ln \bigg(\frac{1}{\Vert z-z'\Vert_2}\bigg)}. 
\]
\end{claim}

\begin{proof}
Define $\mathsf{Bd}_\kappa := \{ y \in \R: ||y| - d| \le \kappa\}$.\ignore{\rnote{Was ``$ \{ y \in \R: ||y| - |d|| \le \kappa\}$'' but I think there is no reason to have absolute value around $d$, right?}}
For any parameter $t>0$ and any $x \in \mathbb{R}^n$, observe that 
\begin{equation}~\label{eq:bounding-diff-vector}
\slab_z(x) \not = \slab_{z'}(x) \ \textrm{only if}  \ (| (z-z') \cdot x | \ge t \Vert z-z' \Vert_2) \ \textrm{and} \ ( z \cdot x  \in \mathsf{Bd}_{t \Vert z-z' \Vert_2}).
\end{equation}
 Let us write $\erfc(t)$ to denote $\Pr_{\bx \sim \Nn}[|\bx| \ge t]$. Recalling that $\erfc(t) \le (e^{-t^2} + e^{-2t^2})/2$ (e.g., see equation 10 in \cite{chiani2003new}), we have that 
\[
\Pro_{\bx \sim \Nn}[| (z-z') \cdot \bx | \ge t \Vert z-z' \Vert_2] \le \frac{e^{-t^2} + e^{-2t^2}}{2}. 
\]
Likewise, using the fact that the density of the standard normal is bounded by $1$ everywhere, we have that
\[
\Pro_{\bx \sim \Nn}[ z \cdot \bx \in \mathsf{Bd}_{t \Vert z-z' \Vert_2}] \le 4t \Vert z-z' \Vert_2. 
\]
Plugging the last two equations back into \eqref{eq:bounding-diff-vector}, we have that
\begin{align*}
\Pro_{\bx \sim \Nn} [\slab_z(\bx) \not = \slab_{z'}(\bx)] &\le \min_{t>0} \left\{ \frac{e^{-t^2} + e^{-2t^2}}{2} + 4t \Vert z-z' \Vert_2\right\}\\
& \le 5 \Vert z-z'\Vert_2 \sqrt{\ln \bigg(\frac{1}{\Vert z-z'\Vert_2}\bigg)}, 
\end{align*}
giving \Cref{clm:vector-diff}.
\end{proof}

The next (standard) claim relates the Poisson point process over a finite set $\mathcal{A}$ to the process of sampling each element independently (with a fixed probability) from $\mathcal{A}$. 
\begin{claim}~\label{clm:Poisson-binomial}
Let $\mathcal{A}$ be any set of size $M$ and let $\Lambda>0$. Consider the following two stochastic processes (a draw from the first process outputs a subset of $\mathcal{A}$ while a draw from the second process outputs a multiset of elements from $\mathcal{A}$): 
\begin{enumerate}
\item The process $\mind(\mathcal{A}, \Lambda)$ produces a subset  $\mathcal{B}_{b} \subseteq \mathcal{A}$ where each element $a \in \mathcal{A}$ is included independently with probability $p= \Lambda/|\mathcal{A}|$.
\item The process $\pois(\mathcal{A}, \Lambda)$ produces a multiset $\mathcal{B}_{p}$ of elements from  $\mathcal{A}$ where we first draw $\bL \sim \mathsf{Poi}(\Lambda)$ and then set $\mathcal{B}_p$ to be a multiset consisting of $\bL$ independent uniform random elements from  $\mathcal{A}$ (drawn with replacement). 
\end{enumerate}
Then the statistical distance $\Vert \mind(\mathcal{A}, \Lambda) -\pois(\mathcal{A}, \Lambda) \Vert_1$ is at most $2\Lambda^2/M$.  
\end{claim}

\begin{proof}
A draw of $\mathcal{B}_p$ from $\pois(\mathcal{A}, \Lambda)$ can equivalently be generated  as follows:  for each $a \in A$, sample $\bx_a \sim \mathsf{Poi}(p)$ independently at random and then include $\bx_a$ many copies of $a$ in $\mathcal{B}_p$.
For $0 \le q \le 1$, let $\ber(q)$ denote a Bernoulli random variable with expectation $q$. Recalling that $\Vert \mathsf{Poi}(q)  - \ber(q) \Vert_1 \le 2q^2$, applying this bound to every $a \in A$ and taking a union bound, we have that 
\[
\Vert \mind(\mathcal{A}, \Lambda) -\pois(\mathcal{A}, \Lambda) \Vert_1 \le \sum_{a \in A}2 p^2  = 2 \frac{\Lambda^2}{M^2} \cdot M = \frac{2 \Lambda^2}{M}. 
\]
\end{proof}

Finally, to prove Lemma~\ref{lem:rel-D-ideal-D-actual}, we will use an intermediate distribution of functions defined as follows:

\begin{definition}~\label{def:inter-D}
For the parameter $\Lambda$ defined earlier, we define the distribution ${\cal D}_{\inter}$ as follows: to sample a draw $\boldf \sim 
{\cal D}_{\inter}$,  we 
(i) first sample $\bL \sim \mathsf{Poi}(\Lambda)$, and (ii) then sample $\bz_1, \ldots, \bz_{\bL} \sim \umemp$. The function $\boldf$ is
\[
\boldf (x) := \bigwedge_{i=1}^{\bL} \slab_{\bz_i}(x).
\]
\end{definition}

As with ${\cal D}_{\actual}$ and ${\cal D}_{\ideal}$, we define ${\cal D}_{\inter}(x)$ and $\opt({\cal D}_{\inter})$ as
\[
{\cal D}_{\inter}(x) := \Pro_{\boldf \sim {\cal D}_{\inter}} [\boldf(x) =1],
\quad \quad \quad
\opt({\cal D}_{\inter}) := \Ex_{\bx \sim \Nn}[\min\{{\cal D}_{\inter}(\bx), 1-{\cal D}_{\inter}(\bx)\} ].
\]
\ignore{
%We then define $\opt({\cal D}_{\inter})$ as 
%\[\opt({\cal D}_{\inter}) = \Ex_{\bx \sim \Nn}[\min\{{\cal D}_{\inter}(\bx), 1-{\cal D}_{\inter}(\bx)\} ]. \]
}

Now we are ready for the proof of \Cref{lem:rel-D-ideal-D-actual}:

\begin{proofof}{\Cref{lem:rel-D-ideal-D-actual}}
We begin with the following easy claim which \red{shows that ${\cal D}_{\inter}(x)$ is very close to ${\cal D}_{\actual}(x)$ for every $x$:}\ignore{relates the accuracy of the Bayes optimal estimator for $\boldf \sim {\cal D}_{\inter}$ versus for $\boldf \sim {\cal D}_{\actual}$ at any point $x$.} 

\begin{claim}~\label{clm:D-Dinter}
For any $x \in \mathbb{R}^n$, 
\[
|{\cal D}_{\inter}(x) - {\cal D}_{\actual}(x)| \leq \frac{2\Lambda^2}{M}.
\]
\end{claim}

\begin{proof}
Observe that $\boldf_{\inter} \sim {\cal D}_{\inter}$ 
($\boldf_{\actual} \sim {\cal D}_{\actual}$, respectively) 
can be sampled as follows: Sample $(\bz_1, \ldots, \bz_{\bL}) \sim \pois (\Suppm, \Lambda)$ ($(\by_1, \ldots, \by_{\boldQ} \sim \mind (\Suppm, \Lambda)$, respectively), and set
\[
\boldf_{\inter} (x)= \bigwedge_{i=1}^{\bL} \slab_{\bz_i}(x) \quad \quad \quad
\text{and} \quad \quad \quad \boldf_{\actual} (x)= \bigwedge_{i=1}^{\boldQ} \slab_{\by_i}(x). 
\]
It follows from \Cref{clm:Poisson-binomial} that
$\Vert  \pois (\Suppm, \Lambda) - \mind (\Suppm, \Lambda) \Vert_1 \le 2\Lambda^2/M$ and consequently 
$\Vert {\cal D}_{\inter} - {\cal D}_{\actual} \Vert_1 \le 2\Lambda^2/M$. This implies that 
\[
|{\cal D}_{\inter}(x) - {\cal D}_{\actual}(x)| = \left|
\Pro_{\boldf_{\inter}  \sim {\cal D}_{\inter}}[ \boldf_{\actual}(x)=1] - \Pro_{\boldf_{\actual}  \sim {\cal D}_{\actual}}[ \boldf_{\actual}(x)=1]\right| \le \frac{2\Lambda^2}{M}.  
\]
%It follows from Claim~\ref{clm:Poisson-binomial} and the definition of ${\cal D}$ and ${\cal D}_{\inter}$ that $\Vert {\cal D} - {\cal D}_{\inter} \Vert_1 = O(\Lambda^2 / M)$. This immediately implies 
%\[
%\big|\Pr_{\boldf \sim {\cal D}_{\inter}} [\boldf(x)=1] -\Pr_{\boldf \sim {\cal D}} [\boldf(x)=1] \big|  = O\bigg( \frac{\Lambda^2}{M}\bigg). 
%\]
%This proves the claim. 
\end{proof}

Next we relate the average value of ${\cal D}_{\inter}$ (for $\bx \sim N(0,1)^n$) to the average value of ${\cal D}_{\ideal}$:
\begin{claim}~\label{clm:d-inter-d-ideal}
\[
\Ex_{\bx \sim \Nn}[|{\cal D}_{\inter}(\bx) -{\cal D}_{\ideal}(\bx)|] = O(\Lambda \zeta \sqrt{\log(1/\zeta)}). 
\]
\end{claim}

\begin{proof} 
Recall that by \Cref{corr:unif-approx} there exists a coupling $\bZ=(\bz_1, \bz_2)$ between $\umemp$ and $\uniform$ such that $\Ex [\Vert \bz_1 - \bz_2 \Vert_1] \le \zeta$.
We consider the following coupling between ${\cal D}_{\inter}$ and ${\cal D}_{\ideal}$: 
\begin{enumerate}
\item Sample $\bL \sim \mathsf{Poi}(\Lambda)$. 
\item Sample $\{(\bz_1^{(j)}, \bz_2^{(j)})\}_{1 \le j \le \bL}$ independently from $\bZ^{\bL}$. 
\item Define 
\[
\boldf_{\itt}(x) = \bigwedge_{j=1}^{\bL} \slab_{\bz_1^{(j)}}(x)
\quad \quad \quad \text{and} \quad \quad \quad \boldf_{\id}(x) = \bigwedge_{j=1}^{\bL} \slab_{\bz_2^{(j)}}(x).
\]
\end{enumerate}
Observe that $\boldf_{\itt}$ follows the distribution ${\cal D}_{\inter}$ and $\boldf_{\id}$ follows the distribution ${\cal D}_{\ideal}$. Thus, the process above indeed describes a coupling between ${\cal D}_{\inter}$ and ${\cal D}_{\ideal}$. We consequently have 
\begin{align}
|\opt({\cal D}_{\ideal})  - \opt({\cal D}_{\inter})|  &\le  \mathop{\mathbf{E}}_{\bx \sim \Nn} \left[\left| \Pro_{\boldf_{\id}}[\boldf_{\id}(\bx)=1] -\Pro_{\boldf_{\itt}}[\boldf_{\itt}(\bx)=1] \right|
\right] \nonumber\\
%&=& \mathop{\mathbf{E}}_{\bx \sim \Nn} \big[\big| \mathop{\mathbf{E}} \big[ \mathop{\wedge}_{i=1}^L \slab_{\bz_1^{(i)}}(\bx) - \mathop{\wedge}_{i=1}^L \slab_{\bz_2^{(i)}}(\bx) \big] \big| \big] \nonumber \\
&= \mathop{\mathbf{E}}_{\bx \sim \Nn} \left[\left| \mathop{\mathbf{E}}_{\bL \sim \mathsf{Poi}(\Lambda)} \mathop{\mathbf{E}}_{\bZ^{\bL}} \left[ \mathop{\wedge}_{i=1}^{\bL} \slab_{\bz_1^{(i)}}(\bx) - \mathop{\wedge}_{i=1}^{\bL} \slab_{\bz_2^{(i)}}(\bx) \right] \right| \right] \nonumber \\ 
&\le \mathop{\mathbf{E}}_{\bx \sim \Nn} \mathop{\mathbf{E}}_{\bL \sim \mathsf{Poi}(\Lambda)}\left[\left|  \mathop{\mathbf{E}}_{\bZ^{\bL}} \left[ \mathop{\wedge}_{i=1}^{\bL} \slab_{\bz_1^{(i)}}(\bx) - \mathop{\wedge}_{i=1}^{\bL} \slab_{\bz_2^{(i)}}(\bx) \right] \right| \right] \nonumber \\ 
&\le \mathop{\mathbf{E}}_{\bx \sim \Nn} \mathop{\mathbf{E}}_{\bL \sim \mathsf{Poi}(\Lambda)}\left[\left|  \mathop{\mathbf{E}}_{\bZ^{\bL}} \left[ \mathop{\sum}_{i=1}^{\bL} \slab_{\bz_1^{(i)}}(\bx) - \mathop{\sum}_{i=1}^{\bL} \slab_{\bz_2^{(i)}}(\bx) \right] \right| \right] \nonumber \\ 
&\le \mathop{\mathbf{E}}_{\bL \sim \mathsf{Poi}(\Lambda)} \mathop{\mathbf{E}}_{\bZ^{\bL}} \sum_{i=1}^{\bL} \bigg(\mathop{\mathbf{E}}_{\bx \sim \Nn} \left[\left| \slab_{\bz_1^{(i)}}(\bx)-\slab_{\bz_2^{(i)}}(\bx)\right|\right] \bigg). \label{eq:last-modify1}
\end{align}
Now, by \Cref{clm:vector-diff}, we have that
\[
\bigg(\mathop{\mathbf{E}}_{\bx \sim \Nn} \left[\left| \slab_{\bz_1^{(i)}}(\bx)-\slab_{\bz_2^{(i)}}(\bx)\right|\right] \bigg) \le 5 \Vert \bz_1^{(i)}- \bz_2^{(i)} \Vert_2 \sqrt{\log\bigg(\frac{1}{\Vert \bz_1^{(i)}- \bz_2^{(i)} \Vert_2}\bigg) }. 
\]
Plugging this back into \eqref{eq:last-modify1}, we have that
\begin{eqnarray}
|\opt({\cal D}_{\ideal})  - \opt({\cal D}_{\inter})| &\leq& \mathop{\mathbf{E}}_{\bL \sim \mathsf{Poi}(\Lambda)} \mathop{\mathbf{E}}_{\bZ^{\bL}} \sum_{i=1}^{\bL} \left[5 \Vert \bz_1^{(i)}- \bz_2^{(i)} \Vert_2 \sqrt{\log\bigg(\frac{1}{\Vert \bz_1^{(i)}- \bz_2^{(i)} \Vert_2}\bigg) } \right] \nonumber \\
 &\leq& \mathop{\mathbf{E}}_{\bL \sim \mathsf{Poi}(\Lambda)}  \sum_{i=1}^{\bL} \big [5 \cdot \zeta \sqrt{\log(1/\zeta)} \big] \nonumber \\
 &\leq& 5 \Lambda \zeta \sqrt{\log(1/\zeta)},
\end{eqnarray}
where the penultimate inequality used $\Ex [\Vert \bz_1 - \bz_2 \Vert_1] \le \zeta$ and the concavity of the function $x \sqrt{\log(1/x)}$.
\end{proof}
\Cref{lem:rel-D-ideal-D-actual} follows from \Cref{clm:D-Dinter} and \Cref{clm:d-inter-d-ideal}, recalling the values of the parameters set in \eqref{eq:set-params}.
\end{proofof}

\subsection{Analyzing query algorithms}  \label{sec:queries}
\Cref{lem:rel-D-ideal-D-actual} and \Cref{lem:BO-accuracy} together imply a bound on  the accuracy of the Bayes optimal classifier for ${\cal D}_{\actual}$ when the algorithm makes zero queries to the target function $\boldf \sim {\cal D}_{\actual}$.
To analyze the effect of queries, it will be useful to first consider an alternate combinatorial formulation of  ${\cal D}_{\actual}(x)$.  %For the purposes of this section, we set $\sac = S_{u,M}$.
For any point $x \in \mathbb{S}^n$, define $\sac (x) = \{ z \in  \sac: \slab_z(x)=0\}$.\ignore{ and $\# \sac(x) = |\sac (x)|$. }
By definition of ${\cal D}_{\actual}$, we have that for any $x \in \R^n$,
\begin{equation} \label{eq:dactual-useful}
\Pro_{\boldf \sim {\cal D}_{\actual}} [ \boldf(x) =1] = (1-p)^{|\sac (x)|}. 
\end{equation}

Restated in these terms, \Cref{lem:BO-accuracy} and \Cref{lem:rel-D-ideal-D-actual} give us that
%For any point $x$, we define $\# S(x) = |\{z: \slab_z(x)=0 \ \wedge \ z \in \Suppm \}|$. Then, note that 
%\[
%\Pro_{\boldf \sim {\cal D}} [ \boldf(x) =1] = (1-p)^{\# S(x)}. 
%\]
%Observe that by we have that 
\begin{equation}~\label{eq:relate-S(x)}
\Ex_{\bx \sim \Nn} \big[\big| (1-p)^{| S_{\actual}(\bx)|} - 1/2 \big|\big] = O\bigg( \frac{\log s \cdot \sqrt{\log n}}{\sqrt{n}}\bigg)
\end{equation}

We return to our overall goal of analyzing the Bayes optimal classifier when the learning algorithm makes at most $s$ queries to the unknown target $\boldf$. While the actual MQ oracle, when invoked on $x \in \R^n$, returns the binary value of $\boldf(x)$, for the purposes of our analysis we consider an augmented oracle which provides more information and is described below. 

\subsubsection{An augmented oracle, and analyzing learning algorithms that use this oracle}

Similar to \cite{BBL:98}, to keep the analysis as clean as possible it is helpful for us to consider an augmented version of the MQ oracle. 
(Note that this is in the context of ${\cal D}_{\actual}$, so the set $S_{\actual}$ is involved in what follows.) 
Fix an ordering of the elements in $S_{\actual}$, and let $f$ be a function in the support of ${\cal D}_{\actual}$. Recalling the definition of $\rel (f)$  from \Cref{def:calD}, we observe that for any point $x \in \mathbb{R}^n$, 
\[
f(x) =1 \  \textrm{  if and only if  } \ S_{\actual}(x) \cap \rel (f) = \emptyset. 
\]
This motivates the definition of our ``augmented oracle" for $f$. Namely, 
\begin{enumerate}
\item On input $x$, if $f(x)=1$ then the oracle returns $1$ (thereby indicating that $S_{\actual}(x) \cap \rel (f) = \emptyset$). 
\item On input $x$, if $f(x)=0$ then the oracle returns the first $z \in S_{\actual}$ (according to the above-described ordering on $S_{\actual}$) for which $z \in S_{\actual}(x) \cap \rel (f)$.\footnote{We note that the need to define this  ``first $z$" is the main reason that we do not work with ${\cal D}_{\ideal}$ directly and instead discretized it to obtain ${\cal D}_{\actual}$.} 
\end{enumerate}
It is clear that on any query string $x$, the augmented oracle for $f$ provides at least as much information as the standard oracle for $f$. Thus, it suffices to prove a query lower bound for learning algorithms which have access to this augmented oracle.  

%Associated with any $f$ in the support of ${\cal D}_{\actual}$, we have a vector $V_f \in \{0,1\}^{|S_{\actual}|}$ which is $1$ in the $i^{th}$ coordinate if the $i^{th}$ element in $S_{\actual}$ is present in $\rel (f)$ and $0$ otherwise. Note that $V_f$ is a complete description of $f$ and vice-versa. 

 At any point in the execution of the $s$-query learning algorithm, let $X$ represent the list of query-answer pairs that have been received thus far from this augmented oracle.  Let ${\cal D}_{\actual, X}$ denote the conditional distribution of $\boldf \sim {\cal D}_{\actual}$ conditioned on the query-answer list given by $X$. 
 As  in \cite{BBL:98}, the distribution ${\cal D}_{\actual, X}$ is quite clean and easy to describe. To do so, consider a vector $V_X$ whose entries are indexed by the elements of $S_{\actual}$. For $z \in S_{\actual}$, we define ${V}_{X}(z)$ as 
 \[
 {V}_{X}(z) := \Pro_{\boldf \sim {\cal D}_{\actual, X}} [z \in \rel (\boldf)]. 
 \]
Let us also define the Bernoulli random variables $\{{\bW}_{X,z}\}_{z \in S_{\actual}}$, where ${\bW}_{X,z}$ is $1$ if $z \in \rel(\boldf)$ for $\boldf \sim {\cal D}_{\actual, X}$. 

We begin by making the following observation:

\begin{claim}
When $X$ is the empty list (i.e. when zero queries have been made), each $V_X(z)$ is equal to $p$, and the Bernoulli random variables $\{{\bW}_{X,z}\}_{z \in S_{\actual}}$ are mutually independent. 
\end{claim}

Let us consider what happens when the ``current'' query-answer list $X$ is extended with a new query $x$. We can view the augmented oracle as operating as follows: it proceeds over each entry $z$ in $S_{\actual}(x)$ (according to the specified ordering), and: 
\begin{enumerate}
\item If $V_X(z) =0$, this means that the query-answer pairs already in $X$ imply that $z \not \in \rel(\boldf)$.  Then the augmented oracle proceeds to the next $z$. 
\item If $V_X(z)=1$, this means that the query-answer pairs already in $X$ imply that $z  \in \rel(\boldf)$. In this case, the oracle stops and returns $z$ (recall that this is a vector in $\R^n$, specifically an element of $S_{\actual}$) to the algorithm. Note that this $z$ is the first $z \in S_{\actual}$ (in order) such that $\slab_z(x)=0$. 
\item Finally, if $V_X(z)=p$, then the oracle fixes $\bW_{X,z}$ to $1$ with probability $p$ and to $0$ with probability $1-p$. (Recall that the random variable $\bW_{X,z}$ corresponds to the event that $z \in \rel(\boldf)$.) If $\bW_{X,z}$ is fixed to 0 then the oracle moves on to the next $z$, and if it is fixed to $1$ then the oracle stops and returns 
$z$. As in the previous case, this is then the first $z$ in $S_{\actual}$ such that $\slab_z(x)=0$. 
\end{enumerate}

Finally, we augment $X$ with the query $x$ and the above-defined response from the oracle. Based on the above description of the oracle, it is easy to see that the following holds:
\begin{claim}~\label{clm:conditional-independence}
For any $X$, each entry of $V_X(z)$ is either $0$, $1$ or $p$. Further, for any $X$, the random variables $\bW_{X,z}$ are mutually independent. Consequently, we can sample $\boldf \sim {\cal D}_{\actual, X}$ as 
\[
\boldf(x)  =  \mathop{\bigwedge}_{z \in S_{\actual}: \bW_{X,z}=1} \slab_z(x).  
\]
\end{claim}

 Next, we have the following two claims (which correspond respectively to Claim~1 and Claim~2 of \cite{BBL:98}:
 \begin{claim}~\label{clm:s-queries-bound1}
 If the learning algorithm makes $s$ queries, then the number of entries in $V_X(\cdot)$ which are set to $1$ is at most $s$. 
 \end{claim}
\Cref{clm:s-queries-bound1} is immediate from the above description of the oracle. The next claim is also fairly straightforward:
 \begin{claim}~\label{clm:s-queries-bound2}
 If the learning algorithm makes $s$ queries, then 
 with probability at least $1-e^{-\frac{s}{4}}$, the number of zero entries in $V_X$ is bounded by $2s/p$. 
 \end{claim}
 \begin{proof}
Given any $X$, on a new query $x$ the oracle iterates over all $z \in S_{\actual}(x)$ and sets $V_X(z)$ to $0$ with probability $1-p$ and $1$ with probability $p$, stopping this process as soon as (a) it sets the first $1$, or (b) it has finished iterating over all $z \in S_{\actual}(x)$, or (c) the current $V_X(z)$ was already set to $1$ in a previous round.
 
Thus, given any $X$, the number of new zeros added to $V_X$ on a new query $x$ is 
 stochastically dominated by $\mathsf{Geom}(p)$, the geometric random variable with parameter $p$. It follows that the (random variable corresponding to the) total number of zeros in $V_X$ is stochastically dominated by a sum of $s$ independent variables, each following $\mathsf{Geom}(p)$. 
We now recall the following standard tail bound for sums of geometric random variables \cite{janson2018tail}:

\begin{theorem}~\label{thm:tail-bound-geometric}
Let $\BR_1$, $\ldots$, $\BR_s$ be independent $\mathsf{Geom}(p)$ random variables.  For $\lambda \ge 1$, 
\[
\Pr\big[\BR_1 + \ldots + \BR_s \ge \frac{\lambda s}{p}\big]  \le e^{-s(\lambda -1 - \ln \lambda)}. 
 \]
\end{theorem}
Substituting $\lambda = 2$, we get that the number of zeros in $V_X$ is bounded by $2s/p$ with probability at least $1- e^{-s/4}$. This finishes the proof. 
 \end{proof}
 
\subsection{Proof of~\Cref{thm:our-BBL-lb}}
All the pieces are now in place for us to finish our proof of \Cref{thm:our-BBL-lb}.
 \red{The high-level idea is that thanks to \Cref{clm:s-queries-bound1} and \Cref{clm:s-queries-bound2}, the distribution ${\cal D}_{\actual,X}$ cannot be too different from ${\cal D}_{\actual}$ as far as the accuracy of the Bayes optimal classifier is concerned; this, together with \Cref{lem:rel-D-ideal-D-actual} and \Cref{lem:BO-accuracy}, gives the desired result.}
 
 Let $\mathcal{E}$ be the event (defined on the space of all possible outcomes of $X$, the list of at most $s$ query-answer pairs)
 that the number of zero entries \red{in $V_X$} is at most $2s/p$. Observe that $\Pr[\overline{\mathcal{E}}] \le e^{-s/4}$ by \Cref{clm:s-queries-bound2}. We now  bound the performance of the Bayes optimal estimator for ${\cal D}_{\actual, X}$ conditioned on the event 
 $\mathcal{E}$. 
 
  Let $\mathcal{A}_1 = \{z \in S_{\actual} : V_X(z) =1\}$ and  $\mathcal{A}_0 = \{z \in S_{\actual} : V_X(z) =0\}$. Using \Cref{clm:conditional-independence} and \Cref{clm:s-queries-bound1}, we have  the following observations: 
\begin{align}
& \bullet \ \ \ \textrm{If }x\in \mathbb{R}^n \ \textrm{is such that }S_{\actual}(x) \cap \mathcal{A}_1 \not =\emptyset, \text{~then~} \Pro_{\boldf \sim {\cal D}_{\actual, X}} [\boldf(x)=0] =1. \nonumber\\
& \bullet \ \ \ \textrm{If }x\in \mathbb{R}^n \ \textrm{is such that }S_{\actual}(x) \cap \mathcal{A}_1  =\emptyset, \text{~then} \Pro_{\boldf \sim {\cal D}_{\actual, X}} [\boldf(x)=1] =(1-p)^{|S_{\actual}(x) \setminus \mathcal{A}_0|}.  \nonumber \\
&\bullet 
\ \ \  \Pro_{\bx \sim \Nn}[S_{\actual}(\bx) \cap \mathcal{A}_1 \not = \emptyset] \le \sum_{z  \in \mathcal{A}_1}  \Pr[\slab_z(\bx)=0]  \leq \frac{|\mathcal{A}_1|}{s^{100}} \le \frac{1}{s^{99}}. \label{eq:negative-points}
\end{align}
 The last inequality uses \Cref{clm:s-queries-bound1} to bound the size of $|\mathcal{A}_1|$ and the definition of $\slab_z(\cdot)$. Next, for any $z \in \mathcal{A}_0$,  observe that 
 \[
 \Ex_{\bx \sim \Nn}[\Ind[z \in S_{\actual}(\bx)]]   = \Pro_{\bx \sim \Nn}[\slab_z(\bx)=1] = \frac{1}{s^{100}}.
 \]
 This immediately implies that 
 \[
 \Ex_{\bx \sim \Nn} \bigg[\sum_{z \in \mathcal{A}_0} \Ind [z \in S_{\actual}(\bx)] \bigg]  = \frac{|\mathcal{A}_0|}{s^{100}} \le \frac{2}{p \cdot s^{99}}. 
 \]
 By Markov's inequality, this implies that 
 \begin{equation}~\label{eq:Markov}
 \Pro_{\bx \sim \Nn} \bigg[ \sum_{z \in \mathcal{A}_0} \Ind [z \in S_{\actual}(\bx)] \ge \frac{2}{ps^{98}}\bigg] \le \frac{1}{s}. 
 \end{equation}
Let us say that $x \in \mathbb{R}^n$ is \emph{good} if $S_{\actual}(x) \cap \mathcal{A}_1 = \emptyset$ and 
 \[
 \sum_{z \in \mathcal{A}_0} \Ind [z \in S_{\actual}(x)] \red{\ge} \frac{2}{ps^{98}}. 
 \]
By \eqref{eq:Markov}, we have that $\Pr_{\bx \sim \Nn}[\bx \textrm{ is good}] \le 1/s$. We observe that for any good $x$, we have
 \[
 |S_{\actual}(x) | - \frac{2}{ps^{98}} \le |S_{\actual}(x) \setminus \mathcal{A}_0| \le |S_{\actual}(x) | .
 \]
It follows that 
 \[
 (1-p)^{|S_{\actual}(x)|}  \cdot (1-p)^{-\frac{2}{ps^{98}}} \ge (1-p)^{|S_{\actual}(x) \setminus \mathcal{A}_0|} \ge (1-p)^{|S_{\actual}(x)|} .
 \]
 Using the fact that $(1-p)^{-\frac{2}{ps^{98}}} \le 1 + \frac{\red{4}}{s^{98}}$, we have that
 \[
 (1-p)^{|S_{\actual}(x)|} \cdot \bigg( 1+ \frac{\red{4}}{s^{98}}\bigg)  \geq  (1-p)^{|S_{\actual}(x) \setminus \mathcal{A}_0|} \ge (1-p)^{|S_{\actual}(x)|}.
 \]
This implies that for any $x \in \R^n$ which is good, 
\begin{equation}~\label{eq:x-good}
\bigg| {\cal D}_{\actual, X}(x) - \frac{1}{2} \bigg| = \bigg| (1-p)^{|S_{\actual}(x) \setminus \mathcal{A}_0|} - \frac{1}{2} \bigg| \le \bigg| (1-p)^{|S_{\actual}(x) |} - \frac{1}{2} \bigg| + \frac{\red{4}}{s^{98}}. 
\end{equation}
Combining this with \eqref{eq:relate-S(x)}, \eqref{eq:Markov} and \eqref{eq:negative-points}, we get that 
\[
\Ex_{\bx \sim \Nn} \bigg[ \bigg| {\cal D}_{\actual, X}(x) - \frac{1}{2} \bigg|\bigg] \le \frac{1}{s} + \frac{\red{4}}{s^{98}} + \frac{1}{s^{99}} + O\bigg( \frac{\log s \cdot \sqrt{\log n}}{\sqrt{n}}\bigg). 
\]
This bounds the error of the Bayes optimal classifier for ${\cal D}_{\actual,X}$ conditioned on $\mathcal{E}$ to be at least $\red{{\frac 1 2}  -} O(\log s \cdot \sqrt{\log n}/\sqrt{n})$. Observing that $\Pr[\mathcal{E}] \ge 1-e^{-s/4}$ and $s\ge n$, the proof of \Cref{thm:our-BBL-lb} is complete. \qed

\section*{Acknowledgement}

We thank Ran Raz for alerting us to an error in an earlier proof of \Cref{claim:raz} and for telling us about the Borel-Kolmogorov paradox.

\bibliography{allrefs}{}
\bibliographystyle{alpha}

\appendix

%!TEX root = draft1b.tex

\section{Hermite analysis over $\R^n$} \label{sec:hermite}

We consider functions $f : \R^n \to \R$, where we think of the inputs $x$ to $f$ as being
distributed according to the standard $n$-dimensional Gaussian  distribution $N(0,1)^n$. In this context we view the space of
all real-valued square-integrable functions as an inner product space with inner product $\langle{f},{h} \rangle= \E_{\bx \sim N(0,1)^n}[f(\bx)h(\bx)]$.
In the case $n = 1$, there is a sequence of Hermite polynomials $h_0(x) \equiv 1, h_1(x) = x, h_2(x) = (x^2 -1)/\sqrt{2},\ldots$ that form a complete orthonormal basis for the space.  These polynomials can be defined via $\exp(\lambda x-\lambda^2/2)=\sum{d=0}^{\infty}(\lambda^d/\sqrt{d!}) h_d(x)$. In the case of general $n$, we have that the collection of $n$-variate
polynomials $\{H_S(x) := \prod_{i=1}^{n} h_{S_i}(x_i)\}_{S \in \N^n}$ forms a complete orthonormal basis for the space. 

Given a square integrable function $f : \R^n \to \R$ we define its Hermite coefficients by $\tilde{f}(S) = \langle{f},{H_S}\rangle,$ for $S\in \N^n$
and we have that $f(x) = \sum_{S}\tilde{f}(S)H_S(x)$ (with the equality holding in $\calL^2$).  Plancherel's and Parseval's identities are easily seen to hold in this setting, i.e.~for square-integrable functions $f,g$ we have $\E_{\bx \sim N(0,1)^n}[f(\bx)g(\bx)] = \sum_{S \in \N^n} \tilde{f}(S) \tilde{g}(S)$ and as a special case $\E_{\bx \sim N(0,1)^n}[f(\bx)^2] = \sum_{S \in \N^n} \tilde{f}(S)^2$.
Since $\tilde{f}(S)=\E_{\bg \sim N(0,1)^n}[f(\bg)]$, we observe that
$\sum_{|S| \geq 1} \tilde{f}(S)^2 = \Var[f(\bg)]$.

For $S \in \N^n$ we write $|S|$ to denote $S_1 + \cdots + S_n$.  For $j=0,1\dots$ we write $\mathsf{W}^{=j}[f]$ to denote the level-$j$ Hermite weight of $f$, i.e.~$\sum_{S\subset \N^n, |S| = j} \tilde{f}(S)^2.$  We similarly write $\mathsf{W}^{\leq j}[f]$ to denote $\sum_{S\subset \N^n, |S| \leq j} \tilde{f}(S)^2.$

\section{Fourier weight of monotone functions} \label{app:hermite-weight}

For completeness we give a proof of the following well known result in the analysis of Boolean functions:

\begin{claim}
Let $f: \bn \rightarrow \bits$ be a monotone function. Then the squared  Fourier weight at levels $0$ and $1$ is $\Omega(\frac{\log^2 n}{n})$, i.e. we have
\[
\sum_{S \subseteq n, |S| \leq 1} \hat{f}(S)^2 = \Omega\left(\frac{\log^2 n}{n}\right).
\]
This lower bound is best possible up to constant factors. 
\end{claim}
\begin{proof}
We first show that the  $\Omega(\frac{\log^2 n}{n})$ lower bound on the level $0$ and $1$ Fourier weight cannot be asymptotically improved by considering the so-called \textsf{TRIBES} function. This is a simple read-once monotone DNF (see Section~4.2 of \cite{ODBook}  for the exact definition). In particular, the $n$-variable function $f=\textsf{TRIBES}_n$ has the following properties: 

\begin{enumerate}
\item $f$ is monotone (and hence the influence $\Inf_i(f)$ of variable $i$ on $f$ equals the degree-1 Fourier coefficient $\hat{f}(i)$);
\item $\Pr_{\bx \in \bn} [f(\bx)=1] = \frac{1}{2} \pm O \big( \frac{\log n}{n}\big)$ (see Proposition~4.12 in \cite{ODBook});
\item For all $1 \le i \le n$, the influence of variable $i$ on $f$ is $O \big( \frac{\log n}{n}\big)$ (see Proposition~4.13 in \cite{ODBook}).
\end{enumerate}
Together, Items~1, 2 and~3 imply that the \textsf{TRIBES} function is indeed a tight example for our claim.

We now prove a lower bound on the squared Fourier weight of any monotone $f$.This is done via a  case analysis: 
\begin{enumerate}
\item If $|\widehat{f}(0)| \geq 0.01$, then the squared Fourier weight at level $0$ is at least $10^{-4}$. 
%\item We now apply Proposition~9.27 from \cite{ODBook} which states that if $f: \bn \rightarrow \bits$ is monotone, then there is a set $J$ of size $O(n/\log n)$ 
\item We can now assume $|\widehat{f}(0)| < 0.01$. This implies that $\Var[f] > 0.99$. Now suppose that $\Inf(f) = \sum_{i=1}^n \widehat{f}(i)$ is at least $\frac{\log n}{C}$ where $C$ is some absolute constant which is fixed in Step 3. Then we immediately get that the squared level-$1$ weight is at least 
 \[
\sum_{i=1}^n \widehat{f}(i)^2 \ge \frac{1}{n} \cdot \big( {\sum_{i=1}^n \widehat{f}(i)} \big)^2 = \frac{\Inf(f)^2}{n} = \Omega\bigg(\frac{\log^2 n}{n}\bigg). 
 \]
 \item The only remaining case is that $\Var[f] >0.99$ and $\Inf(f) < \frac{\log n}{C}$. Now, by choosing $C>0$ to be large enough, by Friedgut's junta theorem (see Section~9.6 of \cite{ODBook} it follows that there is a set $J \subseteq [n]$ such that (i) $|J| \le \sqrt{n}$ and (ii) the Fourier spectrum of 
 $f$ has total mass at most $0.01$ outside of the variables in $J$. From this, we get
 \[
 \sum_{i \in J} \widehat{f}(i) = \sum_{i \in J} \Inf_i(f) \ge  \sum_{S \subseteq J} \widehat{f}(S)^2 \geq 0.99. 
 \]
 Since $|J| \le \sqrt{n}$, we get that 
 $$
 \sum_{i} \widehat{f}(i)^2 \ge \sum_{i \in J} \widehat{f}(i)^2 \ge \frac{0.99}{\sqrt{n}}, 
 $$
and the proof is complete.
\end{enumerate}

%Let $f: \bn \rightarrow \bits$ be the \textsf{TRIBES} function on $n$-variables.   The important

%First of all, Proposition~4.2 in \cite{ODBook} states that 

\end{proof}

\section{Surface area and noise stability in the high noise rate regime}\label{app:stability}

In this section we give a concrete example showing that surface area is not a good proxy for noise stability at high noise rates. We do this by exhibiting two functions $\Psi_1$ and $\Psi_2$ such that (i) they have the same surface area  (up to a $\Theta(1)$ factor), but (ii) at noise rate $t= \Theta(\log n)$, the noise stability of $\Psi_1$ is exponentially smaller than that of $\Psi_2$. We now define the functions $\Psi_1$ and $\Psi_2$. 

\begin{definition}~\label{def:Psi}
Let functions $\Psi_1, \Psi_2: \mathbb{R}^n \rightarrow \{-1,1\}$ be defined as follows: 
\begin{enumerate}
\item Let $T := n^{1/4}$ and define $\Psi_1(x) := \prod_{i=1}^T \sign(x_i)$. 
\item $\Psi_2: \mathbb{R}^n \rightarrow \{-1,1\}$ is the function defined by Nazarov in \cite{Nazarov:03}; it is the indicator function of a convex body with surface area  $\Theta(n^{1/4})$.
% Further, the set $\Psi_2^{-1}(1)$ is centrally symmetric. \red{do we want to say more than this?}
\end{enumerate}
\end{definition} 
We observe that the boundary of $\Psi_1^{-1}(1)$ consists of $T$ hyperplanes that pass through the origin, and consequently we have that $\mathsf{surf}(\Psi_1) = T/\sqrt{2\pi},$ where the constant $1/\sqrt{2\pi}$ is the value of the pdf of a standard univariate $N(0,1)$ Gaussian at zero.\ignore{
%\begin{claim}~\label{surf:area-bound1}$\mathsf{surf}(\Psi_1) \le O(T)$. \end{claim}
%\begin{proof}
%Let $Z \subseteq \{-1,1\}^T$ be the set of all strings such that $\prod_{i=1}^T z_i = 1$. Note that $|Z| = 2^{t-1}$. Now, let $\mathcal{A}_z = \{x: \mathsf{sign}(x_i) = z_i\}$. We will show that for any $z$, 
%\begin{equation}~\label{eq:bound-area-quadrant}\int_{x \in \partial\mathcal{A}_z} \gamma_n(x) d\sigma(x) \le T \cdot 2^{-T+1}. \end{equation}
%This would immediately imply 
%\[\mathsf{surf}(\Psi_1) \le \sum_{z \in Z} \int_{x \in \partial\mathcal{A}_z} \gamma_n(x) d\sigma(x) \le 2^{T-1} \cdot T \cdot 2^{-T+1} \le T,  \]
%thus proving the claim (the last inequality uses \eqref{eq:bound-area-quadrant}). Thus, it remains to prove \eqref{eq:bound-area-quadrant}. To do this, consider any fixed $z$ and note that $\partial \mathcal{A}_z$ is the set 
%\[\partial \mathcal{A}_z = \cup_{i=1}^T (x_i =0) \mathop{\wedge}_{j\not =i} (\mathsf{sign}(x_j) = z_j)\]
%Thus, 
%\[\int_{x \in \partial\mathcal{A}_z} \gamma_n(x) d\sigma(x)  = \sum_{i=1}^T \int_{x : x_i =0 \mathop{\wedge}_{j\not =i} (\mathsf{sign}(x_j) = z_j)} \gamma_n(x) dx_1 \ldots dx_{i-1} dx_{i+1} \ldots dx_T. \]
%However, note that the integral is simply the mass of the Gaussian measure inside the orthant defined by $\mathop{\wedge}_{j\not =i} (\mathsf{sign}(x_j) = z_j)$. By symmetry, this is exactly $2^{-T+1}$ and thus we get \eqref{eq:bound-area-quadrant}. 
%\end{proof}
}
Thus, we have that $\mathsf{surf}(\Psi_1) = \Theta(\mathsf{surf}(\Psi_2))$.  However, the following claim shows that the noise stabilities of these two functions are very different from each other at large noise rates:
\begin{claim}~\label{clm:noise-stability-gap}
For $t \ge 0$, $\mathsf{Stab}_{t} (\Psi_1) = ({\frac 2 \pi} \cdot \arcsin(e^{-t}))^{n^{1/4}}$ and $\mathsf{Stab}_{t} (\Psi_2) = \Omega(n^{\red{-2}} \cdot e^{-2t})$. In particular, for $t = \Theta(1)$ we have that $\mathsf{Stab}_t(\Psi_1) = e^{-\Theta(n^{1/4})}$ and $\mathsf{Stab}_{t} (\Psi_2)  = \Omega(1/\red{n^2})$. 
\end{claim}
\begin{proof}
The lower bound on the noise stability of $\Psi_2$ is simply \Cref{corr:noise-stab}. The noise stability of $\Psi_1$ can be computed as 
\begin{eqnarray}
\mathsf{Stab}_{t} (\Psi_1) &=& \Ex_{\bg, \bg' \sim \Nn}[\Psi_1(\bg) \cdot \Psi_1(e^{-t} \bg + \sqrt{1-e^{-2t}} \bg')] \nonumber\\ 
&=& \prod_{i=1}^T \Ex_{\bg, \bg' \sim \Nn}[ \sign(\bg_i) \sign (e^{-t} \bg_i + \sqrt{1-e^{-2t}} \bg'_i)] \label{eq:stab-t-product}
\end{eqnarray}
The well known Sheppard's Formula  (see e.g.~Example~11.19 of \cite{ODBook}) states that
\[
\Ex_{\bg_i, \bg'_i \sim N(0,1)}[ {\sign}(\bg_i) {\sign} (e^{-t} \bg_i + \sqrt{1-e^{-2t}} \bg'_i)] = \frac{2}{\pi} \arcsin(e^{-t}). 
\]
Plugging this back into \eqref{eq:stab-t-product}, we get the claim. 
\end{proof}

%!TEX root = draft1b.tex

\section{Correlation of a fixed vector with a random unit vector}
In this section, we prove the following lemma. 
\begin{lemma}~\label{lem:corr}
Let $v \in \mathbb{S}^{n-1}$ be a fixed vector and $\bu \in \mathbb{S}^{n-1}$ be a uniformly drawn element of $\mathbb{S}^{n-1}$. For $0 < \eps < 1$ and $1/2 \ge \beta \ge \alpha >\frac{1}{\sqrt{n}}$ such that $\beta = (1+ \epsilon) \alpha$, we have\ignore{\rnote{This was
\[
1 \ge \frac{\Pr[|\langle v, \bu \rangle| \ge \alpha]}{\Pr[|\langle v, \bu \rangle| \ge  \beta]}  \ge 1 {-} O(n \alpha^2 \epsilon). 
\]
but that inequality goes the wrong way, doesn't it --- the numerator is larger than the denominator?  Things I changed are in \red{red} below.
}}
\[
1 \le \frac{\Pr[|\langle v, \bu \rangle| \ge \alpha]}{\Pr[|\langle v, \bu \rangle| \ge  \beta]}  \le 1 + O(n \alpha^2 \epsilon)
\]
provided that $n\cdot \alpha^2 \cdot \epsilon \le\frac{ 1}{8e^2}$.
\end{lemma}
\begin{proof}
%First of all, we will assume that 
It is well-known (see~\cite{Baum:90}) and easy to verify that 
\begin{equation} \label{eq:baum}
\Pr[\langle v, \bu \rangle \ge \alpha] = \frac{A_{n-2}}{A_{n-1}} \int_{z=\alpha}^{1} (1-z^2)^{\frac{n-3}{2}}dz. 
\end{equation}
Here $A_{n-1}$ is the surface area of the sphere $\mathbb{S}^{n-1}$. By symmetry, this implies that 
\begin{equation}
\frac{\Pr[|\langle v, \bu \rangle| \ge \alpha]}{\Pr[|\langle v, \bu \rangle| \ge  \beta]} = \frac{\int_{z=\alpha}^{1} (1-z^2)^{\frac{n-3}{2}}dz}{\int_{z=\beta}^{1} (1-z^2)^{\frac{n-3}{2}}dz. }
\end{equation}
Define $F(\alpha)$ as 
\[
F(\alpha) = (1-\alpha^2)^{\frac{n-3}{2}}. 
\]
Define $\Delta = \frac{1}{n \alpha}$.  Observe that $\Delta \le \alpha$ (for our choice of $\alpha$) and $\Delta \alpha = \frac{1}{n}$.  Using this, we have
\[
(1-(\alpha + \Delta)^2) \ge (1-\alpha^2) (1-4 \alpha \Delta). 
\]
This implies 
\begin{equation}~\label{eq:Falphadelta}
F(\alpha + \Delta) = (1-(\alpha + \Delta)^2)^{\frac{n-3}{2}} \ge (1-\alpha^2)^{\frac{n-3}{2}} \cdot (1-4\alpha \Delta)^{\frac{n-3}{2}} \geq F(\alpha) \cdot \frac{1}{e^2}. 
\end{equation}
Then, using \eqref{eq:Falphadelta}, 
\begin{eqnarray}~\label{eq:Falphadelta2}
\int_{z=\alpha}^{1} (1-z^2)^{\frac{n-3}{2}}dz \ge \int_{z=\alpha}^{z=\alpha + \Delta} (1-z^2)^{\frac{n-3}{2}}dz \ge \frac{\Delta}{e^2} F(\alpha). 
\end{eqnarray}
On the other hand, 
\begin{equation}~\label{eq:Falphadelta3}
\int_{z=\alpha}^{z=\beta} (1-z^2)^{\frac{n-3}{2}} dz \le (\beta-\alpha) F(\alpha) = \epsilon \cdot \alpha \cdot F(\alpha) 
\end{equation}
Note that the assumption  $n \alpha^2 \epsilon \le 1/(8e^2)$  translates to $\epsilon \alpha \le \frac{\Delta}{8e^2}$. 
Combining \eqref{eq:Falphadelta3}, \eqref{eq:Falphadelta2} and this observation, we get 
\ignore{\rnote{This was 
\[
1 \ge \frac{\Pr[|\langle v, \bu \rangle| \ge \alpha]}{\Pr[|\langle v, \bu \rangle| \ge  \beta]}  \ge  1 - O(n \alpha^2 \epsilon). 
\]
}
}
\[
1 \le \frac{\Pr[|\langle v, \bu \rangle| \ge \alpha]}{\Pr[|\langle v, \bu \rangle| \ge  \beta]}  \le 1 + {\frac {\eps \alpha}{\Delta/e^2 - \eps \alpha}} \le 1 + O(n \alpha^2 \epsilon).
\]
\ignore{\rnote{Right now I am confused:  suppose $\Delta=\alpha=1/\sqrt{n}$ and $\eps=0.9$, then it seems the denominator can be negative, which shouldn't be possible.}
}

\end{proof} 

\end{document}